\documentclass{article}
\usepackage{eurosym}
\usepackage{epsfig}
\usepackage{graphicx}
\usepackage{color}
\usepackage{anysize}
\usepackage{tikz}
\usepackage{cite}
\usepackage{amsmath,amsthm,verbatim,amssymb,amsfonts,amscd}
\usepackage{multirow}
\usepackage{subfigure}
\usepackage{float}
\usepackage{epsfig}
\usepackage{epstopdf}
\usepackage{amsxtra}
\usepackage{natbib}
\usepackage{pgfplots}
\usepackage{amsmath,amsthm,verbatim,amsbsy,amsfonts,amscd}
\usepackage{subfigure}
\usepackage{float}
\usepackage{enumerate}
\usepackage{amsxtra}
\usepackage[pdftex, colorlinks=true, linkcolor=blue, citecolor=blue, urlcolor=blue]{hyperref}
\usepackage{multirow}
\usepackage[toc,page]{appendix}

\fontencoding{OT1}\fontfamily{cmdh}\fontseries{bc}\fontshape{it}
\fontsize{14}{17pt}\selectfont
\theoremstyle{plain}
\newtheorem{theorem}{Theorem}
\newtheorem{corollary}{Corollary}
\newtheorem{lemma}{Lemma}
\newtheorem{proposition}{Proposition}

\theoremstyle{definition}

\newtheorem{claim}{Claim}

\newcommand{\RN}[1]{  \textup{\uppercase\expandafter{\romannumeral#1}}}

\providecommand{\keywords}[1]
{
	\small
	\textbf{\textit{Keywords---}} #1
}     

\begin{document}
	
	\title{Information Design in Optimal Auctions\thanks{
			We are very grateful to the editor and two anonymous referees for their
			valuable comments. We also thank Soumen Banerjee, Ben Brooks, Songzi Du,
			Simone Galperti, Gaoji Hu, Jingfeng Lu, Anne-Katrin Roesler, Satoru
			Takahashi, Kai Hao Yang, Wanchang Zhang, Xinhan Zhang, Weijie Zhong, and
			Junjie Zhou for their valuable comments. Financial support from ``the
			Singapore Ministry of Education Academic Research Fund Tier 1'' is
			gratefully acknowledged. All errors are our own.}}
	\author{Yi-Chun Chen\thanks{
			Department of Economics and Risk Management Institute, National University
			of Singapore. Email: ecsycc@nus.edu.sg} \and Xiangqian Yang\thanks{
			School of Economics and Trade, Hunan University, China. Email:
			yangxiangqian@hnu.edu.cn}}
	\date{October 15, 2022}
	\maketitle

	\abstract{
		We study the information design problem in a single-unit auction setting.
		The information designer controls independent private signals according to
		which the buyers infer their binary private values. Assuming that the seller
		adopts the optimal auction due to \cite{Myerson1981} in response, we
		characterize both the buyer-optimal information structure, which maximizes
		the buyers' surplus, and the seller-worst information structure, which
		minimizes the seller's revenue. We translate both information design
		problems into finite-dimensional, constrained optimization problems in which
		one can explicitly solve for the optimal information structure. In contrast
		to the case with one buyer \citep{roesler2017buyer}, we show that with two or more buyers, the symmetric buyer-optimal
		information structure is different from the symmetric seller-worst
		information structure. The good is always sold under the seller-worst
		information structure but not under the buyer-optimal information structure.
		Nevertheless, as the number of buyers goes to infinity, both symmetric
		information structures converge to no disclosure. We also show that in our ex ante symmetric setting, an asymmetric information structure is never seller-worst but can generate a strictly higher surplus for the buyers than the symmetric buyer-optimal information structure.}
	
		\keywords{information design; optimal auction; virtual value distribution;
		buyer-optimal information; seller-worst information.}
	
	\section{Introduction}
	
	Consider a seller who would like to sell one object to a group of buyers.
	The classical optimal auction due to \cite{Myerson1981} assumes that each
	buyer privately knows his own valuation; moreover, each valuation follows a
	distribution which is common knowledge. In this paper, we study an
	information design problem in which each buyer learns his private valuation
	independently via a signal according to which the seller runs the Myersonian
	optimal auction. The seller earns the expected highest nonnegative virtual
	value, whereas the buyers earn the expected total surplus minus the seller's
	revenue. We derive \emph{the buyer-optimal information structure} which
	maximizes the buyers' total surplus, as well as \emph{the seller-worst
		information structure} which minimizes the seller's optimal revenue.
	
	In reality the buyers may not know their own valuations for the good and
	have to assess how well the product suits their need via information sources
	such as advertisements, recommendations from some platform, or product
	descriptions. For instance, personalized advertising communicates privately
	with the buyers according to their individual characteristics such as
	gender, age, economic status, and so on. We consider contexts in which these
	personal characteristics are independently distributed so that the
	information is purely private,\footnote{%
		See, for instance, \cite{chen2022statistical} where the consumers'
		characteristics are modeled as i.i.d. random vectors.} that is, one learns
	nothing about a buyer's information or characteristics from the
	advertisement shown to another buyer. These features motivate our study of
	information design with independent private signals.
	
	Providing more information to the buyers can lead to a higher surplus but
	also a higher payment for them in an optimal auction. Hence, the effect of a
	new information source on the buyers' welfare is not a priori clear. The
	buyer-optimal information structure contributes to our understanding of this
	issue by identifying an information structure which maximizes the buyers'
	aggregate surplus. In this regard, our study builds upon the prior work on
	monopoly pricing by \cite{roesler2017buyer} but expands the scope to an
	auction setup with multiple buyers. The information designer may be a
	regulator who aims to promote consumers' welfare by requiring the seller to
	disclose certain information about the product, or else by restricting the
	seller from doing so.\footnote{\cite{terstiege2020buyer} study a
		buyer-optimal information design problem with monopoly pricing in which the
		information designer may be a regulator of the product information. For
		example, prescription drug ads must list side-effects and contra-indications
		to protect consumers. They focus on a situation where the buyer cannot
		commit to ignore any additional information released by the seller, whereas
		we follow \cite{roesler2017buyer} in setting aside the issue of the seller's
		disclosure.} The information designer may also be a data vendor who can sell
	product-related information for a fee proportional to buyers' (average/
	total) surplus and who therefore looks for an information structure which
	maximizes the buyers' surplus.
	
	The buyers' surplus is the total surplus minus the seller's revenue. Hence,
	studying the seller-worst information structure helps us understand the
	trade-off between minimizing the seller's revenue and maximizing the total
	surplus. The seller-worst information design also offers a \textquotedblleft
	minmax\textquotedblright\ upper bound for the\emph{\ revenue guarantee} of a
	mechanism regardless of the equilibrium and the information structure. Such
	a \textquotedblleft minmax\textquotedblright\ upper bound is crucial to
	establishing the strong duality results in \cite%
	{bergemann2016informationally}, \cite{du2018robust}, and \cite%
	{brooks2019optimal}. Specifically, these strong duality results show, in
	different contexts, that this \textquotedblleft minmax\textquotedblright\
	upper bound is equal to the maximal revenue guarantee achieved by a
	\textquotedblleft maxmin\textquotedblright\ mechanism. In this vein, the
	seller-worst information provides a first step toward establishing such a
	strong duality result---or lack thereof, in an independent private-value
	setting.\footnote{%
		We discuss the tightness of this upper bound in Section \ref%
		{simulationdiscussion}. In particular, \cite{Chen2022informationallyrobust}
		recently prove that when there are only two buyers, a second-price auction
		with a suitably chosen random reserve price guarantees exactly the
		seller-worst revenue that we identify in this paper over all symmetric
		independent information structures and undominated equilibria.}
	
	We assume that the seller has zero reservation value for the good and the
	buyers are ex ante symmetric. In particular, the buyers' ex post valuation
	of the good is equal to either $0$ or $1$ with an identical mean $p$. We
	assume that each signal provides an unbiased estimator about the buyer's
	valuation. Hence, by \cite{blackwell1953equivalent}, an information
	structure is feasible if and only if it consists of a profile of independent
	signal distributions, all with mean $p$.\footnote{%
		Alternatively, our analysis also applies and produces the same result if the
		information designer is allowed to choose any signal distributions with a
		given mean $p$ and support $\left[ 0,1\right] $. In a similar vein, \cite%
		{carrasco2018optimal} study a revenue-maximizing seller with a single buyer
		where the seller has only partial information about the buyer's valuation
		distribution, e.g., its first and second moments.} We begin by solving an
	optimal symmetric signal distribution in both information design problems.\footnote{%
		Focusing on the seller-optimal information, \cite{bergemann2007information}
		considers a Myersonian optimal auction setting where the seller can decide
		to whom to sell at what price and the accuracy by which bidders learn their
		(not necessarily binary) valuation through independent private signals. In
		our setting, the seller-optimal information is full revelation (i.e., the
		prior) against which the seller posts the price 1 and extracts full surplus.
		This is a special case of the full-surplus-extraction result due to \cite%
		{krahmer2020information}.} As is true for deriving symmetric equilibria in
	symmetric auctions, it is also more tractable to derive a symmetric
	information structure in our ex ante symmetric information design problems.\footnote{%
		Even with symmetric information structures, we still allow for irregular
		signal distributions for which the optimal auction need not be a
		second-price auction with a reserve price. Hence, our information design
		problem is not equivalent to the corresponding information design problem
		where the seller is committed to adopting a second-price auction with
		reserve; see Section \ref{rsn2} for more discussions and Appendix \ref%
		{irexample} for an illustrative example.}
	
	We show that as long as there are two or more buyers, the buyer-optimal
	information structure need not be equal to the seller-worst information
	structure. This result sharply contrasts with the results of \cite%
	{roesler2017buyer}, which show that the two information structures are
	equivalent when there is only one buyer. More precisely, when there are two
	or more buyers, we pin down a cutoff $p_{s}$ which is decreasing with the
	number of buyers. If $p$ is no more than $p_{s}$, then the (symmetric)
	seller-worst information structure for each buyer remains the same as in the
	one-buyer case. If $p$ is higher than $p_{s}$, then the seller-worst signal
	distribution remains equal to a truncated Pareto distribution but with
	virtual value $k_{s}>0$ for any signal less than 1, and with virtual value $%
	1 $ when the signal is equal to $1$. Since all virtual values are
	nonnegative at any signal profile, the good is always sold. Indeed, raising
	the low virtual value from $0$ to $k_{s}$ has two countervailing effects.
	First, the seller's revenue increases as the low virtual value increases.
	Second, to satisfy the mean constraint, increasing the low virtual value
	must be compensated for by decreasing the probability of having the high
	virtual value 1.  Through a cost-benefit analysis, we show that the  low virtual value $k_s$ increases  when  either the prior mean or the number of buyers grows.
	
	The (symmetric) buyer-optimal information structure differs from the
	(symmetric) seller-worst information structure in several ways. First, we
	pin down two cutoffs $r_{b}$ and $p_{b}$ which are also decreasing in the
	number of buyers. If $p$ lies between $r_{b}$ and $p_{b}$, then the
	buyer-optimal information structure remains the same as in the one-buyer
	case. When $p$ is less than $r_{b}$, the buyer-optimal signal distribution
	puts a positive mass on signal $0$ and the remaining mass on a truncated
	Pareto distribution with virtual values $0$ and $1$. Since signal $0$
	induces a negative virtual value, with positive probability the seller
	withholds the good. The reason for this is that, to maximize the buyers'
	surplus, the information designer needs to consider not only the seller's
	revenue but also the total surplus. Moreover, the total surplus is convex in
	the buyers' signal and hence favors dispersion. When $p$ is above $p_{b}$,
	the buyer-optimal signal distribution is a truncated Pareto distribution.
	However, also due to the convexity of the total surplus, the distribution
	induces a low virtual value $k_{b}<k_{s}$ for any signal less than 1 and a
	high virtual value $1$ otherwise.
	
	Our results highlight two kinds of allocative inefficiency with multiple
	buyers. First, in our setting the good will be allocated to a buyer with the
	highest interim (virtual) value, who need not have the highest ex post value 
	$1$. Second, as we argue above, the buyer-optimal signal distribution may
	put a positive mass on signal $0$ against which the seller withholds the
	good. Both sorts of inefficiency sharply contrast with the optimal
	information with one buyer under which ex post efficiency is always achieved.
	
	When the number of buyers goes to infinity, the cutoffs $p_{s}$, $r_{b}$,
	and $p_{b}$ all tend to zero; both $k_{b}$ and $k_{s}$ monotonically
	increase to $p$; and the corresponding probabilities assigned to virtual
	value $1$ monotonically decrease. As a result, both the buyer-optimal and
	the seller-worst signal distributions converge to a degenerate distribution
	which puts all mass on the prior mean $p$. In particular, learning \emph{no}
	information is asymptotically buyer-optimal with a large number of buyers.
	This result offers an extreme form of the message due to \cite%
	{roesler2017buyer} that a buyer does not want to learn his valuation
	perfectly in a monopoly setting. This result also sharply contrasts with the
	result of \cite{yang2019buyer} that when the buyers engage in strategic
	information acquisition, they will acquire (in a symmetric equilibrium)
	asymptotically perfect information about their values.
	
	We summarize the results on optimal symmetric information in Table \ref%
	{table1}.
	
	\begin{table}[ptb]
		\caption{Main results and comparisons}
		\label{table1}%
		\begin{tabular}{|c|lc|ll|l|c|}
			\hline
			\begin{tabular}{@{}c}
				Number \\ 
				of buyers%
			\end{tabular}
			& \multicolumn{5}{c|}{Optimal distribution with different prior mean $p$} & 
			\begin{tabular}{@{}c}
				Always \\ 
				sell?%
			\end{tabular}
			\\ \hline
			\multicolumn{1}{|c|}{\multirow{2}{*}{$n=1$}} & Seller-worst & 
			\multicolumn{4}{|l|}{\multirow{2}{*}{ \begin{tabular}{@{}l} Virtual values 0
						and 1 for any prior mean $p$, \\which is consistent with
						\cite{roesler2017buyer} \end{tabular} }} & \multicolumn{1}{|c|}{%
				\multirow{2}{*}{Yes}} \\ \cline{2-2}
			\multicolumn{1}{|l|}{} & Buyer-optimal & \multicolumn{4}{|l|}{} & 
			\multicolumn{1}{|l|}{} \\ \hline
			\multicolumn{1}{|c|}{\multirow{4}{*}{$n\geq2$}} & %
			\multirow{2}{*}{Seller-worst} & \multicolumn{2}{|c|}{$0<p\leq p_{s}$} & 
			\multicolumn{2}{c|}{$p_{s}<p<1$} & Yes \\ \cline{3-6}
			\multicolumn{1}{|l|}{} &  & \multicolumn{2}{|l|}{Virtual values 0 and 1} & 
			\multicolumn{2}{l|}{Virtual values $k_{s}$ and 1} &  \\ \cline{2-7}
			\multicolumn{1}{|l|}{} & \multirow{2}{*}{Buyer-optimal} & 
			\multicolumn{1}{|c|}{$0<p<r_{b}$} & \multicolumn{1}{r}{$r_{b}\leq p\leq
				p_{b} $} &  & \multicolumn{1}{|c|}{$p_{b}<p< 1$} & No \\ \cline{3-6}
			\multicolumn{1}{|l|}{} &  & \multicolumn{1}{|l|}{%
				\begin{tabular}{@{}l}
					A positive mass on $0$ as \\ 
					well as virtual values $0$ and $1$%
				\end{tabular}%
			} & \multicolumn{2}{l|}{%
				\begin{tabular}{@{}l}
					Virtual values \\ 
					0 and 1%
				\end{tabular}%
			} & 
			\begin{tabular}{@{}l}
				Virtual values \\ 
				$k_{b}$ and 1%
			\end{tabular}
			&  \\ \hline
			\multicolumn{1}{|c|}{\multirow{2}{*}{$n\rightarrow\infty$}} & Seller-worst & 
			\multicolumn{4}{|l|}{\multirow{2}{*}{Degenerate distribution (revealing
					nothing) for any mean}} & \multicolumn{1}{|c|}{\multirow{2}{*}{Yes}} \\ 
			\cline{2-2}
			\multicolumn{1}{|l|}{} & Buyer-optimal & \multicolumn{4}{|l|}{} & 
			\multicolumn{1}{|l|}{} \\ \hline
		\end{tabular}%
	\end{table}
	
	We also investigate asymmetric information structure in both information
	design problems. For the seller-worst problem, we show that the optimal
	symmetric information structure remains the unique optimal solution, even if
	the information designer can choose different signal distributions for
	different buyers. Intuitively, averaging a profile of asymmetric virtual
	value distributions grants the seller fewer option values in selecting the
	highest virtual values and thereby earns him less revenue. This means that
	restricting attention to symmetric signal distributions entails no loss in
	minimizing the seller's revenue. However, averaging a profile of asymmetric
	signal distributions may entail loss in the expected total surplus. In
	particular, we explore one case with two buyers and another case with a
	large number of buyers, in both of which an asymmetric information structure
	generates a strictly higher aggregate surplus for the buyers than the
	optimal symmetric information structure does. Our result shows that
	asymmetric information structure can emerge endogenously as a choice of a
	buyer-optimal information designer, even in our ex ante symmetric setting.
	
	Finally, we explain the novelty of our argument. For us to solve our
	information design problems, it is crucial that we transform the control
	variables. More precisely, instead of working with signal/interim value
	distributions, we work with the interim \emph{virtual} value distribution.
	After the change of variable, the information design problem becomes an
	isoperimetric problem in optimal control theory. The Euler-Lagrange equation
	can then be invoked in this problem to argue that the virtual value
	distribution function is a step function with at most two steps. This
	effectively reduces the infinite-dimensional information design problem to a
	tractable finite-dimensional constrained optimization problem. Moreover,
	with the few control variables, such as the two-step virtual value
	distribution functions, we are able to understand their trade-off, as we
	have explained above.
	
	The rest of this paper proceeds as follows. Section \ref{model} describes
	our model and formulates the information design problem. Section \ref{main
		results} presents our main results. Section \ref{Outline of the solution}
	demonstrates how we simplify the control variables of the information design
	problems. Section \ref{Asymmetric information structures} studies the
	information design problem with asymmetric signal distributions. Section \ref%
	{discussion} discusses issues with asymmetric or continuous priors and
	studies the tightness of our seller-worst upper bound for some candidate
	\textquotedblleft maxmin\textquotedblright\ mechanisms. Section \ref%
	{conclusion} concludes. Appendix \ref{proofs} contains all proofs which are
	omitted from the main text.
	
	\section{Model}
	
	\label{model}
	
	There is a seller who has one object to sell to a finite set $N=\left\{
	1,2,...,n\right\} $ of potential buyers. The seller has no value for the
	object. Each buyer's prior valuation, $v_{i}$, is identically and
	independently drawn from a Bernoulli distribution $H$ on $\{0,1\}$. Let $p= 
	\mathbb{E}[v_{i}]=\Pr \left( v_{i}=1\right) $ denote the mean of $H$. To
	rule out trivial cases, we assume that $p\in (0,1)$. Suppose that (i) each
	buyer can observe an independently and identically distributed signal $x_{i}$
	about $v_{i}$ from an information designer, and (ii) the joint distribution
	of $v_{i}$ and $x_{i}$ is common knowledge to the seller as well as among
	the buyers. That is, the information designer commits to a signal structure
	for each agent, and each agent observes that commitment, not just for his
	own signal structure but also for the other agents' signal structures.
	
	\subsection{Information structure}
	
	Following \cite{roesler2017buyer}, we say a signal distribution is \emph{\
		feasible} if each signal of a buyer provides him with an unbiased estimate
	about his valuation. Then, according to the characterization of \cite%
	{blackwell1953equivalent}, the prior valuation distribution $H$ is a
	mean-preserving spread of any feasible distribution of signals. Since $H$ is
	a Bernoulli distribution on $\{0,1\}$, the mean-preserving spread condition
	can be reduced to a mean constraint. Hence, a feasible symmetric information
	structure is a signal distribution $G$ with $G\in \mathcal{G}_{H}$, where 
	\begin{equation*}
		\mathcal{G}_{H}=\left\{ G:[0,1]\rightarrow \lbrack 0,1]\left\vert
		\int_{0}^{1}x\right.\mathrm{d}G(x)=p\text{ and }G\text{ is a CDF}\right\} .
	\end{equation*}
	
	\subsection{Information design problem}
	
	\label{virtual_def}
	
	Given a feasible signal distribution $G$, a revenue-maximizing mechanism is
	an optimal auction due to \cite{Myerson1981}. In an optimal auction, the
	seller's revenue is equal to the expected highest, nonnegative, ironed
	virtual value $\max_{i}\{\hat{\varphi}(x_{i}|G),0\}$. Formally, for any CDF $%
	G$ with $supp(G)\subset \lbrack 0,1]$, let $a=\inf \{x\in \lbrack
	0,1]|G(x)>0\}$, and define 
	\begin{equation*}
		\Psi (x|G)= 
		\begin{cases}
			0, & \mbox{if }x\in \lbrack 0,a); \\ 
			a-x(1-G(x)), & \mbox{if }x\in \lbrack a,1]\text{.}%
		\end{cases}%
	\end{equation*}
	Let $\Phi (x|G)$ be the convexification of $\Psi $ under the $G$-quantile
	space.\footnote{%
		That is, $\Phi (x|G)$ is the largest convex function of $G(x)$ that is
		everywhere weakly below $\Psi (x|G)$. In Appendix \ref{formal_virtual_value}
		, we also give a formal and detailed instruction about the ironed virtual
		value from \cite{monteiro2010optimal} and also \cite{yang2019buyer}.} By
	definition, for any $x\in \lbrack 0,1]$, the (\emph{ironed}) \emph{virtual
		valuation} at $x$, denoted as $\hat{\varphi}(x|G)$, is an infimum of the $G$
	-sub-gradients of $\Phi (x|G)$. If $\Phi (x|G)=\Psi (x|G)$ for any $x$, then
	we say that $G$ is a \emph{\ regular distribution}. We denote by $\hat{
		\varphi}$ an ironed virtual value and use $\varphi $ to denote a virtual
	value induced from a regular distribution. If $G$ is regular, then the
	virtual value has the well-known expression equal to 
	\begin{equation*}
		\varphi (x|G)=x-\frac{1-G\left( x\right) }{G^{\prime }\left( x\right) }\text{
			.}
	\end{equation*}
	We allow the information designer to choose any feasible distribution
	function $G$, whether it is regular or irregular and whether or not it
	admits a density function.
	
	Let $M(x|G)=\{i\in N|\hat{\varphi}(x_{i}|G)\geq \max_{j}\{\hat{\varphi}
	(x_{j}|G),0\}\}$ be the set of buyers who have the largest nonnegative
	virtual value for a given signal realization $x$; and let $M^{\prime
	}(x|G)=\{i\in N|x_{i}\geq {x_{j}},\forall j\in M(x|G)\}$ be the set of
	buyers who have not only the highest nonnegative virtual value but also the
	largest signal among those with the highest virtual value for a given signal
	realization $x$. Define an allocation rule as follows: 
	\begin{equation*}
		q_{i}(x_{i},x_{-i}|G)= 
		\begin{cases}
			\frac{1}{|M^{\prime }(x|G)|}, & \mbox{if }i\in M^{\prime }(x|G); \\ 
			0, & \mbox{if }i\not\in M^{\prime }(x|G).%
		\end{cases}%
	\end{equation*}
	That is, $q_{i}(x_{i},x_{-i}|G)$ is an optimal auction allocation rule which
	breaks a tie in favor of surplus maximization.
	
	We study the following information design problem parameterized by $\alpha
	=0 $ or $1$: 
	\begin{align}
		\max_{G}\int_{[0,1]^{n}}& \sum_{i=1}^{n}\left( \alpha x_{i}-\hat{\varphi}
		(x_{i}|G)\right) q_{i}(x_{i},x_{-i}|G)\prod_{i=1}^{n}\left( \mathrm{d}
		G(x_{i})\right)  \label{info} \\
		\text{s.t. }& \int_{0}^{1}\left( 1-G(x)\right) \mathrm{d}x=p.  \label{mean}
	\end{align}
	The term $\sum_{i=1}^{n}x_{i}q_{i}(x_{i},x_{-i}|G)$ is the total surplus
	generated under the optimal auction allocation rule $q_{i}$. Moreover, the
	term $\sum_{i=1}^{n}\hat{\varphi}(x_{i}|G)q_{i}(x_{i},x_{-i}|G)$ is the
	seller's revenue under the allocation rule $q_{i}$, namely, the expected
	highest nonnegative virtual value. Hence, if $\alpha =0$, the information
	designer aims to minimize the seller's revenue, and this corresponds to the
	seller-worst information design problem. If $\alpha =1$, the information
	designer aims to maximize the buyers' surplus, and this corresponds to the
	buyer-optimal information design problem. Hereafter, we call (\ref{mean})
	the \emph{mean constraint}.
	
	Endow the space of Borel probability measures on $\left[ 0,1\right] $ with
	the weak$^{\ast }$ topology. We say a signal distribution $G$ induces \emph{nonnegative virtual values except at 0} if the virtual values induced by $%
	G $ are nonnegative almost everywhere on $\left( 0,1\right] $. We denote by $%
	\mathcal{G}_{H}^{+}\subset \mathcal{G}_{H}$ the feasible signal distribution
	with nonnegative virtual values except at 0. Also we say a signal
	distribution $G$ is \emph{regular except at 0} if $G$ is regular almost
	everywhere on $\left( 0,1 \right] $. We will argue later in Lemma \ref%
	{almost_nonneg} and in Lemma \ref{regularity}, the optimal signal
	distribution must induce nonnegative virtual values except at 0 and be
	regular except at 0. We now present two lemmas. The first lemma documents
	the existence of the solution to the problem in \eqref{info}. The second
	lemma highlights the additional trade-off which a buyer-optimal information
	designer is facing, on top of minimizing the seller's revenue.
	
	\begin{lemma}
		\label{existence}For the problem in \eqref{info}, an optimal solution exists.
	\end{lemma}
	
	\begin{proof}
		We first establish the seller-worst case. By Theorem 2 of \cite%
		{monteiro2015note}, the expected revenue is a lower semicontinuous function 
		in $G$. Hence, the objective function of the problem in \eqref{info} is an 
		upper semicontinuous function in $G$. Since $\mathcal{G}_{H}$ is a closed 
		subset of the set of Borel probability measures on $[0,1]$, $\mathcal{G}_{H}$
		is compact. Thus, by the extreme value theorem, an optimal solution exists.
		
		For the buyer-optimal problem, proving the existence of an optimal solution 
		is more involved; we provide a formal proof in Appendix \ref{proofexist}. 
	\end{proof}
	
	\begin{lemma}
		\label{convex}For the problem in \eqref{info} and signal distribution $\hat{G} \in \mathcal{G}_{H}^{+}$, if $\hat{G}$ is a mean-preserving spread of
		signal distribution $G$, then $\hat{G}$ will generate more total surplus
		than $G$; and if $\hat{G}$ is a strict mean-preserving spread of signal
		distribution $G$, then $\hat{G}$ will strictly generate more total surplus
		than $G$.\footnote{%
			We say that $\hat{G}$ is a mean-preserving spread of signal distribution $G$
			if $\int_{0}^{x} (\hat{G}(t)-G(t))\mathrm{d}t\geq 0$ for all $x $ with
			equality at $x=1$; we also say that $\hat{G}$ is a strict mean-preserving
			spread if $\int_{0}^{x} (\hat{G}(t)-G(t))\mathrm{d}t\geq 0$ for all $x $
			with strict inequality at some $x$ with $G$-positive probability.}
	\end{lemma}
	
	\begin{proof}
		Observe that $\sum_{i=1}^{n}x_{i}q_{i}(x_{i},x_{-i}\vert G)=\max 
		\{x_{1},\cdots ,x_{n}\}$ if the good is allocated. Therefore,  
		\begin{equation*}
			\int_{\lbrack 0,1]^{n}}\sum_{i=1}^{n}x_{i}q_{i}(x_{i},x_{-i}\vert
			G)\prod_{i=1}^{n}\left(\mathrm{d}G(x_{i})\right) \leq \int_{0}^{1}x \mathrm{%
				d }G^{n}\leq \int_{0}^{1}x\mathrm{d}\hat{G}^{n}=\int_{[0,1]^{n}}
			\sum_{i=1}^{n}x_{i}q_{i}(x_{i},x_{-i}\vert G) \prod_{i=1}^{n}\left(\mathrm{d}
			\hat{G}(x_{i})\right) \text{.}
		\end{equation*}
		The first inequality follows because the good may not be allocated under $
		q_{i}$. The second inequality follows because $\hat{G}$ is a mean-preserving
		spread of $G$, and because $\int_{0}^{1}x \mathrm{d}G^{n}=\int_0^1\max 
		\{x_{1},\cdots ,x_{n}\}\mathrm{d}G$ has an integrand that is convex in $x.$ 
		Moreover, the second inequality is strict if $\hat{G}$ is a strict 
		mean-preserving spread of $G$. The equality follows because $\hat{G}\in  
		\mathcal{G}_{H}^{+}$, and the good is always sold under $q_{i}$ except when $
		x_{i}=0$ for every $i$. (In this case, the total surplus remains the same, 
		whether the good is sold or not.) 
	\end{proof}
	
	\section{Main results}
	
	\label{main results}
	
	In this section, we present our results on both information design problems.
	We will detail their proofs in Section \ref{Outline of the solution}.
	
	The unique symmetric seller-worst information structure is a truncated
	Pareto distribution. The distribution is regular and induces only two
	virtual values, $k_{s}$ and $1$, on the support $\left[ x_{s},1\right],$
	where $\left( k_{s},x_{s}\right) $ depends on the prior mean $p$ and the
	number of buyers $n$ (but we will omit the dependence to simplify the
	notations). The following result summarizes the seller-worst information
	structure:
	
	\begin{theorem}
		\label{result_seller}For each $p$ and $n$, there exists a tuple $\left(
		k_{s},x_{s}\right) \in \left[ 0,1\right] ^{2}$ such that the unique
		symmetric seller-worst information structure is the following truncated
		Pareto distribution: 
		\begin{equation}
			G_{s}(x)= 
			\begin{cases}
				1-\frac{x_{s}-k_{s}}{x-k_{s}}, & \text{ if }x\in \left[ x_{s},1\right); \\ 
				1, & \text{ if }x=1.%
			\end{cases}
			\label{s1}
		\end{equation}
		Moreover, there exists a threshold $p_{s}$ such that (i) $k_s=0$ if $p\in
		\left( 0,p_{s}\right]$; and (ii) $k_s>0$ if $p\in \left( p_{s},1\right).$
	\end{theorem}
	
	Theorem \ref{result_seller} states that the symmetric seller-worst signal
	distribution is a truncated Pareto distribution which induces virtual value $%
	x-\frac{1-G_{s}(x)}{g_{s}(x)}=x-\frac{x_{s}/(x-k_{s})}{x_{s}/(x-k_{s})^{2}}%
	=k_{s}$ for $x\in \left[ x_{s},1\right) $, and virtual value $1$ for $x=1$.
	Let $\theta _{s}\equiv 1-\frac{x_{s}-k_{s}}{1-k_{s}}$ be the probability of $%
	G_{s}$ inducing the low virtual value $k_{s}$. Since the seller-worst
	information induces only nonnegative values, the good is always sold.
	Moreover, since the seller-worst information is regular, the seller's
	optimal revenue can be achieved by a second-price auction with no reserve;
	see Proposition 5.2 of \cite{krishna2009auction}. The seller earns the
	expected virtual surplus equal to: 
	\begin{equation}
		k_{s}\times \theta _{s}^{n}+1\times \left( 1-\theta _{s}^{n}\right) \text{.}
		\label{obj-seller}
	\end{equation}
	
	Since $G_{s}$ is the seller-worst information structure, the tuple $\left(
	k_{s},x_{s}\right) $ or, equivalently, $(k_{s},\theta _{s})$ is chosen to
	minimizes the seller's expected virtual surplus, subject to the mean
	constraint 
	\begin{equation}
		(1-k_{s})(1-\theta _{s})(1-\log (1-\theta _{s}))+k_{s}=p\text{.}
		\label{mc_ks}
	\end{equation}%
	The Lagrangian is linear in the low virtual value with the marginal effect
	being determined by the following term:\footnote{%
		To be clear, the Lagrangian is defined to maximize the additive inverse of
		the objective in (\ref{obj-seller}).} 
	\begin{equation}
		J_{s}(\theta _{s})\equiv \underbrace{\theta _{s}\log (1-\theta _{s})}_{\text{
				cost}(<0)}+\underbrace{n(\theta _{s}+(1-\theta _{s})\log (1-\theta _{s}))}_{%
			\text{ benefit}(>0)}\text{.}  \label{foc_ks}
	\end{equation}%
	Intuitively, raising the low virtual value $k_{s}$ has two countervailing
	effects on a revenue-minimizing information designer. First, by increasing
	the lower virtual value, the seller's revenue increases, which translates
	into a cost in proportion to the first term in \eqref{foc_ks}. Second, to
	obey the mean constraint, the probability of the high virtual value 1 must
	be reduced, which results in a benefit in proportion to the second term in %
	\eqref{foc_ks}. Observe that the cost and benefit only depend on the number
	of buyers $n$ and probability $\theta _{s}$.
	
	As the Lagrangian is linear in the low virtual value, an interior solution
	occurs only when the benefit exactly offsets the cost, namely $J_{s}(\theta
	_{s})=0$, and then $k_{s}$ is pinned down by the mean constraint. With no
	more than two buyers, we can verify that the cost dominates the benefits
	regardless of $\theta _{s}$; hence, we must have a corner solution $k_{s}=0$
	and set $p_{s}=1$. Consequently, the seller-worst information structure is
	the same for $n=1$ and $n=2$. If there are more than three buyers, by
	setting $k_{s}=0$ in Equation \eqref{mc_ks}, we obtain the threshold, 
	\begin{equation}
		p_{s}=(1-\theta _{s})(1-\log (1-\theta _{s})),  \label{cutoffps}
	\end{equation}%
	where $\theta _{s}$ is obtained from $J_{s}(\theta _{s})=0$. Again, for any
	prior mean below the threshold, the cost dominates the benefit and thereby
	we have a corner solution. We provide details in Appendix \ref{fsw}.
	
	The cost-benefit analysis also tells us how the seller-worst information
	varies with the prior mean or the number of buyers. The cost-benefit formula 
	$J_{s}(\theta _{s})$ only depends on $n$, and so does the probability $%
	\theta _{s}$ of the low virtual value. Hence, for any fixed number of
	buyers, $k_{s}$ must go up in proportion with $p$ to obey the mean
	constraint \eqref{mc_ks}. If $n$ goes up, the benefit in $J_{s}(\theta _{s})$
	increases; hence, the probability of the low virtual value should also be
	increased to rebalance the cost and the benefit. Consequently, $k_{s}$
	should be raised to obey the mean constraint \eqref{mc_ks}. In summary, the
	low virtual value $k_{s}$ is increasing in both $p$ and $n$.
	
	We now turn to the buyer-optimal information. The unique buyer-optimal
	information structure is also a truncated Pareto distribution; however, it
	may place some mass $\theta _{0}$ at signal $x=0$. The distribution is
	regular except at $x=0$ and induces only two virtual values $k_{b}$ and $1$
	on the support $[x_{b},1]$, where $\left( \theta _{0},k_{b},x_{b}\right) $
	depends on $p$ and $n$. The following result summarizes the buyer-optimal
	information structure:
	
	\begin{theorem}
		\label{result_buyer}For each $p$ and $n$, there exists a tuple $\left(
		\theta _{0},k_{b},x_{b}\right)\in\left[0,1\right]^3 $ such that the unique
		symmetric buyer-optimal information structure is the following truncated
		Pareto distribution (except at $x=0$): 
		\begin{equation}
			G_{b}(x)= 
			\begin{cases}
				\theta _{0}, & \text{ if }x\in \lbrack 0,x_{b}); \\ 
				1-\frac{\left( x_{b}-k_{b}\right) (1-\theta _{0})}{x-k_{b}}, & \text{ if }
				x\in \lbrack x_{b},1); \\ 
				1, & \text{ if }x=1\text{.}%
			\end{cases}
			\label{b1}
		\end{equation}
		Moreover, there exist two thresholds $r_{b}$ and $p_{b}$ such that (i) $%
		k_{b}=0$ and $\theta _{0}>0$ if $p\in \left(0,r_{b}\right)$; (ii) $k_{b}=0$
		and $\theta _{0}=0$ if $p\in\left[r_{b},p_{b}\right]$; and (iii) $k_{b}>0$
		and $\theta _{0}=0$ if $p\in \left( p_{b},1\right)$.
	\end{theorem}
	
	The buyer-optimal signal distribution places mass $\theta _{0}$ at $x=0$ and
	then becomes a truncated Pareto distribution which induces virtual value $%
	k_{b}$ for $x\in \left[ x_{b},1\right) $, and virtual value $1$ for $x=1$.
	Under the buyer-optimal information, the seller's optimal revenue can be
	achieved by a second-price auction with a positive reserve price $x_{b}>0$;
	see Proposition 5.2 of \cite{krishna2009auction}.
	
	To maximize the buyers' surplus, the information designer must consider not
	only the seller's revenue but also the expected total surplus. As a result,
	the buyer-optimal information structure differs from the seller-worst
	information structure in several ways. First, when the prior mean $p<r_{b}$,
	the buyer-optimal signal distribution puts a mass $\theta _{0}$ on signal $0;$  that is, with probability $\theta _{0}^{n}$ the good is not sold. Second,
	when the prior mean $p>p_{b}$, the low virtual value becomes positive; this
	reflects the same intuition as that of the seller-worst information
	structure when $p>p_{s}$.
	
	Given $G_{b}$, we can compute the buyers' total surplus: 
	\begin{equation*}
		n\left( 1-k_{b}\right) \left( 1-\theta _{b}\right) \left( \sum_{i=1}^{n-1}%
		\frac{-\theta _{b}^{i}}{i}-\log \left( 1-\theta _{b}\right) +\left(
		\sum_{i=1}^{n-1}\frac{\theta _{0}^{i}}{i}+\log (1-\theta _{0})\right)
		\right) ,
	\end{equation*}%
	where $\theta _{b}\equiv 1-\frac{\left( 1-\theta _{0}\right) \left(
		x_{b}-k_{b}\right) }{1-k_{b}}$ is the mass $\theta _{0}$ plus the
	probability of $G_{b}$ inducing virtual value $k_{b}$. If $\theta _{0}=0$,
	then the trade-off between $\theta _{b}$ and $k_{b}$ can be analyzed
	similarly as that of $\theta _{s}$ and $k_{s}$ in the seller-worst
	information structure. In particular, an interior solution $(k_{b},\theta
	_{b})$ is jointly determined by the mean constraint 
	\begin{equation}
		(1-k_{b})(1-\theta _{b})\left( 1-\log (1-\theta _{b})\right) +k_{b}=p,
		\label{mc_kb}
	\end{equation}%
	and the cost-benefit equation 
	\begin{equation}
		\underbrace{\theta _{b}\log (1-\theta _{b})}_{\text{cost}(<0)}+\underbrace{%
			\theta _{b}^{n-1}\left( \theta _{b}+(1-\theta _{b})\log (1-\theta
			_{b})\right) +\sum_{i=1}^{n-1}\frac{\theta _{b}^{i+1}}{i}}_{\text{benefit}%
			(>0)}=0.  \label{foc_kb}
	\end{equation}%
	Again, since Lagrangian is linear in $k_{b}$, Equation \eqref{foc_kb}
	balances the cost and benefit. With no more than two buyers, we can again
	verify that the cost dominates the benefits regardless of $\theta _{b}$;
	hence, we must have a corner solution $k_{b}=0$ and set $p_{b}=1$.\footnote{%
		In contrast to the seller-worst case, however, the buyer-optimal information
		structure for $n=2$ may differ from that of $n=1$, since the former may put
		a mass at $x=0$.} If there are three or more buyers, by setting $k_{b}=0$ in
	Equation \eqref{mc_kb}, we obtain the following threshold: 
	\begin{equation}
		p_{b}=(1-\theta _{b})(1-\log (1-\theta _{b})),  \label{cutoffpb}
	\end{equation}%
	where $\theta _{b}$ solves Equation \eqref{foc_kb}. Then, for any prior mean
	below the threshold $p_{b}$, the cost dominates the benefit and thereby we
	have a corner solution for $k_{b}$.
	
	The cost-benefit equations in \eqref{foc_ks} and \eqref{foc_kb} also reveal
	that the buyer-optimal information designer is more reluctant to raise the
	low virtual value than a seller-worst information designer. Indeed,
	controlling the marginal cost, the buyer-optimal information designer
	receives less marginal benefit from raising the low virtual value than a
	seller-worst information designer does; see Claim \ref{bsgeqbb} in Appendix \ref{fbo} for a formal
	comparison. This is because raising the low virtual value, subject to the
	mean constraint, results in a mean-preserving contraction; hence, it
	decreases the expected total surplus by Lemma \ref{convex}. It also follows
	from the benefit comparison that $p_{b}$ is larger than $p_{s}$, and for any
	given $p$, we have $k_{b}\leq k_{s}$.
	
	As in the seller-worst information design, the cost-benefit Equation %
	\eqref{foc_kb} only depends on $n$, and so does the probability $\theta
	_{b}. $ Hence, for any fixed number of buyers, $k_{b}$ must go up in
	proportion with $p$ to obey the mean constraint \eqref{mc_kb}. If $n$ grows
	up, the benefit in \eqref{foc_kb} increases; hence, the probability of the
	low virtual value should also be increased to rebalance the cost and the
	benefit. Consequently, $k_{b}$ should be raised to obey the mean constraint %
	\eqref{mc_kb}. In summary, the low virtual value $k_{b}$ is increasing in
	both $p$ and $n$.
	
	Moreover, we identify another threshold $r_{b}<p_{b}$ below which the
	buyer-optimal information designer also deviates from the unique
	seller-worst information structure. In particular, for $p<r_{b}$, the
	buyer-optimal information designer puts a positive mass $\theta _{0}$ on $%
	x=0 $ to induce a mean-preserving spread from the seller-worst information
	structure. While the mean-preserving spread generates more revenue for the
	seller, it generates even more expected total surplus to benefit the buyers.
	We provide a similar cost-benefit analysis to pin down $\left( \theta
	_{0},\theta _{b}\right) $ as well as $r_{b}$ in Appendix \ref{fbo}.\footnote{%
		We also show in Appendix \ref{fbo} that $\theta _{0}$ and $k_{b}$ cannot
		both be positive and thereby each cost-benefit tradeoff involves only two of
		the three parameters.}
	
	As the number of buyer goes to infinity, both $\theta _{s}$ (which solves $%
	J_{s}(\theta _{s})=0$), and $\theta _{b}$ (which solves Equation %
	\eqref{foc_kb}) converge to 1; hence, both $p_{s}$ in \eqref{cutoffps} and $%
	p_{b}$ in \eqref{cutoffpb} converge to zero. It then follows from Equations %
	\eqref{mc_ks} and \eqref{mc_kb} that $k_{s}$ and $k_{b}$ will also increase
	to $p$. We summarize this in Corollary \ref{equivalence}.
	
	\begin{corollary}
		\label{equivalence}As $n\rightarrow \infty $, the buyer-optimal information
		structure coincides with the seller-worst information structure in the
		limit. Both are given by the degenerate distribution which assigns
		probability one to $p$.
	\end{corollary}
	
	Corollary \ref{equivalence} implies that when $n$ is large, both the
	buyer-optimal and the seller-worst information structures are close to
	\textquotedblleft no disclosure\textquotedblright; that is, the information
	designer chooses the degenerate distribution which concentrates on $x=p$. In
	the limit, the seller charges and extracts the ex ante expectation $p$ of a
	single buyer's value and leaves no surplus to the buyers. This is consistent
	with Part (ii) of Theorem 5 in \cite{ganuza2010signal} which establishes
	that in a second-price auction with no reserve and a sufficiently large
	number of buyers, a less precise signal produces a lower revenue for the
	seller.
	
	Corollary \ref{equivalence} also contrasts with the result of \cite%
	{yang2019buyer}. Specifically, \cite{yang2019buyer} shows that when buyers
	engage in strategic information acquisition, the unique symmetric
	equilibrium information structure converges to full information, as the
	number of buyers goes to infinity. Hence, the buyers also retain zero
	surplus in the limit, as in our buyer-optimal information structure. Since
	our symmetric buyer-optimal information structure provides an upper bound of
	the buyers' surplus under the symmetric equilibrium in \cite{yang2019buyer},
	the comparison reveals that their gap vanishes as the number of buyers goes
	to infinity.\footnote{\cite{shi2012optimal} studies an optimal auction
		problem in which the buyers acquire information individually from restricted
		feasible information structures \emph{after} the auction is announced; in
		contrast, in \cite{yang2019buyer} as well as in our paper, the optimal
		auction is designed according to the information, whether by design in our
		case or from the buyers' strategic acquisition in the case of \cite%
		{yang2019buyer}.}
	
	In \cite{yang2019buyer}, the buyers' surplus goes to zero because of the
	increasing competition in information acquisition with more buyers. In our
	case, however, the limiting zero surplus is driven by the buyer-optimal
	information designer's intentional choice to increase the low virtual value $%
	k_{b}$ in order to reduce the probability of virtual value $1$. As we
	explained after presenting Theorem \ref{result_seller}, when $n$ is large,
	reducing the probability of a high virtual value becomes a dominant
	consideration for the buyer-optimal information designer. Nevertheless, we
	will show later that the buyers do retain a nonvanishing surplus even when $%
	n $ goes to infinity; as long as the information designer can commit herself
	to choosing an asymmetric information structure; see Section \ref{asyboc}.
	
	Figure \ref{simulation_bosw} provides a numerical example with $p=1/2$ and 
	$n=1,2,...,10$. We document a number of features in this example to
	illustrate our main results.
	
	\begin{itemize}
		\item The top subfigure shows:
		
		\begin{itemize}
			\item As $n$ increases, the low virtual value goes up in both the
			seller-worst and buyer-optimal cases.
			
			\item The seller-worst low virtual value $k_{s}$ is larger than the
			buyer-optimal low virtual value $k_{b}$.
			
			\item $k_{s}$ tends to $0.5$ faster than $k_{b}$.
			
			\item When $n=2$, the buyer-optimal signal distribution puts a mass of 0.06
			on $x=0$ against which the seller withholds the good.
		\end{itemize}
		
		\item The bottom-left subfigure shows:
		
		\begin{itemize}
			\item The buyers' surplus under the buyer-optimal information structure is
			strictly higher than that of the seller-worst information structure.
			
			\item The buyers' surplus grows only from $n=1$ to $n=2$ and then decreases
			with $n$.
		\end{itemize}
		
		\item The bottom-right subfigure shows:
		
		\begin{itemize}
			\item The seller's revenue under the buyer-optimal information structure is
			strictly higher than that of the seller-worst information structure.
			
			\item The seller-worst revenue never exceeds $0.5$ since the information
			designer can opt to disclose no information. In contrast, with $n\geq 4$,
			the seller's revenue under the buyer-optimal information exceeds $0.5$. 
		\end{itemize}
	\end{itemize}
	
	\begin{figure}[tbp]
		\centering
		\begin{minipage}[c]{0.4\linewidth}
			
			\begin{tikzpicture}
				\begin{axis}[
					tick label style={font=\scriptsize},
					xlabel={The number of buyers},
					ylabel={The virtual value $k$},
					ytick={0, 0.35,0.5},
					yticklabels={0,0.35, 0.5},
					xtick={0.1,0.2,0.3,0.4,0.5,0.6,0.7,0.8,0.9,1},
					xticklabels={ 1,2,3,4,5,6,7,8,9,10},
					no markers,
					line width=0.3pt,
					cycle list={{red,solid}},
					samples=200,
					smooth,
					domain=0:1.3,
					xmin=0, xmax=1,
					ymin=-0.1, ymax=0.75,
					width=6cm, height=5.5cm,
					legend cell align=left,
					legend pos=  north west,
					legend style={draw=none,fill=none,name=legend},
					]
					\addplot[sharp plot,blue,thick] coordinates {
						(0.1,0)
						(0.2,0)
						(0.3,0)
						(0.4,0)
						(0.5,0.1008)
						(0.6,0.2088)
						(0.7,0.2687)
						(0.8,0.308)
						(0.9,0.3327)
						(1,0.3564)
					};
					\node [coordinate,pin=45:{mass 0.06 on 0}]
					at (axis cs:0.2,0) {};
					\addplot[sharp plot,red,thick] coordinates {
						(0.1,0)
						(0.2,0)
						(0.3,0.2116)
						(0.4,0.4265)
						(0.5,0.4729)
						(0.6,0.4895)
						(0.7,0.496)
						(0.8,0.4982)
						(0.9,0.499)
						(1,0.4996)
					};
					
					\addplot[dashed]coordinates {
						(0,0.5)
						(1,0.5)
					};
					
					\legend{Buyer-optimal, Seller-worst};
				\end{axis}
			\end{tikzpicture}
		\end{minipage}
		\par
		\begin{minipage}[l]{0.75\linewidth}
			
			\begin{tikzpicture}
				\begin{axis}[
					tick label style={font=\scriptsize},
					xlabel={The number of buyers},
					ylabel={The buyers' surplus},
					ytick={0, 0.33, 0.5},
					yticklabels={0, 0.33,0.5},
					xtick={0.1,0.2,0.3,0.4,0.5,0.6,0.7,0.8,0.9,1},
					xticklabels={ 1,2,3,4,5,6,7,8,9,10},
					no markers,
					line width=0.3pt,
					cycle list={{red,solid}},
					samples=200,
					smooth,
					domain=0:1.3,
					xmin=0, xmax=1,
					ymin=-0.1, ymax=0.55,
					width=6cm, height=5.5cm,
					legend cell align=left,
					legend pos=  north west,
					legend style={draw=none,fill=none,name=legend},
					]
					\addplot[sharp plot,blue,thick] coordinates {
						(0.1,0.31342)
						(0.2,0.328121)
						(0.3,0.2992)
						(0.4,0.265)
						(0.5,0.2325)
						(0.6,0.2132)
						(0.7,0.2008)
						(0.8,0.1922)
						(0.9,0.1859)
						(1,0.181)
					};
					
					\addplot[sharp plot,red,thick] coordinates {
						(0.1,0.31332)
						(0.2,0.32297)
						(0.3,0.241285)
						(0.4,0.117227)
						(0.5,0.0586)
						(0.6,0.0326)
						(0.7,0.0158)
						(0.8,0.009)
						(0.9,0.0052)
						(1,0.0032)
					};
					\addplot[dashed]coordinates {
						(0,0)
						(1,0)
					};
					\addplot[dashed]coordinates {
						(0,0.33)
						(1,0.33)
					};
					
					\legend{Buyer-optimal, Seller-worst};
				\end{axis}
			\end{tikzpicture}
			\begin{tikzpicture}
				\begin{axis}[
					tick label style={font=\scriptsize},
					xlabel={The number of buyers},
					ylabel={The seller's revenue},
					ytick={0.18, 0.5,0.65},
					yticklabels={0.18, 0.5,0.65},
					xtick={0.1,0.2,0.3,0.4,0.5,0.6,0.7,0.8,0.9,1},
					xticklabels={ 1,2,3,4,5,6,7,8,9,10},
					no markers,
					line width=0.3pt,
					cycle list={{red,solid}},
					samples=200,
					smooth,
					domain=0:1.3,
					xmin=0, xmax=1,
					ymin=0.1, ymax=0.7,
					width=6cm, height=5.5cm,
					legend cell align=left,
					legend pos=  north west,
					legend style={draw=none,fill=none,name=legend},
					]
					\addplot[sharp plot,blue,thick] coordinates {
						(0.1,0.18668)
						(0.2,0.3504)
						(0.3,0.462)
						(0.4,0.5624)
						(0.5,0.6126)
						(0.6,0.625)
						(0.7,0.6336)
						(0.8,0.6389)
						(0.9,0.6459)
						(1,0.6483)
					};
					
					\addplot[sharp plot,red,thick] coordinates {
						(0.1,0.18668)
						(0.2,0.3385)
						(0.3,0.4553)
						(0.4,0.4881)
						(0.5,0.4962)
						(0.6,0.4986)
						(0.7,0.4995)
						(0.8,0.4998)
						(0.9,0.4999)
						(1,0.49978)
					};
					\addplot[dashed]coordinates {
						(0,0.5)
						(1,0.5)
					};
					
					\legend{Buyer-optimal, Seller-worst};
				\end{axis}
			\end{tikzpicture}
		\end{minipage}
		\caption{A simulation for different $n$ with $p=0.5$}
		\label{simulation_bosw}
	\end{figure}
	
	\section{Outline of the solution}
	
	\label{Outline of the solution}
	
	Here we outline the steps to solve the information design problems we
	encounter. As we mentioned in the introduction, the key idea is to reduce
	the problems into tractable, finite-dimensional constrained optimization
	problems.
	
	\begin{enumerate}
		\item We first present two preliminary lemmas to restrict the class of
		distributions of interest:
		
		\begin{enumerate}
			\item Any buyer-optimal information structure induces nonnegative\ virtual
			values except at 0, i.e., the buyer-optimal signal distribution must belong
			to $\mathcal{G}_{H}^{+}$. Moreover, any seller-worst information structure
			must induce nonnegative virtual values almost everywhere. (Lemma \ref%
			{almost_nonneg}).
			
			\item Any buyer-optimal signal distribution must be \textquotedblleft
			regular except at 0\textquotedblright\ (i.e., regular everywhere except at $%
			x=0$) and any seller-worst signal distribution must be regular (Lemma \ref%
			{regularity} ).
		\end{enumerate}
		
		\item We change the choice variable in our information design problems from
		the distribution of signals to the distribution of virtual values (Lemma \ref%
		{change_variables}).
		
		\item We show that the reformulation after the change of variable reduces
		the information design problems to tractable, finite-dimensional
		optimization problems. This enables us to derive the explicit solution to
		the information design problems, regardless of the number of buyers. In
		particular, when $n=1 $, we obtain the same solution that appears in \cite%
		{roesler2017buyer}.
	\end{enumerate}
	
	\subsection{Preliminary Lemmas}
	
	In this subsection, we present the two preliminary results, Lemmas \ref%
	{almost_nonneg} and \ref{regularity}, to restrict the class of distributions
	of interest to the information designer. Both results rely on the following
	observation established by \cite{roesler2017buyer} and \cite{yang2019buyer}.
	
	Denote the left- and right-hand limit of a function $\xi\left(\cdot\right)$
	at a signal $x$ by $\xi\left(x^{-}\right)=\lim_{\delta\uparrow0}\xi
	\left(x-\delta\right)$ and $\xi\left(x^{+}\right)=\lim_{
		\delta\downarrow0}\xi\left(x+\delta\right)$, respectively.
	
	\begin{lemma}
		\label{pre}For any distribution $G$, any $x_{0}\in \lbrack 0,1]$, and any $%
		k\in \left[\hat{\varphi}(x_{0}^{-}|G), \hat{\varphi}(x_{0}^{+}|G)\right]$, 
		\begin{equation}
			(x-k)(1-G(x))\leq (x_{0}-k)(1-G(x_{0}^{-})),\quad \forall x\in \lbrack 0,1] 
			\text{.}  \label{lemma-pareto}
		\end{equation}
	\end{lemma}
	
	\begin{proof}
		See Lemma 9 in \cite{yang2019buyer}. 
	\end{proof}
	
	We can rearrange the inequality (\ref{lemma-pareto}) to obtain 
	\begin{equation*}
		G(x)\geq 1-\frac{(x_{0}-k)(1-G(x_{0}^{-}))}{x-k}.
	\end{equation*}
	
	In particular, the right-hand side is a Pareto distribution function which
	generates a constant virtual value $k$. Hence, Lemma \ref{pre} says that the
	Pareto distribution first-order stochastically dominates any other
	distribution $G$ with $k\in \lbrack \hat{\varphi}(x_{0}^{-}|G),\hat{\varphi}
	(x_{0}^{+}|G)]$.
	
	\begin{lemma}
		\label{almost_nonneg}Any optimal signal distribution $G$ which solves the
		information design problem in (\ref{info}) must induce nonnegative virtual
		values except at 0(i.e., $G\in \mathcal{G}_{H}^{+}$); moreover, if $G$ is a
		solution to (\ref{info}) for $\alpha =0$, it must induce nonnegative virtual
		values almost everywhere on $\left[ 0,1\right] $.
	\end{lemma}
	
	\begin{proof}
		See Appendix \ref{proofalmost}. 
	\end{proof}
	
	\begin{figure}[tbp]
		\centering
		\begin{minipage}[c]{0.45\textwidth}
			\centering
			\begin{tikzpicture}
				\begin{axis}[
					tick label style={font=\scriptsize},
					xlabel={$x$},
					ylabel={Distribution: $G(x)$},
					ytick={0, 0.180328, 1},
					yticklabels={0,$\theta_0$, 1},
					xtick={ 0.305,0.5, 1},
					xticklabels={ $x_1$,$x_2$,1},
					no markers,
					line width=0.3pt,
					cycle list={{red,solid}},
					samples=200,
					smooth,
					domain=0:1.3,
					xmin=0, xmax=1,
					ymin=0, ymax=1,
					width=6cm, height=5.5cm,
					legend cell align=left,
					legend pos=  north west,
					legend style={draw=none,fill=none,name=legend},
					]
					\addplot[blue,thick,domain=0:1]{x};
					\addplot[red,thick,domain=0.305:0.5]{1-(0.25/x)};
					\legend{Original , Modified};
					\addplot[red,thick,domain=0.5:1]{x-0.005};
					\addplot[red,dashed] coordinates {
						(0,0.180328)
						(0.305,0.180328)
					};
					
					\addplot[dashed] coordinates {
						(0.305,0)
						(0.305,0.180328)
					};
					\addplot[dashed] coordinates {
						(0.5,-1)
						(0.5,0)
					};
					\addplot[dashed] coordinates {
						(0.5,0)
						(0.5,0.5)
					};
				\end{axis}
			\end{tikzpicture}
		\end{minipage} 
		\begin{minipage}[c]{0.45\textwidth}
			\centering
			\begin{tikzpicture}
				\begin{axis}[
					tick label style={font=\scriptsize},
					xlabel={$x$},
					ylabel={Virtual value: $\hat{\varphi}(x)$},
					ytick={-0.5,0,   1},
					yticklabels={$\hat{\varphi}(0)$,0, 1},
					xtick={ 0.305,0.5 1},
					xticklabels={ $x_1$,$x_2$,1},
					no markers,
					line width=0.3pt,
					cycle list={{red,solid}},
					samples=200,
					smooth,
					domain=0:1.3,
					xmin=0, xmax=1,
					ymin=-1, ymax=1,
					width=6cm, height=5.5cm,
					legend cell align=left,
					legend pos=  north west,
					legend style={draw=none,fill=none,name=legend},
					]
					\addplot[blue,thick,domain=0:1]{2*x-0.999};
					\addplot[red,thick,domain=0.305:0.5]{0};
					\legend{Maybe negative,  nonnegative except at 0};
					\addplot[red,thick,domain=0.5:1]{2*x-1};
					\addplot[red,dashed] coordinates {
						(0.305,-0.5)
						(0.305,0)
					};
					\addplot[red,thick] coordinates {
						(0,-0.5)
						(0.305,-0.5)
					};
					\addplot[dashed] coordinates {
						(0.305,-1)
						(0.305,-0.5)
					};
					\addplot[dashed] coordinates {
						(0,0)
						(0.5,0)
					};
					\addplot[dashed] coordinates {
						(0.5,-1)
						(0.5,0)
					};
				\end{axis}
			\end{tikzpicture}
		\end{minipage}
		\caption{An improvement by a distribution with nonnegative virtual value
			except at 0}
		\label{nonneg_figure}
	\end{figure}
	
	Figure \ref{nonneg_figure} illustrates that how to improve the information
	designer's objective (for both the buyer-optimal case and the seller-worst
	case) . First, in the left subfigure, the blue curve is an arbitrary
	candidate distribution, which generates negative virtual values over the
	interval $[0,x_{2})$. The red curve coincides with the blue curve for $x\geq
	x_{2}$. For $x<x_{2}$, the red curve is a Pareto distribution which
	generates a virtual value $0$ over the interval $(x_{1},x_{2})$. By Lemma %
	\ref{pre}, we also put a positive mass $\theta _{0}$ on $x=0$ for the red
	curve to generate the same mean as the blue curve. By construction, the red
	curve is a strict mean-preserving spread of the blue curve; therefore, by
	Lemma \ref{convex}, the red curve can generate more expected total surplus.
	Second, by construction, for $x\in (0,x_{2})$, both of the ironed virtual
	values are no more than zero, and for $x\in \lbrack x_{2},1]$ the
	nonnegative ironed virtual value of the red curve coincides with that of the
	blue curve in the right subfigure. Therefore, the red curve generates
	strictly higher surplus for the buyers than the blue curve. In summary, the
	modification strictly benefits a buyer-optimal information designer and
	causes no loss to a seller-worst information designer. Hence, a
	buyer-optimal information structure must induce nonnegative virtual values
	except at 0. 
	\begin{figure}[tbp]
		\centering
		\begin{minipage}[c]{0.45\textwidth}
			\centering
			\begin{tikzpicture}
				\begin{axis}[
					tick label style={font=\scriptsize},
					xlabel={$x$},
					ylabel={Distribution: $G(x)$},
					ytick={0, 0.180328, 1},
					yticklabels={0, $\theta_0$, 1},
					xtick={0.22, 0.305,0.5,0.69, 1},
					xticklabels={$x_0$, $x_1$,$x_2$,$x_3$,1},
					no markers,
					line width=0.3pt,
					cycle list={{red,solid}},
					samples=200,
					smooth,
					domain=0:1.3,
					xmin=0, xmax=1,
					ymin=0, ymax=1,
					width=6cm, height=5.5cm,
					legend cell align=left,
					legend pos=  north west,
					legend style={draw=none,fill=none,name=legend},
					]
					\addplot[red,thick,domain=0.305:0.5]{1-(0.25/x)};
					\addplot[green,thick,domain=0.22:0.69]{1-(0.22/x)};
					\legend{Modified, Further modified};
					\addplot[red,thick,domain=0.5:1]{x};
					\addplot[red,dashed] coordinates {
						(0,0.180328)
						(0.305,0.180328)
					};
					\addplot[dashed] coordinates {
						(0.69,0)
						(0.69,0.6811)
					};
					\addplot[dashed] coordinates {
						(0.305,0)
						(0.305,0.180328)
					};
					
					\addplot[dashed] coordinates {
						(0.5,0)
						(0.5,0.5)
					};
				\end{axis}
			\end{tikzpicture}
		\end{minipage} 
		\begin{minipage}[c]{0.45\textwidth}
			\centering
			\begin{tikzpicture}
				\begin{axis}[
					tick label style={font=\scriptsize},
					xlabel={$x$},
					ylabel={Virtual value: $\hat{\varphi}(x)$},
					ytick={-0.5,0,   1},
					yticklabels={$\hat{\varphi}(0)$,0, 1},
					xtick={0.22, 0.305,0.5,0.69, 1},
					xticklabels={$x_0$, $x_1$,$x_2$,$x_3$,1},
					no markers,
					line width=0.3pt,
					cycle list={{red,solid}},
					samples=200,
					smooth,
					domain=0:1.3,
					xmin=0, xmax=1,
					ymin=-1, ymax=1,
					width=6cm, height=5.5cm,
					legend cell align=left,
					legend pos=  north west,
					legend style={draw=none,fill=none,name=legend},
					]
					\addplot[red,thick,domain=0.305:0.5]{0};
					\addplot[green,thick,domain=0.22:0.69]{0};
					\legend{Nonnegative except at 0,Nonnegative};
					\addplot[red,thick,domain=0.5:1]{2*x-1};
					\addplot[red,dashed] coordinates {
						(0.305,-0.5)
						(0.305,0)
					};
					\addplot[red,thick] coordinates {
						(0,-0.5)
						(0.305,-0.5)
					};
					\addplot[red,thick] coordinates {
						(0.305,0)
						(0.5001,0)
					};
					\addplot[green,thick] coordinates {
						(0.22,0)
						(0.305,0)
					};
					\addplot[green,dashed] coordinates {
						(0.69,0)
						(0.69,0.38)
					};
					\addplot[dashed] coordinates {
						(0.69,-1)
						(0.69,0)
					};
					\addplot[dashed] coordinates {
						(0.5,-1)
						(0.5,0)
					};
					\addplot[dashed] coordinates {
						(0.305,-1)
						(0.305,-0.5)
					};
					\addplot[dashed] coordinates {
						(0.22,-1)
						(0.22,0)
					};
				\end{axis}
			\end{tikzpicture}
		\end{minipage}
		\caption{A further improvement by a nonnegative virtual value distribution}
		\label{diff_figure}
	\end{figure}
	
	Even though the seller-worst information designer does not benefit from the
	modification in Figure \ref{nonneg_figure}, we use Figure \ref{diff_figure}
	to illustrate how we can further modify the red curve in Figure \ref%
	{nonneg_figure} to decrease the seller's revenue. Intuitively, since the
	seller gets no revenue as long as the virtual value is nonpositive, instead
	of inducing a negative virtual value at zero, a seller-worst information
	designer can do better by raising the negative virtual value to zero while
	decreasing the probability assigned to a positive virtual value to maintain
	the mean. More precisely, consider the green curve, which coincides with the
	red curve for $x\geq x_{3}$. For $x<x_{3}$, the green curve becomes a Pareto
	distribution which generates virtual value $0$ over the interval $%
	(x_{0},x_{3})$ such that the green curve generates the same mean as the red
	curve. By construction, and as shown in the right subfigure, the green curve
	generates strictly lower virtual values than the red curve over the interval 
	$(x_{2},x_{3})$; hence, the green curve generates strictly less revenue.
	
	\begin{lemma}
		\label{regularity}Any optimal signal distribution $G$ which solves the
		information design problem in (\ref{info}) must be regular except at 0;
		moreover, if $G$ is a solution to (\ref{info}) for $\alpha =0$, it must be
		regular.
	\end{lemma}
	
	\begin{proof}
		See Appendix \ref{proofregular}. 
	\end{proof}
	
	Figure \ref{regular_figure} illustrates how to improve the buyer-optimal
	information designer's objective with regular distributions except at 0.
	First, in the left subfigure, the blue curve is a candidate distribution,
	which is irregular over the interval $(x_{1},x_{2})$ and generates ironed
	virtual value $k\geq 0$ over the interval $[x_{0},x_{2}]$. The red curve
	coincides with the blue curve for $x\geq x_{2}$. For $x<x_{2}$, the red
	curve becomes a Pareto (and hence regular) distribution which generates a
	same virtual value $k$ over the interval $(x_{1},x_{2})$ and puts remaining
	mass $\theta _{0}$ on $0$. Observe that the red curve is a strict
	mean-preserving spread of the blue curve in the left subfigure.\footnote{%
		We draw the blue curve above the red curve on the ironed interval $%
		[x_{1},x_{2}]$ because the Pareto signal distribution with virtual value $k$
		first-order stochastically dominates the original signal distribution by
		Lemma \ref{pre}.} Hence, by Lemma \ref{convex}, the red curve can generate
	more expected total surplus. Second, the red curve induces weakly lower
	ironed virtual values than the blue curve; hence, the seller earns less
	under the red curve than under the blue curve. Overall, the buyer-optimal
	information designer's objective value is higher under the red curve than
	under the blue curve. Moreover, by Lemma \ref{almost_nonneg}, the red curve
	induces negative virtual value at $x=0$ and hence neither the red curve nor
	the blue curve can be seller-worst.
	
	\begin{figure}[tbp]
		\centering
		\begin{minipage}[c]{0.45\textwidth}
			\centering
			\begin{tikzpicture}
				\begin{axis}[
					tick label style={font=\scriptsize},
					xlabel={$x$},
					ylabel={Distribution: $G(x)$},
					ytick={0, 0.09366, 1},
					yticklabels={0, $\theta_0$, 1},
					xtick={0.1287, 0.25,0.5, 1},
					xticklabels={$x_0$, $x_1$,$x_2$, 1},
					no markers,
					line width=0.3pt,
					cycle list={{red,solid}},
					samples=200,
					smooth,
					domain=0:1.3,
					xmin=0, xmax=1,
					ymin=0, ymax=1,
					width=6cm, height=5.5cm,
					legend cell align=left,
					legend pos=  north west,
					legend style={draw=none,fill=none,name=legend},
					]
					\addplot[blue,thick,domain=0.1287:0.25]{1-(0.1287/(x))};
					\addplot[red,thick,domain=0.142:0.5]{1-(0.1287/(x))};
					\legend{Irregular, Regular except at 0};
					\addplot[blue,thick,domain=0.25:1/3]{2*x};
					\addplot[blue,thick,domain=1/3:1] {(1/2)*x+0.5};
					\addplot[red,thick,domain=0.5:1]{(1/2)*x+0.495};
					\addplot[red,dashed] coordinates {
						(0,0.09366)
						(0.142,0.09366)
					};
					\addplot[dashed] coordinates {
						(0.25,0)
						(0.25,0.5)
					};
					\addplot[dashed] coordinates {
						(0.5,0)
						(0.5,0.75)
					};
					\node [coordinate,pin=right:{$1-\frac{x_0-k}{x-k}$}]
					at (axis cs:0.35,0.632286) {};
				\end{axis}
			\end{tikzpicture}
		\end{minipage} 
		\begin{minipage}[c]{0.45\textwidth}
			\centering
			\begin{tikzpicture}
				\begin{axis}[
					tick label style={font=\scriptsize},
					xlabel={$x$},
					ylabel={Virtual value: $\hat{\varphi}(x)$},
					ytick={-0.4,  0,   1},
					yticklabels={$\hat{\varphi}(0)$,  $k$, 1},
					xtick={0.124,  0.25,0.5, 1},
					xticklabels={$x_0$, $x_1$, $x_2$, 1},
					no markers,
					line width=0.3pt,
					cycle list={{red,solid}},
					samples=200,
					smooth,
					domain=0:1.3,
					xmin=0, xmax=1,
					ymin=-0.5, ymax=1,
					width=6cm, height=5.5cm,
					legend cell align=left,
					legend pos=  north west,
					legend style={draw=none,fill=none,name=legend},
					]
					\addplot[blue,dashed,domain=0.25:1/3]{2*x-0.49};
					\addplot[red,thick,domain=0.142:0.5]{-0.01};
					\legend{Irregular, Regular  except at 0};
					\addplot[blue,thick,domain=0.124:0.25]{0};
					\addplot[blue,thick,domain=0.5:1] {2*x-0.99};
					\addplot[blue,dashed,domain=1/3:0.5] {2*x-0.99};
					\addplot[red,thick,domain=0.5:1]{2*x-1};
					\addplot[blue,thick,domain=0.25:0.5] coordinates {
						(0.25,0)
						(0.5,0)
					};
					\addplot[red,thick] coordinates {
						(0,-0.4)
						(0.142,-0.4)
					};
					\addplot[blue,dashed] coordinates {
						(1/3,2/3-1)
						(1/3,2/3-0.5)
					};
					\addplot[dashed] coordinates {
						(0.25,-0.5)
						(0.25,0)
					};
					\addplot[blue,dashed] coordinates {
						(0.124,-0.5)
						(0.124,0)
					};
					\addplot[dashed] coordinates {
						(0.5,-0.5)
						(0.5,0)
					};
					\addplot[red,dashed] coordinates {
						(0.142,-0.4)
						(0.142,0)
					};
					\node [coordinate,pin=90:{$x-\frac{1-G}{g}=k$}]
					at (axis cs:0.3,0.1) {};
					\node [coordinate,pin=310:{Ironing}]
					at (axis cs:0.4,0) {};
				\end{axis}
			\end{tikzpicture}
		\end{minipage}
		\caption{An improvement by a regular distribution except at 0}
		\label{regular_figure}
	\end{figure}
	
	By Lemma \ref{regularity}, we will hereafter use $\varphi $ instead of $\hat{
		\varphi}$ to denote the virtual value.
	
	\subsection{Change of variable}
	
	We now introduce a key step in solving the information design problems,
	namely, we change our control variable from a signal distribution to a
	virtual value distribution. Let $F(k)$ be the distribution of virtual values
	given a feasible signal distribution $G$. Since $G$ is regular except at 0, $%
	F\left( k\right) =Prob_{G}\left\{ x|\varphi (x)\leq k\right\} $. Then,
	except at $x=0 $, the virtual value of $G$ at signal $x$ is 
	\begin{equation}
		\varphi (x)=x-\frac{1-G(x)}{G^{\prime }(x)}.  \label{varphi}
	\end{equation}
	First, assume that the virtual valuation function $\varphi (\cdot )$ is 
	\emph{strictly increasing}, so that its inverse $x(k)=\varphi ^{-1}(k)$ is
	well defined. Consequently, $F(k)=G(x(k))$ and $F^{\prime }(k)=G^{\prime
	}(x(k))x^{\prime }(k)$. By Equation \eqref{varphi}, we have the following
	ordinary differential equation of $x(k)$: 
	\begin{equation}
		k=x\left( k\right) -\frac{1-F(k)}{F^{\prime }(k)}x^{\prime }\left( k\right) .
		\label{ODE}
	\end{equation}
	Solving the differential equation, we obtain Lemma \ref{changeofvariable}
	below.
	
	\begin{lemma}
		\label{changeofvariable}Suppose that $\varphi $ is strictly increasing in $x$
		. For each $k$ with $\varphi (x)=k$, $x(k)$ is the buyer's expected virtual
		value conditional on his virtual value being greater than or equal to $k$,
		i.e., 
		\begin{equation}
			x\left( k\right) =\mathbb{E}[\varphi |\varphi \geq k]=k+\frac{
				\int_{k}^{1}(1-F(s))\mathrm{d}s}{1-F(k)}\text{.}  \label{cv}
		\end{equation}
	\end{lemma}
	
	Alternatively, we can also derive the expression of $x(k)$ by applying the
	Envelope Theorem.\footnote{%
		We thank Kai Hao Yang for suggesting this elegant argument to us.} Observe
	that $x(k)$ is the solution to the following monopoly pricing problem with
	marginal cost $k$: 
	\begin{equation*}
		V(k)=\max_{x}(x-k)(1-G(x))\text{.}
	\end{equation*}
	By the Envelope Theorem, we can derive: 
	\begin{equation*}
		\frac{\partial V(k)}{\partial k}=-\left( 1-G\left( x(k)\right) \right) .
	\end{equation*}
	Hence, 
	\begin{align*}
		V(1)-V(k)& =-(x\left( k\right) -k)(1-G(x\left( k\right)
		))=\int_{k}^{1}-\left( 1-G\left( x(s)\right) \right) \mathrm{d}s \\
		\Longrightarrow x(k)& =k+\frac{\int_{k}^{1}\left( 1-G\left( x(s)\right)
			\right) \mathrm{d}s}{1-G(x(k))}=k+\frac{\int_{k}^{1}\left( 1-F\left(
			s\right) \right) \mathrm{d}s}{1-F(k)}\text{.}
	\end{align*}
	
	Moreover, we can still make use of the expression for $x\left( k\right) $ in %
	\eqref{cv}, even when $\varphi (x)$ is only weakly increasing (by
	regularity). To see this, note that any weakly increasing function can be
	uniformly approximated by a strictly increasing function. Let $\{\varphi
	_{m}\}_{m=1}^{\infty }$ be a sequence of strictly increasing functions
	converging uniformly to $\varphi $. For each $m$, let $G_{m}$ and $F_{m}$ be
	sequences of signal distributions and virtual value distributions
	corresponding to $\varphi _{m}$. Specifically, by solving Equation %
	\eqref{varphi}, we obtain $G_{m}(x)=1-\exp \left\{ \int_{0}^{x}(\varphi
	_{m}(t)-t)^{-1}\mathrm{d}t\right\} $ and $F_{m}$ is the virtual value
	distribution induced by $G_{m}$. Since $\left\{ \varphi _{m}\right\} $
	converges uniformly to $\varphi $, we also have $\left\{ G_{m}\right\} $ and 
	$\left\{ F_{m}\right\} $ uniformly converge to $G$ and $F$, respectively. In
	Appendix \ref{proofchange2}, we use this convergence result and the
	expression in \eqref{cv} to establish the following equation: 
	\begin{align}
		& \int_{0}^{1}x\mathrm{d}G^{n}(x)=1-\int_{0}^{1}G^{n}(x)\mathrm{d}x  \notag
		\\
		=& \int_{0}^{1}n(1-F(k))\left( \sum_{i=1}^{n-1}\frac{-F^{i}(k)}{i}-\log
		(1-F(k))+ \sum_{i=1}^{n-1}\frac{F^{i}(0^{-})}{i}+\log (1-F(0^{-})) \right)
		-F^{n}(k)\mathrm{d}k+1,  \label{m2}
	\end{align}
	where $G\left( 0\right) =F\left( 0^{-}\right) $. Indeed, we show in Appendix %
	\ref{proofchange2} that Equation \eqref{m2} follows for each $G_{m}$ and $%
	F_{m}$ and we then apply the bounded convergence theorem. To see the
	intuition, suppose that $n=1$ and there is no mass on signal 0. Then,
	Equation \eqref{m2} is reduced to: 
	\begin{equation}
		\int_{0}^{1}(1-F(k))(1-\log (1-F(k)))\mathrm{d}k=p.  \label{mc}
	\end{equation}%
	We may regard \eqref{mc} as a mean constraint for a general virtual value
	distribution that generalizes the special case \eqref{mc_ks} for a binary
	virtual value distribution on $\left\{ k_{s},1\right\}.$
	
	Equation \eqref{m2} is the total surplus given the distribution $G$, and
	when $n=1$, Equation \eqref{m2} is reduced to the expected mean. Hence, we
	have the following lemma.
	
	\begin{lemma}
		\label{change_variables}After the change of variable, the information
		designer's problem in \eqref{info} can be written as follows: 
		\begin{align}
			\max_{F}& \int_{0}^{1}\alpha n(1-F(k))\left( \sum_{i=1}^{n-1}\frac{-F^{i}(k) 
			}{i}-\log (1-F(k))+\left( \sum_{i=1}^{n-1}\frac{F^{i}(0^{-})}{i}+\log
			(1-F(0^{-}))\right) \right) \,\mathrm{d}k  \label{info_object} \\
			& +(1-\alpha )\int_{0}^{1}F^{n}(k)\mathrm{d}k+(\alpha -1)  \notag \\
			\text{s.t. }& \int_{0}^{1}(1-F(k))(1-\log (1-F(k))+\log (1-F(0^{-})))\mathrm{%
				d}k=p.  \notag
		\end{align}
	\end{lemma}
	
	\begin{proof}
		It directly follows from Equation \eqref{m2}. 
	\end{proof}
	
	\subsection{The case with $n=1$: \protect\cite{roesler2017buyer} revisited}
	
	\label{rs17}
	
	We are now ready to solve the information design problem for the case with $%
	n=1$, which is analyzed in \cite{roesler2017buyer}. For $n=1$, we can
	rewrite the information designer's problem as 
	\begin{align*}
		& \max_{F}\quad \alpha \underbrace{\int_{0}^{1}x\mathrm{d}G(x)}_{\text{the
				total surplus}}-\underbrace{\int_{0}^{1}k\mathrm{d}F(k)}_{\text{the seller's
				revenue}} \\
		\text{s.t. }& \int_{0}^{1}(1-F(k))(1-\log (1-F(k))+\log (1-F(0^-)))\mathrm{d}
		k=p.
	\end{align*}
	When $n=1$, the total surplus $\int_{0}^{1}x\mathrm{d}G$ is linear in $G$;
	moreover, $\int_{0}^{1}x\mathrm{d}G=p$ by the mean constraint. Therefore,
	the value of $\alpha $ has no effect on the optimization. This implies that
	the buyer-optimal information structure is equivalent to the seller-worst
	information structure. Since the virtual value is always nonnegative for the
	seller-worst case, we have $G(0)=F(0^-)=0$, i.e., there is no mass on $x=0$.
	
	The information design problem is an isoperimetric problem in optimal
	control theory; see Theorem 4.2.1 of \cite{van2004isoperimetric}. To solve
	the optimal control problem, we can write the Lagrangian as 
	\begin{equation*}
		\mathcal{L}(F,\lambda )=\alpha p+\int_{0}^{1}\left( F-\lambda ((1-F)(1-\log
		(1-F))-p)\right) \mathrm{d}k-1.
	\end{equation*}
	Let $\theta =F(k)$. Then, for each $k$, the Euler-Lagrange condition implies
	that 
	\begin{equation*}
		\partial \mathcal{L}/\partial \theta =1-\lambda \log (1-\theta )=0.
	\end{equation*}
	
	Since the solution of the Euler-Lagrange equation cannot be either $\theta
	=0 $ or $\theta =1$, there exists a constant $\lambda =1/\log (1-\theta )<0$
	. Since $\log (1-\theta )$ is monotone decreasing, there is at most a
	solution $\theta \in (0,1)$, such that the Euler-Lagrange equation holds.
	Therefore, $F(k)$ has only three values $0$, $\theta$, and $1$; hence, $F$
	is a two-point distribution. Moreover, since the constant $\lambda$ is
	fixed, and since $\log (1-\theta )|_{\theta \uparrow 1}\rightarrow -\infty $
	, we have 
	\begin{equation*}
		\left. \frac{\partial \mathcal{L}}{\partial \theta }\right\vert _{\theta
			\uparrow 1}=\left. 1-\lambda \log (1-\theta )\right\vert _{\theta \uparrow
			1}\rightarrow -\infty .
	\end{equation*}
	Therefore, it will never be optimal for $F(k)$ to jump to $1$ before $k=1$;
	therefore, the larger value with respect to $F$ must be $1$. That is, the
	support of $F$ is $\{k,1\}$.
	
	Now, the information designer only needs to choose $\{k,F(k)=\theta \}$ to
	maximize 
	\begin{align*}
		& \max_{k\geq 0,\theta }\alpha p-\left( \theta \times k+(1-\theta )\times
		1\right) \\
		\text{s.t. }& k+(1-k)(1-\theta )(1-\log (1-\theta ))=p.
	\end{align*}
	The Lagrangian is 
	\begin{equation*}
		\mathcal{L}(k,\theta ,\lambda ,\mu )=\alpha p-k\theta -(1-\theta )+\lambda
		\left( p-(k+(1-k)(1-\theta )(1-\log (1-\theta )))\right).
	\end{equation*}
	and the Euler-Lagrange equation for $\theta $ is 
	\begin{equation*}
		\dfrac{\partial \mathcal{L}}{\partial \theta }=(1-k)\left( 1-\lambda (\log
		(1-\theta ))\right) =0.
	\end{equation*}
	Therefore, 
	\begin{align*}
		\lambda =1/\log (1-\theta ).
	\end{align*}
	Hence, the Euler-Lagrange equation for $k$ is 
	\begin{align*}
		\dfrac{\partial \mathcal{L}}{\partial k} =-\theta -\lambda (\theta
		+(1-\theta )\log (1-\theta )) =\frac{\theta +\log (1-\theta )}{-\log
			(1-\theta )}<0.
	\end{align*}
	Therefore, the optimal $k=0$, and $F$ is a two-point distribution which puts
	mass only on the virtual values 0 and 1.
	
	In summary, we have reproduced the optimal signal distribution derived in 
	\cite{roesler2017buyer}, namely, 
	\begin{equation*}
		G(x)= 
		\begin{cases}
			1-\frac{1-\theta }{x}, & \text{ if }x\in \lbrack 1-\theta ,1); \\ 
			1, & \text{ if }x=1.%
		\end{cases}%
	\end{equation*}
	Under the optimal signal distribution, for $x\in \lbrack 1-\theta ,1)$, the
	virtual value is zero with probability $\theta $. For $x=1$, the virtual
	value is $1$ with probability $1-\theta $.
	
	\subsection{The case with $n \geq 2$}
	
	\label{rsn2}
	
	As in the case with $n=1$, we can reduce the infinite-dimensional
	information design problem to a finite-dimensional problem by the following
	lemma.
	
	\begin{lemma}
		\label{two-point} The support of any optimal virtual value distribution $F$
		has two points, say, $\{k,1\}$.
	\end{lemma}
	
	\begin{proof}
		The information design problem is also an isoperimetric problem in optimal 
		control theory. Define the following Lagrangian,  
		\begin{align*}
			\mathcal{L}(F,\lambda )=& \int_{0}^{1}\alpha n(1-F(k))\left(
			\sum_{i=1}^{n-1} \frac{-F^{i}(k)}{i}-\log (1-F(k))+\left( \sum_{i=1}^{n-1} 
			\frac{F^{i}(0^-)}{i} +\log (1-F(0^-))\right) \right)\mathrm{d}k \\
			& + \int_{0}^{1}(1-\alpha )F^{n}(k)-\lambda (1-F(k))(1-\log (1-F(k))+\log
			(1-F(0^-)))\mathrm{d}k+p\lambda +(\alpha -1)\text{.}
		\end{align*}
		Let $\theta _{0}=F(0^-)=G\left( 0\right)$, and let $\theta =F\left( k\right)
		$. By Theorem 4.2.1 of \cite{van2004isoperimetric}, the Euler-Lagrange 
		equation for $\mathcal{L}$ and for each state $k$ should be satisfied as 
		follows:  
		\begin{align*}
			\frac{\partial \mathcal{L}}{\partial \theta}=&-n\alpha
			\left(\sum_{i=1}^{n-1}-\theta ^{i}/i-\log(1-\theta)+\left(
			\sum_{i=1}^{n-1}\theta _{0}^{i}/i+\log (1-\theta
			_{0})\right)\right)+n\alpha(1-\theta)\left(\sum_{i=1}^{n-1}-\theta ^{i-1}+ 
			\frac{1}{1-\theta}\right) \\
			&+n(1-\alpha)\theta ^{n-1}-\lambda\log(1-\theta)+\lambda\log(1-\theta_0) \\
			=& n\alpha \sum_{i=1}^{n-1}\theta ^{i}/i+n\alpha \log (1-\theta )-\lambda
			\log (1-\theta ) +\left( -n\alpha \sum_{i=1}^{n-1}\theta _{0}^{i}/i-(n\alpha
			-\lambda )\log (1-\theta _{0})\right) \\
			&+n\alpha(1-\theta)\left(\frac{\theta^{n-1}-1}{1-\theta}+\frac{1}{1-\theta}
			\right)+n(1-\alpha)\theta^{n-1} \\
			=& n\alpha \sum_{i=1}^{n-1}\theta ^{i}/i+n\theta ^{n-1}+n\alpha \log
			(1-\theta )-\lambda \log (1-\theta ) +\left( -n\alpha \sum_{i=1}^{n-1}\theta
			_{0}^{i}/i-(n\alpha -\lambda )\log (1-\theta _{0})\right) =0.
		\end{align*}
		
		Denote $\partial\mathcal{L}/\partial\theta$ by $I_{\alpha }(\theta )$. Then 
		by taking the derivative of $I_{\alpha }(\theta )$ with respect to $\theta,$ we have:  
		\begin{align*}
			I_{\alpha }^{\prime }(\theta
			)&=n\alpha\sum_{i=1}^{n-1}\theta^{i-1}+n(n-1)\theta^{n-2}+\frac{
				-n\alpha+\lambda}{1-\theta}=\frac{n\alpha(1-\theta^{n-1})}{1-\theta}
			+n(n-1)\theta^{n-2}+\frac{-n\alpha+\lambda}{1-\theta} \\
			&=\frac{-n\alpha\theta^{n-1}+n(n-1)\theta^{n-2}(1-\theta)+\lambda}{1-\theta}
			= \frac{\lambda+n\theta^{n-2}\left(n-1-(n-1+\alpha)\theta\right)}{1-\theta}.
		\end{align*}
		
		We prove the following lemma in Appendix \ref{proofstablesoln}:
		
		\begin{lemma}
			\label{stablesoln}There is at most one $\theta $ with $I_{\alpha }(\theta 
			)=0 $ which also satisfies the second-order condition.\footnote{%
				In Figure \ref{local_figure}, we also draw the curve of $I_{\alpha }(\theta)$
				and $I_{\alpha }^{\prime }(\theta )$ to illustrate this lemma. In Figure \ref%
				{local_figure}, we choose the parameters to be $n=3$, $\alpha =1$, $\lambda
				=-0.5$, and $\theta _{0}=0$.} 
		\end{lemma}
		
		Therefore, $F(k)$ has only three values $0$, $\theta$, and $1$; hence, $F$ 
		is a two-point distribution. We then argue that for both $\alpha=0$ and $%
		\alpha=1$, signal 1 is on the support of $F$.
		
		\begin{enumerate}
			\item When $\alpha =0$, it is only when $\lambda <0$ that feasible solutions
			of $\theta $ exist; see cases 2 and 3 in Appendix \ref{proofstablesoln}. 
			Therefore,  
			\begin{align*}
				\left. I_{0}(\theta )\right\vert _{\theta \uparrow 1}& =\left. n\theta
				^{n-1}\right\vert _{\theta \uparrow 1}-\lambda (\log (1-\theta )|_{\theta
					\uparrow 1}+\lambda \log (1-\theta _{0})) \\
				& =\left. n+\lambda \log (1-\theta _{0})-\lambda \log (1-\theta )\right\vert
				_{\theta \uparrow 1}\rightarrow -\infty .
			\end{align*}
			Therefore, having $F(k)$ jump to $1$ before $k=1$ will never be optimal.
			
			\item When $\alpha =1$, it is only when $\lambda <n$, feasible solutions of $%
			\theta $ exist; see cases 2, 3 and 4 in Appendix \ref{proofstablesoln}. 
			Therefore,  
			\begin{align*}
				\left. I_{1}(\theta )\right\vert _{\theta \uparrow 1}=& \left. \left(
				n\sum_{i=1}^{n-1}\theta ^{i}/i+n\theta ^{n-1}+n\log (1-\theta )-\lambda \log
				(1-\theta )\right) \right\vert _{\theta \uparrow 1}\  \\
				& +\left( -n\sum_{i=1}^{n-1}\theta _{0}^{i}/i-(n-\lambda )\log (1-\theta
				_{0})\right) \\
				=& n+n\sum_{i=1}^{n-1}1/i+\left( -n\sum_{i=1}^{n-1}\theta
				_{0}^{i}/i-(n-\lambda )\log (1-\theta _{0})\right) \\
				& +\left. (n-\lambda )\log (1-\theta )\right\vert _{\theta \uparrow
					1}\rightarrow -\infty .
			\end{align*}
			Therefore, having $F(k)$ jump to $1$ before $k=1$ will never be optimal. 
		\end{enumerate}
		
		In summary, the support of $F$ has two points $\{k,1\}$. 
	\end{proof}
	
	\begin{figure}[tbp]
		\centering
		\begin{minipage}[c]{0.45\textwidth}
			\centering
			\begin{tikzpicture}
				\begin{axis}[
					tick label style={font=\scriptsize},
					xlabel={$\theta$},
					ylabel={ $I'(\theta)$},
					ytick={-0.7,0,  0.7},
					yticklabels={-0.7,0, 0.7},
					xtick={0, 0.099,0.566},
					xticklabels={0, $\theta_1$, $\theta_2$},
					no markers,
					line width=0.3pt,
					cycle list={{red,solid}},
					samples=200,
					smooth,
					domain=-2:1.3,
					xmin=0, xmax=0.8,
					ymin=-0.7, ymax=0.7,
					width=6cm, height=6cm,
					legend cell align=left,
					legend pos=  north west,
					legend style={draw=none,fill=none,name=legend},
					]
					\addplot[blue,thick,domain=0:1]{-0.5+3*x*(2-3*x)};
					\addplot[thick] coordinates {
						(0,0)
						(1,0)
					};
					\addplot[dashed] coordinates {
						(0.099,0)
						(0.099,-0.7)
					};
					\addplot[dashed] coordinates {
						(0.566,0)
						(0.566,-0.7)
					};
					\node at (axis cs:0.33,0.2) {$I_1(\theta)\nearrow$ };
					\node at (axis cs:0.7,-0.3) {$I_1(\theta)\searrow$};
					\node at (axis cs:0.1,-0.3) {$I_1(\theta)\searrow$};
				\end{axis}
			\end{tikzpicture}
		\end{minipage} 
		\begin{minipage}[c]{0.45\textwidth}
			\centering
			\begin{tikzpicture}
				\begin{axis}[
					tick label style={font=\scriptsize},
					xlabel={$\theta$},
					ylabel={ $I_1(\theta)$},
					ytick={-0.5,0,  0.5},
					yticklabels={-0.5,0, 0.5},
					xtick={0, 0.099,0.566},
					xticklabels={0, $\theta_1$, $\theta_2$},
					no markers,
					line width=0.3pt,
					cycle list={{red,solid}},
					samples=200,
					smooth,
					domain=-2:1.3,
					xmin=0, xmax=0.8,
					ymin=-0.5, ymax=0.5,
					width=6cm, height=6cm,
					legend cell align=left,
					legend pos=  north west,
					legend style={draw=none,fill=none,name=legend},
					]
					\addplot[red,thick,domain=0:1]{1.5*x*(2+3*x)+3.5*ln(1-x)};
					\addplot[thick] coordinates {
						(0,0)
						(1,0)
					};
					\addplot[dashed] coordinates {
						(0.099,-0.03)
						(0.099,-0.7)
					};
					\addplot[dashed] coordinates {
						(0.566,0.22)
						(0.566,-0.7)
					};
					\node [coordinate,pin=90:{local min}]
					at (axis cs:0.2,0) {};
					\node [coordinate,pin=245:{local max}]
					at (axis cs:0.735,0) {};
				\end{axis}
			\end{tikzpicture}
		\end{minipage}
		\caption{ The curve of $I_1^{\prime }(\protect\theta)$ and $I_1(\protect%
			\theta).$ }
		\label{local_figure}
	\end{figure}
	
	Therefore, the information designer will choose the optimal $k,\theta $, and 
	$\theta _{0}$ (with $\theta =F(k)$ and $\theta _{0}=G(0)=F(0^{-})$) such
	that 
	\begin{align}
		\max_{\theta _{0},k,\theta }& \left( \alpha n(1-k)(1-\theta )\left(
		\sum_{i=1}^{n-1}\frac{-\theta ^{i}}{i}-\log (1-\theta )+\left(
		\sum_{i=1}^{n-1}\frac{\theta _{0}^{i}}{i}+\log (1-\theta _{0})\right)
		\right) \right)  \notag \\
		& +(\alpha -1)+(1-\alpha )(1-k)\theta ^{n}+(1-\alpha )k\theta _{0}^{n}
		\label{finiteinfo} \\
		\text{s.t. }& (1-k)\left( (1-\theta )(1-\log (1-\theta )+\log (1-\theta
		_{0})\right) +k(1-\theta _{0})=p,  \notag \\
		& k\geq 0,\quad \theta _{0}\geq 0,\quad \text{and}\quad \theta _{0}\leq
		\theta \leq 1.  \notag
	\end{align}
	The solution to this finite-dimensional optimization problem is standard and
	we present the details in Appendices \ref{fsw} and \ref{fbo}. In fact, the
	solution to this finite-dimensional problem is unique. Since the information
	design problem has at least one solution by Lemma \ref{existence}, this
	unique solution is globally optimal; this concludes the proofs of Theorems %
	\ref{result_seller} and \ref{result_buyer}.
	
	We now wish to briefly comment on how our approach differs from that of \cite%
	{roesler2017buyer}. When there is only one buyer, a posted price mechanism
	is optimal. Hence, a signal distribution matters only in determining the
	optimal posted price, i.e., $\varphi ^{-1}\left( 0\right) $. This is how 
	\cite{roesler2017buyer} are able to argue that it entails no loss of
	generality to focus on a class of Pareto distributions, and the two-point
	Pareto distribution with virtual values $0$ and $1$ is the buyer-optimal
	information structure. When there are multiple buyers, we may take the
	second-price auction with an optimal reserve price to be an extension of the
	posted price mechanism. Unlike \cite{roesler2017buyer}, however, a
	second-price auction with a reserve price need not be an optimal auction
	against an irregular signal distribution.
	
	To explain, while an irregular signal distribution can be ironed into a
	regular signal distribution, the optimal expected revenue under the
	irregular distribution is the same as the revenue of a second-price auction
	with a reserve price under the regular distribution obtained from ironing
	rather than the irregular distribution.\footnote{%
		For the sake of completeness, we provide an example in Appendix \ref%
		{irexample} to illustrate this point; see also the working paper version of 
		\cite{monteiro2010optimal} for a similar example. In particular, our example
		admits a density and nonnegative (ironed) virtual values for every signal.}
	We show in Section \ref{continuous prior} that when the seller is committed
	to using a second-price auction with reserve price $0$, then for $n=2$, the
	seller-worst information structure is the binary prior (i.e., full
	revelation). Hence, contrary to Theorem \ref{result_seller}, the information
	structure which minimizes the seller's revenue in a second-price auction is
	an irregular distribution.
	
	\section{Asymmetric information structures}
	
	\label{Asymmetric information structures}
	
	So far we have assumed that the information designer chooses the same signal
	distribution across all buyers. A natural question to ask is whether the
	information designer can do better by choosing different signal
	distributions for different buyers. The short answer is \textquotedblleft
	No\textquotedblright\ for the seller-worst information design problem and
	\textquotedblleft Yes\textquotedblright\ for the buyer-optimal information
	design problem. We elaborate further below.
	
	To allow for asymmetric signal distributions, we first reformulate the
	information design problems. Let $M(x\vert\mathbf{G})=\{i\in N|\hat{\varphi}
	(x_{i}|G_{i})\geq \max_{j}\{\hat{\varphi}(x_{j}|G_{j}),0\}\}$ be the set of
	buyers who have the largest nonnegative virtual value for a given signal
	realization $x$, where $\mathbf{G}$ stands for $\{G_{i}\}_{i=1}^{n}$; and
	let $M^{\prime }(x\vert\mathbf{G})=\{i\in N|x_{i}\geq \max_{j}{x_{j}}
	,\forall j\in M(x\vert\mathbf{G})\}$ be the set of buyers who not only have
	the largest virtual value but also the largest signal among those with the
	highest virtual value for a given signal realization $x$. Then, the optimal
	auction allocation rule for buyer $i$ when all buyers report their signals
	is given by the following: 
	\begin{equation*}
		q_{i}(x_{i},x_{-i}\vert\mathbf{G})= 
		\begin{cases}
			\frac{1}{|M^{\prime }(x\vert\mathbf{G})|}, & \mbox{if }i\in M^{\prime
			}(x\vert\mathbf{G}); \\ 
			0, & \mbox{if }i\not\in M^{\prime }(x\vert\mathbf{G}).%
		\end{cases}%
	\end{equation*}
	
	We now study the following information design problem: 
	\begin{align}
		\max_{\{G_{i}\}_{i=1}^{n}}\int_{[0,1]^{n}}& \sum_{i=1}^{n}\left( \alpha
		x_{i}-\hat{\varphi}(x_{i}|G_{i})\right) q_{i}(x_{i},x_{-i}\vert\mathbf{G}
		)\prod_{i=1}^{n}\left( \mathrm{d}G_{i}(x_{i})\right)  \label{info2} \\
		\text{s.t. }& \int_{0}^{1}1-G_{i}(x_{i})\mathrm{d}x_{i}=p,\quad\forall
		i=1,\cdots ,n.  \notag
	\end{align}
	While we allow for asymmetric signal distributions in the information design
	problem (\ref{info2}), we still assume that the agents' binary priors have
	the same mean $p$. We relegate the discussion about asymmetric priors to
	Section \ref{apm}.
	
	\subsection{The seller-worst case}
	
	In this section, we show that the optimal symmetric seller-worst information
	structure in Theorem \ref{result_seller} remains the unique optimal
	solution, even if the information designer can choose an asymmetric
	information structure. We first document the existence of the solution to
	problem \eqref{info2} when $\alpha =0$; see Lemma \ref{existence2}. The
	proof is similar to the proof of Lemma \ref{existence} and is therefore
	omitted.
	
	\begin{lemma}
		\label{existence2}For the problem in \eqref{info2} with $\alpha =0$, an
		optimal solution exists.
	\end{lemma}
	
	The following Lemma \ref{nonneg}, corresponds to Lemmas \ref{almost_nonneg}
	and \ref{regularity}. The proof is similar so we provide only a sketch of it
	in Appendix \ref{proofnonneg}.
	
	\begin{lemma}
		\label{nonneg}Any profile of optimal signal distributions $\left\{
		G_{i}\right\} _{i=1}^{n}$ which solves the information design problem in %
		\eqref{info2} must be regular and induce nonnegative virtual values almost
		everywhere on $\left[ 0,1\right] $.
	\end{lemma}
	
	Then, as in Lemma \ref{change_variables}, it follows from Lemma \ref{nonneg}
	that the asymmetric information design problem can be reformulated as: 
	\begin{align}
		\max_{\{F_{i}\}_{i=1}^{n}}& \int_{0}^{1}\prod_{i=1}^{n}F_{i}(k)\mathrm{d} k-1
		\label{adp} \\
		\text{s.t. }& \int_{0}^{1}(1-F_{i}(k))(1-\log (1-F_{i}(k)))\mathrm{d}
		k=p,\quad \forall i=1,\cdots ,n.  \notag
	\end{align}
	
	We now state and prove the following theorem.
	
	\begin{theorem}
		\label{asy_seller}The unique seller-worst information structure in Theorem %
		\ref{result_seller} remains the unique seller-worst information structure
		which solves the problem in \eqref{info2} with $\alpha =0$.
	\end{theorem}
	
	\begin{proof}
		For any profile of virtual value distributions $\{F_{i}\}_{i=1}^{n}$, let $%
		F(k)\equiv \frac{1}{n}\sum_{i=1}^{n}F_{i}(k)$, and denote by $F$ a symmetric
		signal distribution profile where each buyer receives his signal according 
		to $F$.
		
		First, the symmetric signal distribution $F$ yields weakly less revenue than
		$\{F_{i}\}_{i=1}^{n}$ does. For any $k$, by the inequality of arithmetic and
		geometric means, we have  
		\begin{equation}
			F^{n}(k)=\left( \frac{1}{n}\sum_{i=1}^{n}F_{i}(k)\right) ^{n}\geq
			\prod_{i=1}^{n}F_{i}(k).  \label{ag}
		\end{equation}
		Moreover, the equality in (\ref{ag}) holds if and only if $F_{i}(k)=F$ for 
		all $i$. Integrating both sides yields $\int_{0}^{1}F^{n}(k)\mathrm{d} k\geq
		\int_{0}^{1}\prod_{i=1}^{n}F_{i}(k)\mathrm{d}k$. Hence, $F$ yields weakly 
		less revenue than $\{F_{i}\}_{i=1}^{n}$ does and strictly less revenue when $%
		F_{i}\neq F$ for some $i$.
		
		Second, $F$ generates a weakly higher mean than $p$. Indeed, the integral 
		term in the mean constraint is strictly concave with respect to $F(k)$. That
		is, $I^{\prime \prime }(\theta )=-1/(1-\theta )<0$ where $I(\theta 
		)=(1-\theta )(1-\log (1-\theta ))$. Next, define  
		\begin{equation*}
			\hat{F}\left( k\right) = 
			\begin{cases}
				F(k_{1}^{-}), & \mbox{if }k\in \lbrack 0,k_{1}); \\ 
				F(k), & \mbox{if }k\in \lbrack k_{1},1],%
			\end{cases}%
		\end{equation*}
		for some $k_{1}$ so that $\hat{F}$ satisfies the constraint in (\ref{adp}). 
		Since $\hat{F}(k)\geq F(k)$ for any $k$, it follows that $\int_{0}^{1}\hat{F}%
		^{n}(k)\mathrm{d}k\geq \int_{0}^{1}F^{n}(k)\mathrm{d}k$ and $\hat{F}$ yields
		less revenue than $F$. Therefore, the improvement from $\{F_{i}\}_{i=1}^{n}$
		to $\hat{F}$ is strict when $F_{i}\neq F$ for some $i$. It follows that the 
		symmetric seller-worst information structure in Theorem \ref{result_seller} 
		remains the unique solution to the problem in (\ref{adp}) with $\alpha =0$. 
	\end{proof}
	
	\subsection{The buyer-optimal case}
	
	\label{asyboc}
	
	For the buyer-optimal information design problem, we demonstrate that an
	asymmetric signal distribution can strictly improve upon the optimal
	symmetric signal distribution in Theorem \ref{result_buyer}. We demonstrate
	such an improvement for the case of $n=2$ with $p<r_{b}$, and the case of $%
	n\rightarrow \infty $.
	
	\begin{proposition}
		\label{asy_buyer}For $n=2$ with $p<r_{b}$ or for $n\rightarrow \infty $,
		there exist asymmetric information structures which can strictly improve
		upon the optimal symmetric signal distribution in Theorem \ref{result_buyer}
		and Corollary \ref{equivalence}, respectively.
	\end{proposition}
	
	\begin{proof}
		See Appendix \ref{proofasybuyer}. 
	\end{proof}
	
	To see the main idea, for the case where $n=2$ with $p<r_{b}$, we focus our
	search for improvement on the signal distributions which put a positive mass
	only on signal $0$, on signals with virtual value $0$, and on signals with
	virtual value $1 $ for each buyer. Since each buyer's virtual value
	distribution needs to satisfy the mean constraint, it is uniquely determined
	by its probability assigned to signal $0$. We then optimize within the
	specific class of information structures with these two probabilities
	assigned to signal $0$ (one for each buyer). It turns out that within this
	class of distributions, the two buyers' aggregate surplus is maximized when
	one buyer's signal distribution assigns positive probability on signal $0$,
	whereas the other buyer's signal distribution assigns zero probability.
	Moreover, the former buyer benefits, while the latter one loses relative to
	the symmetric buyer-optimal information; see Appendix \ref{proofasybuyer}
	for more details.
	
	For the case where $n\rightarrow \infty $, we also consider a specific class
	of information structures in which (i) buyer $i$'s signal distribution puts
	positive mass only on signal $0$, on signals with virtual value $p$, and on
	signals with virtual value $1$; and (ii) all the other buyers' signal
	distributions are a degenerate distribution which puts the entire mass on
	signal $p$. We will construct an asymmetric information structure in which
	no buyer is worse off and some buyer is strictly better off, relative to the
	symmetric buyer-optimal information structure.
	
	Proposition \ref{asy_buyer} shows that the buyer-optimal signal
	distributions, whether there exists one or many, are asymmetric in general.
	Also in general, the total surplus as a max function of buyers' value favors
	dispersion. Indeed, while averaging the signal distributions ($F=\frac{1}{n}
	F_{i}$) can reduce the revenue, it might reduce the total surplus. The issue
	is reminiscent of the result in \cite{bergemann2007information} which shows
	that a seller-optimal information structure is asymmetric across all buyers.
	In \cite{bergemann2007information}, an information structure is chosen to
	maximize the expected nonnegative virtual surplus, which also favors
	asymmetry/dispersion. Therefore, relative to the optimal symmetric
	information structure, an asymmetric information structure increases the
	total surplus as well as the seller's revenue. Proposition \ref{asy_buyer}
	demonstrates that we can choose an asymmetric information structure which
	increases the expected total surplus more than the expected (nonnegative)
	virtual surplus.\footnote{%
		However, the existence or the exact shape of an asymmetric buyer-optimal
		information structure remains unknown to us at this time.}
	
	\section{Discussion}
	
	\label{discussion}
	
	So far we have studied our information design problem with ex ante symmetric
	binary priors and a Myersonian optimal auction. In this section, we will
	first discuss the issues with asymmetric priors and continuous priors. Next,
	we will discuss the tightness of our seller-worst revenue upper bound for
	the revenue guarantee of informationally robust auctions.
	
	\subsection{Asymmetric prior mean}
	
	\label{apm}
	
	Suppose that each buyer $i$ has his own prior mean $p_{i}$. We discuss only
	the seller-worst problem for which we know a solution exists. The
	seller-worst information problem is to maximize the same objective function
	in \eqref{adp} with each buyer $i$'s individual mean constraint now being: 
	\begin{equation*}
		\int_{0}^{1}(1-F_{i}(k))(1-\log (1-F_{i}(k))+\log (1-F_{i}(0^{-})))\mathrm{d}
		k=p_{i}.
	\end{equation*}
	In this case, our arguments for regularity and nonnegative virtual values
	are still valid except at zero. For $n=2$, we obtain the seller-worst
	information structure in the following proposition.
	
	\begin{proposition}
		\label{ap}Suppose that $n=2$ and each buyer $i$ has a prior mean $p_{i}$.
		Then, the seller-worst information structure is a profile of signal
		distributions $\{G_1,G_2\}$ such that, for each buyer $i$, $G_{i}$ is a
		truncated Pareto distribution as follows: 
		\begin{equation*}
			G_{i}(x)= 
			\begin{cases}
				1-\frac{x_{i}}{x}, & \text{ if }x\in \lbrack x_{i},1); \\ 
				1, & \text{ if }x=1,%
			\end{cases}%
		\end{equation*}
		where $x_{i}$ is determined by buyer $i$'s mean constraint.
	\end{proposition}
	
	\begin{proof}
		See Appendix \ref{proofap}. 
	\end{proof}
	
	Thus, both buyers have the virtual values $\{0,1\}$ and the degree of
	asymmetry between $p_{1}$ and $p_{2}$ is reflected only by $x_{1}$ and $%
	x_{2} $. For $n\geq 3$, the seller-worst problem remains an isoperimetric
	problem and we can similarly reduce it to a finite-dimensional constrained
	optimization problem as we do in Theorem \ref{result_seller}. The
	finite-dimensional problem, however, becomes intractable and its closed-form
	solution(s) remain unknown to us.
	
	\subsection{Continuous prior distributions}
	
	\label{continuous prior}
	
	We now consider continuous prior distributions. Our analysis in Sections \ref%
	{main results}--\ref{Asymmetric information structures} remains applicable,
	if the information designer is allowed to choose any profile of signal
	distributions with a given mean $p$. However, if the information designer
	has full information about the prior, then a signal distribution is feasible
	if and only if it is a mean-preserving contraction of the prior; see \cite%
	{blackwell1953equivalent}. That is, the set of feasible signal distributions
	becomes 
	\begin{equation*}
		\mathcal{G}_{H}=\left\{ G:[0,1]\rightarrow \lbrack 0,1]\bigg \vert %
		\int_{0}^{1}x\,\mathrm{d}G(x)=p,\int_{0}^{x}G(t)\mathrm{d}t\leq
		\int_{0}^{x}H(t)\mathrm{d}t,\forall x\in \lbrack 0,1]\right\} .
	\end{equation*}
	
	The main issue here is to handle the mean-preserving spread constraint on
	the signal distributions. Changing the variable is no longer useful so we
	need a different tool akin to the method developed in \cite%
	{dworczak2019simple}, although our objective function is still different
	from that in \cite{dworczak2019simple}. Consider, for instance, the case of
	a second-price auction (with no reserve). The seller-worst problem can be
	expressed as 
	\begin{align}
		\max_{G}& \int_{0}^{1}\left( nG^{n-1}\left( x\right) -(n-1)G^{n}\left(
		x\right) -1\right) \mathrm{d}x  \label{2nd} \\
		\text{s.t. }& H\text{ is a mean-preserving spread of }G.  \notag
	\end{align}
	
	We obtain the following result for the case with two buyers:
	
	\begin{proposition}
		\label{fullyreveal} \label{cps}For $n=2$, full revelation (i.e., $G=H$)
		solves the problem in \eqref{2nd}.
	\end{proposition}
	
	\begin{proof}
		The objective is  
		\begin{equation*}
			\int_{0}^{1}2G-G^{2}-1\mathrm{d}x=-\int_{0}^{1}x\mathrm{d}\left(
			2G-G^{2}\right) =\int_{0}^{1}x\mathrm{d}G^{2}(x)-2p.
		\end{equation*}
		Moreover, $\int_{0}^{1}x\mathrm{d}G^{2}(x)$ is maximized if $G=H$, since $%
		G^{2}$ is the CDF of the convex function, $\max \left\{ x_{1},x_{2}\right\}$%
		, when $x_{1}$ and $x_{2}$ are independently distributed according to $G$. 
		Hence, full revelation minimizes the sellers' revenue. 
	\end{proof}
	
	Proposition \ref{cps} has an implication for both the seller-worst problem
	and the buyer-optimal problem with two buyers. Specifically, we can still
	argue that the seller-worst signal distribution must be regular and
	symmetric, and must admit only nonnegative virtual values. Since a
	second-price auction with a reserve price is optimal for such signal
	distributions, identifying a seller-worst information structure amounts to
	solving the problem in (\ref{2nd}). If the prior distribution is regular and
	admits nonnegative virtual values, it is also a feasible choice in this
	problem. Hence, the proof of Proposition \ref{cps} demonstrated that the
	full revelation is the seller-worst information structure. Moreover, since
	full revelation maximizes the total surplus, it is also the buyer-optimal
	information structure. The following corollary summarizes our findings.\footnote{
	Based on 
	our result that the seller-worst information is regular and induces 
	nonnegative virtual value, Corollary \ref{coro-continuous} can alternatively obtained from	 Part (i) in Theorem
		5 of \cite{ganuza2010signal}, which shows that in a second-price auction
		with two bidders, the seller's revenue is nonincreasing in the precision of
		the signals.}
	
	\begin{corollary}
		\label{coro-continuous}For $n=2$, if the prior distribution is regular and
		admits nonnegative virtual values, then full revelation is both the unique
		symmetric buyer-optimal and the unique symmetric seller-worst information
		structure.
	\end{corollary}
	
	Corollary \ref{coro-continuous} restores the equivalence between the
	buyer-optimal information structure and the seller-worst information
	structure. However, it requires the prior to be regular and induce only
	nonnegative virtual values. In particular, it rules out the binary prior
	which we analyze in Theorems \ref{result_seller} and \ref{result_buyer}.\footnote{%
		If $n\geq 3$ or the prior is irregular, we can still solve the problem in %
		\eqref{2nd}; we report the solution in \cite{chen2020quantile}. However, the
		optimal signal distributions which we obtain in these situations are no
		longer regular. As we demonstrate in Example \ref{irexample}, the optimal
		auction need not be a second-price auction with a reserve price.} In
	Appendix \ref{inequivir}, we construct a regular continuous prior under which the buyer-optimal and seller-worst
	information structures are not equivalent. The example clarifies that the
	inequivalence of the seller-worst information and the buyer-optimal
	information arises not because the prior is discrete but rather it induces
	negative virtual values.
	
	The uniqueness of the buyer-optimal information structure in Corollary \ref%
	{coro-continuous} also contrasts with the multiplicity of buyer-optimal
	information structures available with a single buyer.\footnote{%
		For example, when $n=1$ and the prior is uniformly distribution on $[1/2,1]$%
		, the buyer-optimal information structure identified in \cite%
		{roesler2017buyer} is a truncated Pareto distribution $G(x)=1-\frac{1}{2x}$
		for $x\in \lbrack 1/2,0.824)$ and $G(x)=1$ for $x\in \lbrack 0.824,1]$.
		Against both the truncated Pareto distribution as well as the prior, the
		seller chooses the posted price $1/2$ and obtains revenue $1/4$.} Since the
	information designer selects only signal distributions which are regular and
	admit nonnegative virtual values, the seller adopts a second-price auction.
	When $n=2$, the objective function in \eqref{2nd} is strictly convex in $G$;
	hence, full revelation is the unique buyer-optimal information structure.
	When $n=1$, on the other hand, the objective becomes linear in $G$. Hence,
	any signal distribution that induces (i) the same posted price and (ii) the
	same probability of exceeding the posted price of a buyer-optimal signal
	distribution is also buyer-optimal.
	
	\subsection{Tightness of the seller-worst upper bound}
	
	\label{simulationdiscussion}
	
	As we mentioned in the introduction, our seller-worst revenue provides an
	upper bound for the revenue guarantee of any mechanism over all symmetric
	independent information structures and undominated equilibria. To illustrate
	the effectiveness of this upper bound, we consider a specific example with
	prior mean $p=0.5$ and two buyers. It follows from Theorem \ref%
	{result_seller} that the seller-worst revenue is $2x_s-x_s^{2}=0.3385$,
	where $x_s\approx 0.1867$ solves $x_s-x_s\log (x_s)=p$.
	
	We now explore the revenue guaranteed by a second-price auction with a
	random reserve price distributed under the truthful equilibrium. This is
	consistent with the approach as taken by \cite{che2019distributionally}, 
	\cite{Park2021public}, and \cite{Zhang2021bilateraltrade} in their study of
	the optimal revenue guarantee by a dominant-strategy mechanism with a mean
	constraint.\footnote{\cite{Park2021public} studies the optimal revenue
		guarantee in a public good provision setting, \cite{Zhang2021bilateraltrade}
		studies the problem in a bilateral trade setting, and \cite%
		{che2019distributionally} studies the problem in an auction setting where
		mechanisms satisfy a condition called \textquotedblleft
		competitiveness\textquotedblright. These papers all adopt a duality approach
		and allow for correlated signal distributions. In our example, the
		guaranteed revenue in \cite{che2019distributionally} is $0.317$ which is
		about $93.6$ percent of our seller-worst revenue. The optimal revenue
		guarantee in general dominant-strategy auctions remains unknown to us.}
	Formally, denote by $R$ the distribution function of a random reserve price.
	Let $\sigma ^{T}$ be the truth-telling equilibrium, i.e., $\sigma
	_{i}^{T}\left( x_{i}\right) =x_{i}$ for each $x_{i}$; in addition, let $\Pi
	\left( R,G,\sigma ^{T}\right) $ denote the seller's revenue under $\sigma
	^{T}$. The revenue guaranteed by $R $ (under the truthful equilibrium) is
	defined as $\min_{G}\Pi \left( R,G,\sigma ^{T}\right) $.
	
	First, we examine the performance of a deterministic reserve price. In
	particular, \cite{suzdaltsev2020distributionally} proves that $r=0$ solves 
	\begin{equation*}
		\max_{r}\min_{G}\Pi \left( r,G,\sigma ^{T}\right) \text{.}
	\end{equation*}
	It follows from Proposition \ref{cps} that the prior/full revelation solves $%
	\min_{G}\Pi \left( 0,G,\sigma ^{T}\right) $. As a result, $\min_{G}\Pi
	\left( 0,G,\sigma ^{T}\right) =p^{2}=0.25$. Hence, any deterministic reserve
	guarantees the revenue $0.25$, which is about $73.85$ percent of the
	seller-worst revenue. In fact, by the same argument, for any $p$, any
	deterministic reserve guarantees at most $p^{2}/(2x_s-x_s^{2})$ of the
	seller-worst revenue.
	
	Second, we examine the performance of a specific random reserve price.
	Consider 
	\begin{equation*}
		R_{b}(r)= 
		\begin{cases}
			0, & \text{ if }r\in \lbrack 0,e^{-1/b}); \\ 
			1+b\times \log (r), & \text{ if }r\in \lbrack e^{-1/b},1].%
		\end{cases}%
	\end{equation*}
	In particular, $R_{b}$ becomes the random posted price due to \cite%
	{carrasco2018optimal} for $b=\frac{1}{-\log x_s}\approx 0.5958$.\footnote{%
		\cite{carrasco2018optimal} proves that a random posted price attains the
		seller-worst revenue when there is a single buyer.} Our simulation result
	shows that when we choose $b=0.5958\,$, we have $\min_{G}\Pi \left(
	R_{b},G,\sigma ^{T}\right) =0.2939$, which is about $86.8$ percent of the
	seller-worst revenue. In contrast, when $b=0.339$, we have $\min_{G}\Pi
	\left( R_{b},G,\sigma ^{T}\right) =0.3382$, which is about $99.9$ percent of
	the seller-worst revenue.
	
	Third, \cite{Chen2022informationallyrobust} has identified, for $n=2$ and
	any prior mean $p$, a random reserve price distribution $R^{\ast }(r)$ which
	guarantees the seller-worst revenue $2x_s-x_s^{2}$.\footnote{\cite%
		{bachrach2022distributional} incorporates material from these three other
		recent independent working papers: \cite{bachrach2022distributional}, \cite%
		{chen2022note}, and \cite{zhang2022information}.} In particular, 
	\begin{equation*}
		R^{\ast }(r)= 
		\begin{cases}
			\frac{r(1-x_s)}{r-x_s}\left(1+\frac{1}{-\log(x_s)}\times\log(r)\right), & 
			\text{ if }r\in ( 0,x_{s}) \cup \left(x_s,1\right]; \\ 
			\frac{1-x_{s}}{-\log (x_{s})}, & \text{ if }r=x_{s}; \\ 
			0, & \text{ if }r=0.%
		\end{cases}%
	\end{equation*}
	
	\section{Conclusion}
	
	\label{conclusion} In this paper, we characterize the symmetric
	buyer-optimal information structure as well as the symmetric seller-worst
	information structure with symmetric binary priors and a Myersonian optimal
	auction. We show that with a binary i.i.d. prior on $0$ and $1$, the two
	information structures are not equivalent, and yet both converge to ``no
	disclosure'' when the number of buyers goes to infinity. We also demonstrate
	that an asymmetric information structure is never seller-worst but can
	generate a strictly higher surplus for the buyers on an aggregate level.
	
	The independent private-value setting enables us to express both the buyers'
	surplus as well as the seller's revenue in terms of the buyers' (ironed)
	virtual values. This approach thereby neatly subsumes the IC and IR
	constraints and leads to an information design problem amenable to optimal
	control. For general correlated signal distributions, however, we know of no
	way to express the seller's optimal revenue or the buyers' surplus with the
	intractable (binding) IC and IR constraints.\footnote{\cite%
		{mathevet2020information} study the information design problem allowing for
		correlated signals under a fixed game, whereas the game in our information
		design problem is chosen optimally by the seller in response to the
		information structure.} One way to bypass this difficulty is to appeal to
	the strong duality approach used in \cite{du2018robust} and \cite%
	{brooks2019optimal}. That approach requires identifying a
	seller-worst/minmax information structure (among all correlated signal
	distributions) together with a maxmin mechanism that achieves the
	seller-worst revenue upper bound. We leave this important yet challenging
	question for future research.
	
	Do the optimal information structures that we have provided resemble any
	real-world information structures? We do not have an answer. As information
	structures are inherently harder to observe than, say, contracts or selling
	mechanisms, we must also maintain the awareness that the predictions we have
	derived, like some of those in contract theory, might be entirely
	counterfactual. As we have demonstrated, however, the optimal information
	structures do provide useful theoretical benchmarks which shed light on
	other problems such as strategic information acquisitions or optimal revenue
	guarantees.
	
	\bibliographystyle{econometrica}
	\bibliography{IDOA_R2}

@article{bergemann2016informationally,
  title={Informationally Robust Optimal Auction Design},
  author={Bergemann, Dirk and Brooks, Benjamin A and Morris, Stephen},
 journal={Working paper},
  year={2016}
}

@article{chen2022statistical,
	title={A Statistical Learning Approach to Personalization in Revenue Management},
	author={Chen, Xi and Owen, Zachary and Pixton, Clark and Simchi-Levi, David},
	journal={Management Science},
	volume={68},
	number={3},
	pages={1923--1937},
	year={2022},
	publisher={INFORMS}
}

@article{Chen2022informationallyrobust,
	title={Distributionally Robust  Auction Design},
	author={Bachrach, Nir and Chen, Yi-Chun and Talgam-Cohen, Inbal and Yang, Xiangqian and Zhang, Wanchang},
	journal={Working paper},
	year={2022}
}

@article{ganuza2010signal,
	title={Signal Orderings Based on Dispersion and the Supply of Private Information in Auctions},
	author={Ganuza, Juan-Jos{\'e} and Penalva, Jose S},
	journal={Econometrica},
	volume={78},
	number={3},
	pages={1007--1030},
	year={2010},
	publisher={Wiley Online Library}
}

@article{chen2022note,
	title={Note on Revenue Guarantee for Second-Price Auctions with Random Reserve},
	author={Chen, Yi-Chun and Yang, Xiangqian},
		journal={Working paper},
	year={2022}
}

@article{zhang2022information,
title={Information-Robust Optimal Auctions},
author={Zhang, Wanchang},
journal={Working paper},
year={2022}
}

@article{bachrach2022distributional,
	title={Distributional Robustness: From Pricing to Auctions},
	author={Bachrach, Nir and Talgam-Cohen, Inbal},
	journal={Working paper},
	year={2022}
}

@article{shi2012optimal,
	title={Optimal Auctions with Information Acquisition},
	author={Shi, Xianwen},
	journal={Games and Economic Behavior},
	volume={74},
	number={2},
	pages={666--686},
	year={2012},
	publisher={Elsevier}
}

@article{mathevet2020information,
title={On Information Design in Games},
author={Mathevet, Laurent and Perego, Jacopo and Taneva, Ina},
journal={Journal of Political Economy},
volume={128},
number={4},
pages={1370--1404},
year={2020},
publisher={The University of Chicago Press Chicago, IL}
}

@article{monteiro2015note,
  title={A Note on the Continuity of the Optimal Auction},
  author={Monteiro, Paulo Klinger},
  journal={Economics Letters},
  volume={137},
  pages={127--130},
  year={2015},
  publisher={Elsevier}
}

@article{suzdaltsev2020distributionally,
  title={Distributionally Robust Pricing in Independent Private Value Auctions},
  author={Suzdaltsev, Alex},
  journal={Working paper},
  year={2020}
}

@book{van2004isoperimetric,
  title={The Calculus of Variations},
  author={van Brunt, Bruce},
  year={2004},
  publisher={Springer}
}

@article{bergemann2007information,
  title={Information Structures in Optimal Auctions},
  author={Bergemann, Dirk and Pesendorfer, Martin},
  journal={Journal of Economic Theory},
  volume={137},
  pages={580--609},
  year={2007},
  publisher={Elsevier}
}

@article{che2019distributionally,
  title={Distributionally Robust Optimal Auction Design under Mean Constraints},
  author={Che, Ethan},
  journal={Working paper},
  year={2022}
}

@Article{Myerson1981,
  author    = {Myerson, Roger B},
  journal   = {Mathematics of Operations Research},
  title     = {Optimal Auction Design},
  year      = {1981},
  pages     = {58--73},
  volume    = {6},
  publisher = {INFORMS},
}

@article{yang2019buyer,
  title={Buyer-Optimal Information with Nonlinear Technology},
  author={Yang, Kai Hao},
  journal={Working paper},
  year={2018}
}

@article{dworczak2019simple,
  title={The Simple Economics of Optimal Persuasion},
  author={Dworczak, Piotr and Martini, Giorgio},
  journal={Journal of Political Economy},
  volume={127},
  pages={1993--2048},
  year={2019},
  publisher={The University of Chicago Press Chicago, IL}
}

@article{krahmer2020information,
	title={Information Disclosure and Full Surplus Extraction in Mechanism Design},
	author={Kr{\"a}hmer, Daniel},
	journal={Journal of Economic Theory},
	volume={187},
	pages={105020},
	year={2020},
	publisher={Elsevier}
}

@article{carrasco2018optimal,
  title={Optimal Selling Mechanisms under Moment Conditions},
  author={Carrasco, Vinicius and Luz, Vitor Farinha and Kos, Nenad and Messner, Matthias and Monteiro, Paulo and Moreira, Humberto},
  journal={Journal of Economic Theory},
  volume={177},
  pages={245--279},
  year={2018},
  publisher={Elsevier}
}

@article{chen2020quantile,
  title={Information Design with Rank-Dependent Utilities},
  author={Chen, Yi-Chun and Yang, Xiangqian},
   journal={Working paper},
  year={2021}
}

@article{Park2021public,
	title={Informational Regulation of Public Good Monopolist},
	author={Park, Junrok},
	journal={Working paper},
	year={2021}
}

@article{Zhang2021bilateraltrade,
	title={Robust Bilateral Trade Mechanism with Known Expectations},
	author={Zhang, Wanchang},
	journal={Working paper},
	year={2021}
}

@article{monteiro2010optimal,
  title={Optimal Auction with a General Distribution: Virtual Valuation without Densities},
  author={Monteiro, Paulo Klinger and Svaiter, Benar Fux},
  journal={Journal of Mathematical Economics},
  volume={46},
  pages={21--31},
  year={2010},
  publisher={Elsevier}
}

@Article{du2018robust,
  author    = {Du, Songzi},
  journal   = {Econometrica},
  title     = {Robust Mechanisms under Common Valuation},
  year      = {2018},
  pages     = {1569--1588},
  volume    = {86},
  publisher = {Wiley Online Library},
}

@article{brooks2019optimal,
  title={Optimal Auction Design with Common Values: An Informationally Robust Approach},
  author={Brooks, Benjamin and Du, Songzi},
  journal={Econometrica},
  volume={89},
  number={3},
  pages={1313--1360},
  year={2021},
  publisher={Wiley Online Library}
}

@article{terstiege2020buyer,
  title={Buyer-Optimal Extensionproof Information},
  author={Terstiege, Stefan and Wasser, C{\'e}dric},
  journal={Journal of Economic Theory},
  pages={105070},
   volume = {188},
  year={2020},
  publisher={Elsevier}
}

@article{blackwell1953equivalent,
  title={Equivalent Comparisons of Experiments},
  author={Blackwell, David},
  journal={The Annals of Mathematical Statistics},
  volume={24},
  pages={265--272},
  year={1953},
  publisher={Institute of Mathematical Statistics}
}

@book{krishna2009auction,
	title={Auction Theory},
	author={Krishna, Vijay},
	year={2009},
	publisher={Academic Press}
}

@Article{roesler2017buyer,
  author  = {Roesler, Anne-Katrin and Szentes, Bal{\'a}zs},
  journal = {American Economic Review},
  title   = {Buyer-Optimal Learning and Monopoly Pricing},
  year    = {2017},
  pages   = {2072--80},
  volume  = {107},
}
	
	\bigskip \newpage
	
	\appendix
	
	\section*{Appendix}
	
	\addcontentsline{toc}{section}{Appendices} \renewcommand{\thesubsection}{%
		\Alph{subsection}}
	
	\subsection{Omitted proofs}
	
	\label{proofs}
	
	\subsubsection{Formal definition of ironed virtual valuations}
	
	\label{formal_virtual_value}
	
	For any CDF $G$ with supp$(G)\subset \lbrack 0,1]$, let $a=\inf \{x\in
	\lbrack 0,1]|G(x)>0\}$, and define 
	\begin{equation*}
		\Psi (x|G)= 
		\begin{cases}
			0, & \mbox{if }x\in \lbrack 0,a); \\ 
			a-x(1-G(x)), & \mbox{if }x\in \lbrack a,1]\text{.}%
		\end{cases}%
	\end{equation*}
	Let $\Theta =\{(\alpha ,\beta )\in \mathbb{R}^{2}|\alpha +\beta G(x)\leq
	\Psi (x|G),\forall x\in \lbrack 0,1]\}$ and let 
	\begin{equation*}
		\Phi (x|G)=\sup \{\alpha +\beta G(x)|(\alpha ,\beta )\in \Theta \},
	\end{equation*}
	where $\Phi (x|G)$ is called the convexification of $\Psi $ under under the $%
	G$-quantile space.
	
	We say that $w(x)$ is a sub-gradient of $\Phi (x|G)$ at $x\in \lbrack 0,1]$
	if 
	\begin{equation*}
		\Phi (z|G)-\Phi (x|G)\geq w(x)(G(z)-G(x)),\quad \forall z\in \lbrack 0,1].
	\end{equation*}
	For each $x\in \lbrack 0,1]$, let $\partial \Phi (x|G)$ denote the set of
	sub-gradients of $\Phi (\cdot |G)$ at $x$. Finally, let 
	\begin{equation*}
		\hat{\varphi}(x|G)=\inf \partial \Phi (x|G).
	\end{equation*}
	Then, $\hat{\varphi}(x|G)$ is defined as the \emph{ironed virtual valuation}
	induced by $G$.
	
	\subsubsection{Proof of Lemma \protect\ref{existence}}
	
	\label{proofexist}
	
	\begin{proof}
		When $\alpha =1$, we first consider the following information design 
		problem:  
		\begin{equation}
			\max_{G\in \mathcal{G}_{H}}\left( \int_{[0,1]^{n}}\max \{x_{1},\cdots
			,x_{n}\}\prod_{i=1}^{n}\left( \mathrm{d}G(x_{i})\right)
			-\int_{[0,1]^{n}}\sum_{i=1}^{n}\left( \hat{\varphi}(x_{i}|G)\right)
			q_{i}(x_{i},x_{-i}\vert G)\prod_{i=1}^{n}\left( \mathrm{d}G(x_{i})\right)
			\right) \text{.}  \label{relaxinfo}
		\end{equation}
		
		Denote by $V_{1}(G)$ and $V_{2}(G)$ the objective of problem \eqref{info} 
		and problem \eqref{relaxinfo} under the signal distribution $G$ respectively.
		
		First, since $\max \{x_{1},\cdots ,x_{n}\}$ is continuous in $x$, the first 
		term in \eqref{relaxinfo} is continuous in $G$. Moreover, by Theorem 2 of  
		\cite{monteiro2015note}, the expected revenue is a lower semicontinuous 
		function in $G$. Hence, $V_{2}(G)$ is an upper semicontinuous function in $G$%
		.  Also since $\mathcal{G}_{H}$ is a closed subset of the set of Borel 
		probability measures on $[0,1]$, $\mathcal{G}_{H}$ is compact. Thus, by the 
		extreme value theorem, an optimal solution of the problem in  %
		\eqref{relaxinfo} exists. Let $G^{\ast }$ be the optimal solution to the 
		problem \eqref{relaxinfo}.
		
		Second, for any signal distribution $G$ which induces negative virtual 
		values with positive probability, we take the same modified distribution $ 
		\tilde{G}^{\theta _{0}}\in \mathcal{G}_{H}^{+}$ as follows,  
		\begin{equation*}
			\tilde{G}^{\theta _{0}}(x)= 
			\begin{cases}
				\theta _{0}, & \text{ if }x\in \lbrack 0,x_{\theta _{0}}); \\ 
				1-\frac{x_{0}(1-G(x_{0}^{-}))}{x}, & \text{ if }x\in \lbrack x_{\theta
					_{0}},x_{0}); \\ 
				G(x), & \text{ if }x\in \lbrack x_{0},1],%
			\end{cases}%
		\end{equation*}
		where $\theta _{0}$ denotes the mass on $x=0$. Since the virtual value on $
		x_{0}$ is $0\in \lbrack \hat{\varphi}(x_{0}^{-}|G),\hat{\varphi} 
		(x_{0}^{+}|G)]$, by Lemma \ref{pre}, when $x\in \lbrack x_{\theta
			_{0}},x_{0}]$, we have $G(x)\geq \tilde{G}^{\theta _{0}}\left( x\right) $. 
		If $\theta _{0}=0$, then $\tilde{G}^{0}$ first-order stochastically 
		dominates $G(x)$, and thereby $\int_{0}^{1}x\mathrm{d}\tilde{G}^{0}\geq 
		\int_{0}^{1}x\mathrm{d}G(x)$. If $\theta _{0}=G(x_{0}^{-})$, then $x_{\theta
			_{0}}=x_{0}$. Hence, $\tilde{G}^{G(x_{0}^{-})}$ is first-order 
		stochastically dominated by $G(x)$ and $\int_{0}^{1}x\mathrm{d}\hat{G} 
		^{G(x_{0}^{-})}\leq \int_{0}^{1}x\mathrm{d}G(x)$. Since $\int_{0}^{1}x  
		\mathrm{d}\hat{G}^{\theta _{0}}(x)$ is continuous and strictly decreasing in
		$\theta _{0}$, it follows from the intermediate-value theorem that there 
		exists a unique $\theta _{0}\in \lbrack 0,G(x_{0}^{-})]$ such that $
		\int_{0}^{1}x\mathrm{d}\tilde{G}^{\theta _{0}}=p$. Moreover, by Lemma \ref%
		{pre}, since $G$ assigns positive probabilities on negative virtual values 
		for signals in $\left( 0,1\right] $, we have $\theta _{0}\in 
		(0,G(x_{0}^{-})) $. Hence, $\tilde{G}^{\theta _{0}}$ is a feasible signal 
		distribution with nonnegative virtual values except at 0; moreover, $\tilde{G%
		} ^{\theta _{0}}$ is a strict mean-preserving spread of $G$. First, since 
		the first term of the objective is convex in signal profile, the expectation
		of the first term under $\tilde{G}^{\theta _{0}}$ is strictly greater than 
		that under $G$. Meanwhile, since $\hat{\varphi}\left(x\vert \tilde{G} 
		^{\theta _{0}}\right)=0$ and the seller only allocates the good to a buyer 
		with a non-negative virtual value, the expected virtual value is the same 
		under $G$ and $\tilde{G}^{\theta _{0}}$. Therefore, $V_2\left(\tilde{G} 
		^{\theta _{0}}\right)>V_2(G)$. Hence, $G^{\ast }\in \mathcal{G}_{H}^{+}$.
		
		Third, we claim that $G^{\ast }$ also solves the problem \eqref{info}. To 
		see this, observe that for any $G$ and any signal realization $(x_{1},\cdots
		,x_{n})$, $\sum_{i=1}^{n}x_{i}q_{i}(x_{i},x_{-i}\vert G)\leq \max 
		\{x_{1},\cdots ,x_{n}\}$. If $G\in \mathcal{G}_{H}^{+}$, then for any signal
		realization $(x_{1},\cdots ,x_{n})$, $\sum_{i=1}^{n}x_{i}q_{i}(x_{i},x_{-i} 
		\vert G)=\max \{x_{1},\cdots ,x_{n}\}$. Hence,  
		\begin{equation}
			V_{1}(G)\leq V_{2}(G)\text{ for any }G\in \mathcal{G}_{H},\text{ and }
			V_{1}(G)=V_{2}(G)\text{ for any }G\in \mathcal{G}_{H}^{+}.  \label{step3}
		\end{equation}
		
		Finally,  
		\begin{equation*}
			\max_{G\in \mathcal{G}_{H}}V_{1}(G)\leq \max_{G\in \mathcal{G}
				_{H}}V_{2}(G)=V_{2}(G^{\ast })=V_{1}(G^{\ast }),
		\end{equation*}
		where the first inequality and the third equality follow from \eqref{step3} 
		and the second equality follows from the definition of $G^{\ast }$. Hence, $
		G^{\ast }$ solves the problem in \eqref{info}. Hence an optimal solution 
		exists in the problem \eqref{info}. 
	\end{proof}
	
	\subsubsection{Proof of Lemma \protect\ref{almost_nonneg}}
	
	\label{proofalmost}
	
	\begin{proof}
		\emph{The case with }$\alpha =1$: Let $G$ be a signal distribution which 
		assigns positive probability on negative virtual values for signals in $
		\left( 0,1\right] $. Let $x_{0}=\inf \{x|\hat{\varphi}(x|G)\geq 0\}$. Define
		$\tilde{G}^{\theta _{0}}(x)$ such that  
		\begin{equation*}
			\tilde{G}^{\theta _{0}}(x)= 
			\begin{cases}
				\theta _{0}, & \text{ if }x\in \lbrack 0,x_{\theta _{0}}); \\ 
				1-\frac{x_{0}(1-G(x_{0}^{-}))}{x}, & \text{ if }x\in \lbrack x_{\theta
					_{0}},x_{0}); \\ 
				G(x), & \text{ if }x\in \lbrack x_{0},1],%
			\end{cases}%
		\end{equation*}
		where $\theta _{0}$ denotes the mass on $x=0$. Since the virtual value on $
		x_{0}$ is $0\in \lbrack \hat{\varphi}(x_{0}^{-}|G),\hat{\varphi} 
		(x_{0}^{+}|G)]$, by Lemma \ref{pre}, when $x\in \lbrack x_{\theta
			_{0}},x_{0}]$, we have $G(x)\geq \tilde{G}^{\theta _{0}}\left( x\right) $. 
		If $\theta _{0}=0$, then $\tilde{G}^{0}$ first-order stochastically 
		dominates $G(x)$, and thereby $\int_{0}^{1}x\mathrm{d}\tilde{G}^{0}\geq 
		\int_{0}^{1}x\mathrm{d}G(x)$. If $\theta _{0}=G(x_{0}^{-})$, then $x_{\theta
			_{0}}=x_{0}$. Hence, $\tilde{G}^{G(x_{0}^{-})}$ is first-order 
		stochastically dominated by $G(x)$ and $\int_{0}^{1}x\mathrm{d}\hat{G} 
		^{G(x_{0}^{-})}\leq \int_{0}^{1}x\mathrm{d}G(x)$. Since $\int_{0}^{1}x  
		\mathrm{d}\hat{G}^{\theta _{0}}(x)$ is continuous and strictly decreasing in
		$\theta _{0}$, it follows from the intermediate-value theorem that there 
		exists a unique $\theta _{0}\in \lbrack 0,G(x_{0}^{-})]$ such that $
		\int_{0}^{1}x\mathrm{d}\tilde{G}^{\theta _{0}}=p$. Moreover, by Lemma \ref%
		{pre}, since $G$ assigns positive probabilities on negative virtual values 
		for signals in $\left( 0,1\right] $, we have $\theta _{0}\in 
		(0,G(x_{0}^{-})) $.
		
		Hence, $\tilde{G}^{\theta _{0}}$ is a feasible signal distribution with 
		nonnegative virtual values except at 0; moreover, $\tilde{G}^{\theta _{0}}$ 
		is a strict mean-preserving spread of $G$. Therefore, by Lemma \ref{convex},
		the total surplus under $\tilde{G}^{\theta _{0}}$ is strictly greater than 
		the total surplus under $G$. In addition, since $\hat{\varphi}\left(x\vert  
		\tilde{G}^{\theta _{0}}\right)=0$, the seller only allocates the good to a 
		buyer with a non-negative virtual value, the expected virtual value is the 
		same under $G$ and $\tilde{G}^{\theta _{0}}$. Hence, the objective value 
		will be strictly higher as the expected total surplus becomes strictly 
		higher and the seller's revenue remains the same. Hence, any optimal signal 
		distribution $G$ must induce nonnegative virtual values with probability one
		on $\left( 0,1\right] $.
		
		\emph{The case with }$\alpha =0$: Suppose that $\tilde{G}^{\theta _{0}}$ 
		puts some positive mass on $x=0$. For $\alpha =0$, we can further modify the
		distribution $\tilde{G}^{\theta _{0}}$ to reduce the seller's revenue as 
		follows. Define another signal distribution,  
		\begin{equation*}
			\hat{G}^{x_{1}}(x)= 
			\begin{cases}
				1-\frac{1-\tilde{G}^{\theta _{0}}(x_{1})}{x}, & \text{ if }x\in \lbrack 1- 
				\hat{G}^{\theta _{0}}(x_{1}),x_{1}]; \\ 
				\tilde{G}^{\theta _{0}}\left( x\right), & \text{ if }x\in (x_{1},1].%
			\end{cases}%
		\end{equation*}
		For $x_{1}=x_{0}$, we have $\hat{G}^{x_{1}}(x)=\tilde{G}^{0}\left( x\right) $
		which first-order stochastically dominates $\tilde{G}^{\theta _{0}}$, then $
		\int_{0}^{1}x\mathrm{d}\hat{G}^{x_{0}}\geq \int_{0}^{1}x\mathrm{d}\tilde{G} 
		^{\theta _{0}}$, and for $x_{1}=1$, it is a degenerate distribution with all
		mass on $x=0$. Also since $\int_{0}^{1}x\mathrm{d}\hat{G}^{x_{1}}$ is 
		continuous and strictly decreasing in $x_{1}$, the intermediate-value 
		theorem implies that there exists an unique $x_{1}\in (x_{0},1)$ such that $
		\int_{0}^{1}x\mathrm{d}\hat{G}^{x_{1}}=p$. Thus, $\hat{G}^{x_{1}}$ is a 
		feasible information structure.
		
		Since $\hat{\varphi}(x_{1}|\hat{G}^{x_{1}})=0\leq \hat{\varphi}(x_{1}|\tilde{%
			G}^{\theta _{0}})$, for each realization of signals, we have $\max \{0,  
		\hat{ \varphi}(x|\hat{G}^{x_{1}})\}\leq \max \{0,\hat{\varphi}(x|\tilde{G} 
		^{\theta _{0}})\}$. Moreover, the inequality is strict with positive 
		probability. Hence, the seller's revenue is strictly lower under $\hat{G} 
		^{x_{1}}$ than under $\tilde{G}^{\theta _{0}}$. 
	\end{proof}
	
	\subsubsection{Proof of Lemma \protect\ref{regularity}}
	
	\label{proofregular}
	
	\begin{proof}
		For any irregular distribution $G$ with nonnegative virtual values except at
		0, there exists some $x^{\prime }>0$ such that $\Psi (x^{\prime }|G)<\Phi 
		(x^{\prime }|G)$. Since $G(x)$ is right continuous with respect to $
		x^{\prime }$, so is $\Psi (x^{\prime }|G)$. Therefore, there exists an 
		interval $[x^{\prime},x ^{\prime \prime }]$, such that for any $x\in \lbrack
		x^{\prime},x ^{\prime \prime }]$, $\Psi (x|G)<\Phi (x|G)$. Let $
		[x_{1},x_{2}] \supseteq[x^{\prime},x ^{\prime \prime }]$ be an ironed 
		interval such that $\hat{\varphi}(x|G)=k$ is constant for $x\in \lbrack 
		x_{1},x_{2}]$. And for $x\in (x_{1},x_{2})$, $\Psi (x|G)\leq \Phi (x|G)$. 
		For $x=x_{1}$ and $x=x_{2}$ , $\Psi (x|G)=\Phi (x|G)$. Moreover, since $\Psi
		(x_{1}|G)=\Phi (x_{1}|G)$, $G(x)$ is continuous at $x_{1}$. Then let $\hat{G}
		$ be  
		\begin{equation*}
			\hat{G}= 
			\begin{cases}
				G(x), & \text{ if }x\not\in \lbrack x_{1},x_{2}]; \\ 
				1-\frac{(1-G(x_{1}))(x_{1}-k)}{x-k}, & \text{ if }x\in \lbrack x_{1},x_{2}].%
			\end{cases}%
		\end{equation*}
		
		The modified distribution $\hat{G}$ has two key features: firstly, it 
		generates the same virtual value as $G$ for any realized signal $x$; 
		secondly, by lemma \ref{pre}, since we have $\hat{G}(x)\leq G(x)$ and $\hat{%
			G		}(x)< G(x)$ on $x\in(x_1,x_2)$ (otherwise, $G(x)$ is regular), $\hat{G}
		$ will strictly first-order stochastically dominate $G(x)$, which implies 
		that $\hat{G}$ will generate a strictly higher mean than $p$.
		
		Hence, we can modify $\hat{G}$ again to satisfy the mean constraint. Let $ 
		\hat{G}^{ \theta_0}$ be  
		\begin{align*}
			\hat{G}^{\theta_0}= 
			\begin{cases}
				\theta_0, & \text{ if } x\in[0,x_{ \theta_0}); \\ 
				\hat{G}(x), & \text{ if } x\in[x_{ \theta_0},1],%
			\end{cases}%
		\end{align*}
		where $\theta_0=\hat{G}(x_{ \theta_0})$ denotes the mass on signal $x=0$.
		
		It is clear that  
		\begin{equation*}
			\int_{0}^{1}x\mathrm{d}\hat{G}^{0}(x)=\int_{0}^{1}x\mathrm{d}\hat{G}
			(x)>p>0=\int_{0}^{1}x\mathrm{d}\hat{G}^{1}(x).
		\end{equation*}
		Since $\int_{0}^{1}x\mathrm{d}\hat{G}^{\theta _{0}}(x)$ is continuous and 
		strictly decreasing in $\theta _{0}$, the intermediate-value theorem implies
		that there exists a unique $\theta _{0}\in (0,1)$ such that $\int_{0}^{1}x  
		\mathrm{d}\hat{G}^{\theta _{0}}=p$. Therefore, $\hat{G}^{\theta _{0}}$ is a 
		feasible signal distribution.
		
		By construction,  
		\begin{align*}
			\hat{\varphi}(x\vert \hat{G}^{ \theta_0})= 
			\begin{cases}
				-\frac{1- \theta_0}{ \theta_0}<0, & \text{ if }x\in[0,x_{ \theta_0}); \\ 
				\hat{\varphi}(x\vert \hat{G})=\hat{\varphi}(x\vert G), & \text{ if }x\in[x_{
					\theta_0},1].%
			\end{cases}%
		\end{align*}
		Hence, $\max\{\hat{\varphi}(x\vert \hat{G}^{ \theta_0}),0\}\leq \max\{\hat{
			\varphi}(x\vert G),0\}$.
		
		Hence, given any signal distribution $G$, we can modify a buyer's signal 
		distribution to $\hat{G}^{\theta _{0}}$. As the arguments in Lemma \ref%
		{almost_nonneg}, this modification has two effects: first, the seller's 
		revenue (the second term of the objective function) as a max function of 
		nonnegative virtual values becomes weakly less; second, since the 
		constructed distribution $\hat{G}^{\theta _{0}}$ is a strict mean-preserving
		spread of the original distribution $G$, by Lemma \ref{almost_nonneg} and 
		Lemma \ref{convex}, the total surplus is strictly higher. Therefore, with 
		this modification, the buyers' total surplus is strictly higher. Hence,the 
		buyer-optimal distribution must be regular except at 0.
		
		Moreover, this modification still puts a positive mass on signal $0$. By 
		Lemma \ref{almost_nonneg}, we can further modify this distribution into 
		another distribution with nonnegative virtual values and generate strictly 
		less revenue for the seller. Thus, the seller-worst distribution must also 
		be regular. 
	\end{proof}
	
	\subsubsection{Change of variable}
	
	\label{proofchange2}
	
	Denote $G_m(0)$ by $F_m(0^-)$, we have 
	\begin{align*}
		&\int_{0}^{1}x\mathrm{d}G_m^{n}(x)= \int_{0}^{1}\left( k+\frac{
			\int_{k}^{1}(1-F_m(s))\mathrm{d}s}{1-F_m(k)}\right) \mathrm{d}F_m^{n}(k) \\
		=&\int_{0}^{1}k\mathrm{d}F_m^n(k)+n\int_{0}^{1}\left( \frac{
			\int_{k}^{1}(1-F_m(s))\mathrm{d}s}{1-F_m(k)}\right) F_m^{n-1}(k)\mathrm{d}
		F_m(k) \\
		=&\int_{0}^{1}k\mathrm{d}F_m^n(k)+n\int_{0}^{1}\int_{k}^{1}(1-F_m(s))\frac{
			F_m^{n-1}(k)}{1-F_m(k)}\mathrm{d}s\mathrm{d}F_m(k) \\
		=&\int_{0}^{1}k\mathrm{d}F_m^n(k)+n\int_{0}^{1}(1-F_m(s))\int_{0}^{s}\frac{
			F_m^{n-1}(k)}{1-F_m(k)}\mathrm{d}F_m(k)\mathrm{d}s \\
		=&\int_{0}^{1}k\mathrm{d}F_m^n(k)+n\int_{0}^{1}(1-F_m(s))
		\int_{F_m(0^-)}^{F_m(s)}\frac{\theta ^{n-1}}{1-\theta }\mathrm{d}\theta 
		\mathrm{d}s \\
		=&\int_{0}^{1}n(1-F_m(k))\left( \sum_{i=1}^{n-1}\frac{-F_m^{i}(k)}{i}-\log
		(1-F_m(k))+\left( \sum_{i=1}^{n-1}\frac{F_m^{i}(0^-)}{i}+\log
		(1-F_m(0^-))\right) \right)-F_m^n(k)\mathrm{d}k+1.
	\end{align*}
	
	Then by bounded convergence theorem, since $G_{m}$ and $F_{m}$ uniformly
	converge to $G$ and $F$ respectively, Equation \eqref{m2} holds.
	
	\subsubsection{Proof of Lemma \protect\ref{stablesoln}}
	
	\label{proofstablesoln}
	
	\begin{proof}
		First, define  
		\begin{equation*}
			\zeta_{\alpha }(\theta )=n\theta
			^{n-2}\left(n-1-(n-1+\alpha)\theta\right)=(1-\theta )I_{\alpha}^{\prime
			}(\theta )-\lambda.
		\end{equation*}
		We discuss the cases with $\alpha=0$ and $\alpha=1$, respectively.
		
		\begin{enumerate}
			\item For $\alpha=0$, we have $\zeta_0(\theta)=n(n-1)(1-\theta)\theta ^{n-2}$%
			, then  
			\begin{align*}
				\frac{\partial \zeta_0(\theta)}{\partial\theta}=n(n-1)\theta^{n-3}
				\left(n-2-(n-1)\theta\right).
			\end{align*}
			Therefore, $\zeta_0(\theta)$ is increasing in $\theta$ when $\theta 
			\in(0,(n-2)/(n-1))$, and then decreasing in $\theta $ when $
			\theta\in((n-2)/(n-1),1)$. Moreover, we have $\zeta_0(0)=\zeta_{0}(1)=0$. We
			then have the following three cases:
			
			\begin{enumerate}
				\item \emph{Case 1.} $\lambda \geq 0$: In this case, $I_0^{\prime }(\theta )$
				is always positive for both $\alpha =0$. Hence, $I_0(\theta )$ is 
				increasing. Thus, $I_0(\theta )$ will cross the $\theta $-axis from below at
				most once. The objective function takes a local minimal and hence we ignore 
				this case.
				
				\item \emph{Case 2.} $\lambda \in \lbrack \lambda _{0}^{\ast },0)$ where $
				\lambda _{0}^{\ast}=-\zeta_{0}((n-2)/(n-1))<0$: $I_0^{\prime }(\theta )$ is 
				first negative, then positive, and eventually becomes negative. In addition,
				$I_0(\theta )=0$ when $\theta =\theta _{0}$. Hence $I_0(\theta )$ will first
				decrease from 0, then increase, and finally decrease. In this case, $
				I_0(\theta )$ will cross the $\theta $-axis at most twice. Only for the 
				second time, $I_0(\theta )$ will cross the $\theta $-axis from above.
				
				\item \emph{Case 3.} $\lambda <\lambda _{0}^{\ast }$: Then $I_0(\theta )$ is
				always decreasing. Thus, $I_0(\theta )$ crosses the $\theta $-axis from 
				above at most once. 
			\end{enumerate}
			
			\item For $\alpha=1 $, we have $\zeta_{1 }(\theta )=n\theta 
			^{n-2}\left(n-1-n\theta\right)$, then  
			\begin{align*}
				\frac{\partial \zeta_1(\theta)}{\partial\theta}=n\theta^{n-3}
				\left((n-2)(n-1-n
				\theta)-n\theta\right)=n(n-1)\theta^{n-3}\left((n-2)-n\theta\right).
			\end{align*}
			Therefore, $\zeta_1(\theta)$ is increasing in $\theta$ when $\theta 
			\in(0,(n-2)/n) $, and then decreasing in $\theta $ when $\theta 
			\in((n-2)/n,1) $. Moreover, we have $\zeta_1(0)=0$ and $\zeta_{1}(1)=-n$. We
			then have the following four cases:
			
			\begin{enumerate}
				\item \emph{Case 1}. $\lambda \geq n$: In this case, $I_1^{\prime }(\theta )$
				is always positive. Hence, $I_1(\theta )$ is increasing. Thus, $I_1(\theta )$
				will cross the $\theta $-axis from below at most once. The objective 
				function takes a local minimal and hence we ignore the case.
				
				\item \emph{Case 2.} $\lambda \in \lbrack 0,n)$: $I_1^{\prime }(\theta )$ is
				first positive and then negative, therefore, $I_1(\theta )$ will first 
				increase and then decrease. In this case, $I_1(\theta )$ will cross the $
				\theta $-axis at most twice. However, only for the second time, $I_1(\theta 
				) $ will cross the $\theta $-axis from above. Hence there is only one $
				\theta $ such that $I_1(\theta )=0$ and also satisfying the second-order 
				condition.
				
				\item \emph{Case 3.} $\lambda \in \lbrack \lambda _{1}^{\ast },0)$ where $
				\lambda _{1}^{\ast }=-\zeta_{1}((n-2)/n)<0$: $I_1^{\prime }(\theta )$ is 
				first negative, then positive, and eventually becomes negative. In addition,
				$I_1(\theta )=0$ when $\theta =\theta _{0}$. Hence $I_1(\theta )$ will first
				decrease from 0, then increase, and finally decrease. In this case, $
				I_1(\theta )$ will cross the $\theta $-axis at most twice. Only for the 
				second time, $I_1(\theta )$ will cross the $\theta $-axis from above.
				
				\item \emph{Case 4.} $\lambda <\lambda _{1}^{\ast }$: Then $I_1(\theta )$ 
				will be always decreasing. Thus $I_1(\theta )$ will also cross the $\theta $
				-axis from above at most once. 
			\end{enumerate}
		\end{enumerate}
	\end{proof}
	
	\subsubsection{Finite-dimensional seller-worst optimization}
	
	\label{fsw}
	
	To obtain the solution, we solve Problem \eqref{finiteinfo} with $\alpha =0$
	: 
	\begin{align*}
		& \max_{k\geq 0,\theta }-\left( \theta ^{n}\times k+(1-\theta ^{n})\times
		1\right) \\
		\text{s.t. }& k+(1-k)(1-\theta )(1-\log (1-\theta ))=p\text{. }
	\end{align*}
	Given the Lagrangian multiplier $\lambda _{s}$, the Lagrangian is 
	\begin{equation*}
		\mathcal{L}_{s}(\theta ,k,\lambda _{s})=(1-k)\theta ^{n}-\lambda
		_{s}(k+(1-k)(1-\theta )(1-\log (1-\theta ))+\lambda _{s}p-1.
	\end{equation*}
	The Euler-Lagrange condition with respect to $\theta $ implies: 
	\begin{equation*}
		\dfrac{\partial \mathcal{L}_{s}}{\partial \theta }=(1-k)\left( n\theta
		^{n-1}-\lambda _{s}\log (1-\theta )\right) =0.
	\end{equation*}
	Hence, 
	\begin{equation}
		\lambda _{s}=\dfrac{n\theta ^{n-1}}{\log (1-\theta )}.  \label{lambda}
	\end{equation}
	Recall 
	\begin{equation*}
		J_{s}(\theta )=\left( \theta \log (1-\theta )+n(\theta +(1-\theta )\log
		(1-\theta ))\right) \text{.}
	\end{equation*}
	Then, taking the derivative of $\mathcal{L}_{s}$ with respect to $k$ and
	using (\ref{lambda}), we obtain 
	\begin{align*}
		\dfrac{\partial \mathcal{L}_{s}}{\partial k}=& -\theta ^{n}-\lambda
		_{s}(1-(1-\theta )(1-\log (1-\theta ))) \\
		=& \dfrac{\theta ^{n-1}}{-\log (1-\theta )}\left( \theta \log (1-\theta
		)+n(\theta +(1-\theta )\log (1-\theta ))\right) =\dfrac{\theta ^{n-1}}{-\log
			(1-\theta )}J_{s}(\theta ).
	\end{align*}
	Hence, the $\text{sign of }\dfrac{\partial \mathcal{L}_{s}}{\partial k}\text{
		is determined by the sign of }J_{s}(\theta )$. Moreover, we have 
	\begin{align*}
		J_{s}^{\prime }(\theta )& =\log (1-\theta )-\frac{\theta }{1-\theta }
		+n-n-n\log (1-\theta )=-(n-1)\log (1-\theta )-\frac{\theta }{1-\theta }, \\
		J_{s}^{\prime \prime }(\theta )& =\frac{(n-2)-(n-1)\theta }{(1-\theta )^{2}}.
	\end{align*}
	
	\begin{itemize}
		\item When $n=2$, for $\theta>0$, we have $J_{s}^{\prime \prime }(\theta)<0$
		; therefore, $J_{s}^{\prime }(\theta )<J_{s}^{\prime }(0)=0$. Hence $%
		J_s(\theta )$ is always less than $J_{s}(0)=0$. As a result, the optimal $%
		k_{s}=0$ and we set $p_{s}=1$.
		
		\item When $n\geq 3$, $J_{s}^{\prime \prime }(\theta )=0$ has a unique
		solution $\theta _{1}=(n-2)/(n-1)$. Moreover, $J_{s}^{\prime }(\theta )$ is
		increasing in $\theta $ if $\theta \in (0,\theta _{1})$ and decreasing in $%
		\theta $ if $\theta \in (\theta _{1},1)$. Since $1/(1-\theta )$ decreases
		faster than $\log (1-\theta )$, $\lim_{\theta\uparrow1}J_{s}^{\prime
		}(\theta )\rightarrow -\infty $. Furthermore, since $J_{s}^{\prime }(0)=0$,
		there exists $\theta _{2}\in (\theta _{1},1)$ such that for $\theta \in
		(0,\theta _{2})$, $J_{s}^{\prime }(\theta )$ is greater than 0 and for $%
		\theta \in (\theta _{2},1)$, $J_{s}^{\prime }(\theta )$ is less than 0.
		Therefore, $J_{s}(\theta )$ is increasing for $\theta \in (0,\theta _{2})$
		and decreasing for $\theta \in (\theta _{2},1)$. In addition, $J_{s}(0)=0$
		and $\lim_{\theta\uparrow1}J_{s}(\theta )\rightarrow -\infty $. Therefore,
		there exists a unique $\theta _{s}\in (\theta _{2},1)$ such that $%
		J_{s}(\theta _{s})=0$; moreover, $J_{s}(\theta )>0$ if $\theta <\theta _{s}$
		and $J_{s}(\theta )<0$ if $\theta >\theta _{s}$. Recall that the threshold 
		\begin{equation*}
			p_{s}=(1-\theta _{s})(1-\log (1-\theta _{s})),
		\end{equation*}
		where $\theta _{s}$ satisfies $J_{s}(\theta _{s})=0$. If $k_{s}>0$ with the
		mass $\theta _{s}$, the mean constraint requires $p>p_{s}$. Moreover, the
		corner solution occurs when $0<p\leq p_{s}$. Notice that for an interior
		solution, $\theta _{s}$ only depends on $n$ and so does $p_{s}$.
	\end{itemize}
	
	In summary, as stated in Theorem \ref{result_seller}, the seller-worst
	information is the truncated Pareto distribution $G_{s}$ in \eqref{s1}
	parametrized by $\left( k_{s},x_{s}\right) $.
	
	\begin{itemize}
		\item When $n=2$, for any $p\in (0,1)$, we have a corner solution $k_{s}=0$.
		As a result, $x_{s}$ is pinned down by the mean constraint $x_{s}(1-\log
		(x_{s}))=p$.
		
		\item When $n\geq 3$ and $p\in \left( 0,p_{s}\right] $, we have a corner
		solution $k_{s}=0$. As a result, $x_{s}$ is pinned down by the mean
		constraint $x_{s}(1-\log (x_{s}))=p$.
		
		\item When $n\geq 3$ and $p\in \left( p_{s},1\right) $, $k_{s}$ is an
		interior solution and $\theta _{s}>0$ solves $J_{s}(\theta _{s})=0$; and $%
		k_{s}$ is pinned down by the mean constraint 
		\begin{equation*}
			(1-k_{s})(1-\theta _{s})(1-\log (1-\theta _{s}))+k_{s}=p.
		\end{equation*}
		Finally, plugging $\theta _{s}$ and $k_{s}$ into $\theta _{s}\equiv 1-\left(
		x_{s}-k_{s}\right) /\left( 1-k_{s}\right) $, we obtain $x_{s}$.
	\end{itemize}
	
	\subsubsection{Finite-dimensional buyer-optimal optimization}
	
	\label{fbo}
	
	To obtain the solution, we solve Problem \eqref{finiteinfo} with $\alpha =1 $%
	:
	
	\begin{align}
		\max_{\{\theta _{0}\geq0,k\geq0,\theta \}}& n(1-k)(1-\theta )\left(
		\sum_{i=1}^{n-1}\frac{-\theta ^{i}}{i}-\log (1-\theta )+\left(
		\sum_{i=1}^{n-1}\frac{\theta _{0}^{i}}{i}+\log (1-\theta _{0})\right) \right)
		\label{finiteinfoalpha1} \\
		\text{s.t. }& (1-k)\left( (1-\theta )(1-\log (1-\theta )+\log (1-\theta
		_{0})\right) +k(1-\theta _{0})=p.  \notag
	\end{align}
	Given the Lagrangian multiplier $\lambda_b$, the Lagrangian is 
	\begin{align*}
		\mathcal{L}_b(\theta _{0},k,\theta,\lambda_b )=& n(1-k)(1-\theta )\left(
		\sum_{i=1}^{n-1}\frac{-\theta^{i}}{i}-\log (1-\theta )+\left(
		\sum_{i=1}^{n-1}\frac{\theta _{0}^{i}}{i}+\log (1-\theta _{0})\right) \right)
		\\
		& -\lambda_b \left( (1-k)(1-\theta)(1-\log (1-\theta)+\log (1-\theta
		_{0}))+k(1-\theta _{0})\right) +\lambda_b p.
	\end{align*}
	The Euler-Lagrange condition with respect to $\theta$ implies: 
	\begin{align*}
		\dfrac{\partial \mathcal{L}_b}{\partial \theta}=& (1-k)\left( -n\left(
		\sum_{i=1}^{n-1}\frac{-\theta^{i}}{i}-\log (1-\theta )+\left(
		\sum_{i=1}^{n-1}\frac{\theta _{0}^{i}}{i}+\log (1-\theta _{0})\right)
		\right) +n(1-\theta )\left( \sum_{i=1}^{n-1}-\theta^{i-1}+\frac{1}{ 1-\theta 
		}\right) \right) \\
		& -\lambda_b (1-k)\left( -(1-\log (1-\theta)+\log (1-\theta _{0}))+1\right)
		\\
		=& (1-k)\left( -n\left( \sum_{i=1}^{n-1}\frac{-\theta^{i}}{i}-\log
		(1-\theta)+\left( \sum_{i=1}^{n-1}\frac{\theta _{0}^{i}}{i}+\log (1-\theta
		_{0})\right) \right) +n(1-\theta)\left( -\frac{1-\theta ^{n-1} }{1-\theta}+ 
		\frac{1}{1-\theta}\right) \right) \\
		& -\lambda_b (1-k)\left( -(1-\log (1-\theta )+\log (1-\theta _{0}))+1\right)
		\\
		=& (1-k)\left( n\left( \theta ^{n-1}+\left( \sum_{i=1}^{n-1}\frac{
			\theta^{i} }{i}+\log (1-\theta)\right) -\left( \sum_{i=1}^{n-1}\frac{ \theta
			_{0}^{i}}{i }+\log (1-\theta _{0})\right) \right) -\lambda (\log
		(1-\theta)-\log (1-\theta _{0}))\right) =0.
	\end{align*}
	Thus, we have 
	\begin{equation}
		\lambda_b =\frac{n\left( \theta^{n-1}+\left( \sum_{i=1}^{n-1}\frac{
				\theta^{i}}{i}+\log (1-\theta)\right) -\left( \sum_{i=1}^{n-1}\frac{ \theta
				_{0}^{i}}{i}+\log (1-\theta _{0})\right) \right) }{\log (1-\theta )-\log
			(1-\theta _{0})}.  \label{lambdat}
	\end{equation}
	
	Also, taking the derivative of $\mathcal{L}_b$ with respect to $\theta_0$,
	we have 
	\begin{align}
		\frac{\partial \mathcal{L}_b}{\partial \theta _{0}}=&(1-k)(1-\theta
		)n\left(\sum_{i}^{n-1}\theta _{0}^{i-1}-\frac{1}{1-\theta_0}
		\right)-\lambda_b \left(-\frac{(1-k)(1-\theta)}{1-\theta_0}-k\right)  \notag
		\\
		=&\frac{-(1-k)(1-\theta )n\theta_0^{n-1}}{1-\theta_0}-\lambda_b\frac{
			-(1-k)(1-\theta)-k(1-\theta_0)}{(1-\theta_0)}  \notag \\
		=&-\frac{(1-k)(1-\theta
			)n\theta_0^{n-1}-\lambda_b(1-\theta+k(\theta-\theta_0)) }{(1-\theta_0)}.
		\label{dtheta0}
	\end{align}
	Taking the derivative of $\mathcal{L}_b$ with respect to $k$, we have 
	\begin{align*}
		\frac{\partial \mathcal{L}_b}{\partial k}=& (1-\theta )\left( n\left(
		\sum_{i=1}^{n-1}\frac{\theta^{i}}{i}+\log (1-\theta)\right) -n\left(
		\sum_{i=1}^{n-1}\frac{\theta _{0}^{i}}{i}+\log (1-\theta _{0})\right) \right)
		\\
		& +\lambda_b(1-\theta)(1 - \log (1-\theta)+ \log (1-\theta _{0}))
		-\lambda_b(1-\theta _{0}).
	\end{align*}
	
	Since $k$ and $\theta _{0}$ may have corner solutions, we will study the
	corner solutions and interior solutions by cases. And we will show that $k$
	and $\theta _{0}$ cannot be both interior.
	
	\begin{itemize}
		\item If $k>0$, by Equation \eqref{lambdat}, we have 
		\begin{align*}
			&\lambda_b(1-\theta)(\log (1-\theta)+ \log (1-\theta _{0})) \\
			=&n(1-\theta)\theta^{n-1}+n(1-\theta)\left( \left( \sum_{i=1}^{n-1} \frac{
				\theta^{i}}{i}+\log (1-\theta)\right) -\left( \sum_{i=1}^{n-1} \frac{\theta
				_{0}^{i}}{i}+\log (1-\theta _{0})\right) \right).
		\end{align*}
		Therefore, 
		\begin{align*}
			\frac{\partial \mathcal{L}_b}{\partial k}=&\lambda_b(1-\theta)(\log
			(1-\theta)+ \log (1-\theta _{0}))-n(1-\theta)\theta
			^{n-1}+\lambda_b(1-\theta)(1 - \log (1-\theta)+ \log (1-\theta _{0}))
			-\lambda_b (1-\theta _{0}) \\
			=&\lambda_b(1-\theta)-n(1-\theta)\theta^{n-1}-\lambda_b (1-\theta _{0})=0.
		\end{align*}
		Hence, we have 
		\begin{align*}
			\lambda_b=\frac{-n(1-\theta)\theta^{n-1}}{\theta-\theta_0}.
		\end{align*}
		Since $\theta_b \geq \theta _{0}$, we have 
		\begin{align*}
			\frac{\partial \mathcal{L}_b}{\partial \theta _{0}}=& \dfrac{ (1-k)(1-\theta
				)n\theta _{0}^{n-1}+\lambda_b ((1-\theta)+k(\theta -\theta _{0}))}{
				-(1-\theta _{0})} \\
			=& \frac{n(1-\theta)}{-(1-\theta _{0})}\left( (1-k)\theta _{0}^{n-1}+ \frac{
				\theta^{n-1}(1-\theta)}{\theta-\theta _{0}}+k\theta^{n-1}\right) \\
			=& \frac{n(1-\theta)}{-(1-\theta _{0})}\left( \theta _{0}^{n-1}+\frac{
				\theta^{n-1}(1-\theta)}{\theta-\theta _{0}}+k\left( \theta ^{n-1}-\theta
			_{0}^{n-1}\right) \right)<0.
		\end{align*}
		Therefore, as long as $k>0$, we have $\frac{\partial \mathcal{L}_b}{
			\partial \theta _{0}} <0$, and thereby $\theta _{0}$ will go down to $0$ and
		be a \emph{corner} solution.
		
		We then pin down the support of $p$ such that the optimal $k>0$. Given $%
		\theta _{0}=0$, let 
		\begin{align}
			J_{b}(\theta )\equiv & \log (1-\theta )+\theta ^{n-2}\left( \theta
			+(1-\theta )\log (1-\theta )\right) +\sum_{i=1}^{n-1}\frac{\theta ^{i}}{i} 
			\notag \\
			=& \theta ^{n-1}+\sum_{i=1}^{n-1}\frac{\theta ^{i}}{i}+\log (1-\theta
			)(1+(1-\theta )\theta ^{n-2})  \label{foc_kbapp}
		\end{align}%
		We then have 
		\begin{equation*}
			\frac{\partial \mathcal{L}_{b}}{\partial k}=-\frac{n\theta \left( \theta
				^{n-1}+\left( \sum_{i=1}^{n-1}\frac{\theta ^{i}}{i}+\log (1-\theta )\right)
				\right) }{\log (1-\theta )}-n(1-\theta )\theta ^{n-1}=\dfrac{n\theta }{-\log
				(1-\theta )}J_{b}(\theta ).
		\end{equation*}
		
		Hence, the $\text{sign of }\dfrac{\partial \mathcal{L}_{b}}{\partial k}\text{
			is determined by the sign of }J_{b}(\theta)$.\footnote{%
			We present instead $\theta J_{b}(\theta )$ in the cost-benefit equation %
			\eqref{foc_kb} for the ease of comparison with $J_{s}\left( \theta \right) $.%
		} Moreover, we have 
		\begin{align*}
			& J^{\prime }_b(\theta)=\frac{\theta^{n-3}\left(\theta+(1-\theta)\log(1-
				\theta)\right)}{1-\theta}\left((n-2)-(n-1)\theta\right), \\
			\Longrightarrow & \text{Sign }J^{\prime }_b(\theta)= \text{Sign }
			\left((n-2)-(n-1)\theta\right).
		\end{align*}
		
		As in the seller-worst case,
		
		\begin{itemize}
			\item When $n=2$, for any $\theta >0$, $J^{\prime }_b(\theta)<0$ and $%
			J_b(\theta)<J_b(0)=0$. As a result, the optimal $k_b=0$ when $n=2$. In this
			case, we set $p_b=1$.
			
			\item When $n\geq3$, for $\theta\in(0,(n-2)/(n-1))$, $J^{\prime
			}_b(\theta)>0 $ and for $\theta\in((n-2)/(n-1),1]$, $J^{\prime
			}_b(\theta)<0. $ In addition that $J_b(0)=0$ and $\lim\limits_{\theta
				\uparrow1}J_b(\theta)\rightarrow-\infty$, there exists a unique $\theta_{b}
			\in (0,1)$ such that $J_b(\theta _{b} )=0$; moreover, $J_b(\theta)>0$ if $%
			\theta<\theta_{b} $ and $J_b(\theta)<0$ if $\theta>\theta_{b} $.
			
			Recall that the threshold 
			\begin{equation*}
				p_{b}=(1-\theta _{b})(1-\log (1-\theta _{b})),
			\end{equation*}
			where $\theta _{b}$ satisfies $J_{b}(\theta _{b})=0$. If $k_b>0$ with the
			mass $\theta_b$, the mean constraint requires $p> p_b$. Moreover, the corner
			solution $k_b=0$ occurs when $0<p\leq p_s$.
		\end{itemize}
		
		\item If $k_b=0$, we have $p\in\left( 0,p_b\right]$. We then only need to
		solve $\theta_0$ and $\theta$. Moreover, if the optimal $\theta_0$ is
		chosen, $\theta$ will be automatically pinned down by the mean constraint.
		We also have two cases for the optimal $\theta_0$: $\theta_0>0$ and $%
		\theta_0=0$.
		
		Suppose that $\theta_0>0$, then the Euler-Lagrange condition for $\theta_0 $
		in Equation \eqref{dtheta0} implies 
		\begin{align}
			\frac{\partial \mathcal{L}_b}{\partial \theta _{0}} =&-\frac{
				(1-k_b)(1-\theta )n\theta_0^{n-1}-\lambda_b(1-\theta+k(\theta-\theta_0))}{
				(1-\theta_0)}  \notag \\
			=&\underbrace{-\frac{(1-\theta)n\theta_0^{n-1}}{(1-\theta_0)}}_{\text{cost}%
				(<0)}+\underbrace{\frac{\lambda_b(1-\theta)}{(1-\theta_0)}}_{\text{benefit}%
				(>0)}=0.  \label{cost-benefittheta0}
		\end{align}
		
		Raising the mass $\theta _{0}$ at signal $0$ has two countervailing effects
		on the objective in \eqref{finiteinfoalpha1}. First, by increasing $\theta
		_{0}$, $\sum_{i=1}^{n-1}\frac{\theta _{0}^{i}}{i}+\log (1-\theta _{0})$ in %
		\eqref{finiteinfoalpha1} decreases, which translates into a cost in
		proportion to the first term in \eqref{cost-benefittheta0}. Second, to obey
		the mean constraint, the probability $\theta $ is reduced, thus $%
		\sum_{i=1}^{n-1}\frac{-\theta ^{i}}{i}-\log (1-\theta )$ in %
		\eqref{finiteinfoalpha1} increases, which results in a benefit in proportion
		to the second term in \eqref{cost-benefittheta0}.
		
		Together with Equation \eqref{lambdat}, we have 
		\begin{align}  \label{lambda_app}
			\lambda_b= n\theta_0^{n-1}=\frac{n\left( \theta^{n-1}+\left(
				\sum_{i=1}^{n-1} \frac{\theta^{i}}{i}+\log (1-\theta)\right) -\left(
				\sum_{i=1}^{n-1}\frac{ \theta _{0}^{i}}{i}+\log (1-\theta _{0})\right)
				\right) }{\log (1-\theta)-\log (1-\theta _{0})}.
		\end{align}
		We have if $\lambda_b>0$, $\theta_0$ is positive; if $\lambda_b\leq0$, $%
		\theta_0=0$. In order that there exists a positive solution $%
		(\theta,\theta_0)$, we also have some restrictions on the mean $p$.
		
		More precisely, there exists another threshold,\footnote{%
			Note that $\theta_{r_b}>\theta_{b}$ and thereby $r_b<p_b$. This is because
			the factor $(1-\theta)\theta^{n-2}$ in Expression \eqref{foc_kbapp} makes $%
			J_b(\theta)$ arrives $0$ faster; and the function $(1-\theta )(1-\log
			(1-\theta ))$ is decreasing in $\theta$.} 
		\begin{align*}
			r_b=(1-\theta_{r_b})(1- \log(1-\theta_{r_b})),
		\end{align*}
		where $\theta_{r_b}$ satisfies Equation \eqref{lambda_app} with $\lambda_b=0$
		and $\theta_0=0$, that is, 
		\begin{align*}
			\theta_{r_b} ^{n-1}+\left( \sum_{i=1}^{n-1}\frac{ \theta_{r_b} ^{i}}{i}+\log
			(1-\theta_{r_b} )\right)=0.
		\end{align*}
		
		When $p\in\left(0,r_b\right)$, $\lambda_b$ will be positive. In this case, $%
		\theta_0$ will be also positive. When $p\in[r_b,p_b]$, $\frac{ \partial 
			\mathcal{L}_b}{\partial \theta_{0}}$ is always less than $0$, $\theta_0$
		will go down to 0 and both $k_b$ and $\theta_0$ will be corner solutions.
	\end{itemize}
	
	In summary, as stated in Theorem \ref{result_buyer}, the buyer optimal
	information ($\alpha =1$) is the truncated Pareto distribution $G_{b}$ in %
	\eqref{b1} parameterized by $\left(\theta_0,k_b,\theta_b\right)$.
	
	\begin{itemize}
		\item When $n=2$, for any $p\in(0,1)$, the optimal $k_b=0$ is a corner
		solution.
		
		\begin{itemize}
			\item When $p\in\left(0,r_b \right)$, $\theta_0 $ is an interior solution. $%
			\left(\theta_0,\theta_b\right)$ is jointly determined by the Euler-Lagrange
			condition and the mean constraint, 
			\begin{align*}
				& \theta_0 =\frac{\left( 2\theta_b +\log (1-\theta_b )\right) -\left( \theta
					_{0} +\log (1-\theta _{0})\right) }{\log (1-\theta_b )-\log (1-\theta _{0})},
				\\
				& \left( 1-\theta_b\right) \left( 1-\log \left( 1-\theta _{b}\right) +\log
				(1-\theta _{0})\right) =p.
			\end{align*}
			
			\item When $p\in\left[r_b,1\right)$, both $k_b=0$ and $\theta_0=0$ are
			corner solutions. The only parameter $x_{b}$ is then pinned down by the mean
			constraint $x_{b}(1-\log (x_{b}))=p.$
		\end{itemize}
		
		\item When $n=3$, there are two thresholds $r_b$ and $p_b$.
		
		\begin{itemize}
			\item When $p\in\left(p_b,1\right)$, $k_{b} $ is an interior solution and $%
			\theta_0=0$ is a corner solution. $\theta_b>0$ solves $J_b(\theta_b)=0$; and 
			$k_b$ is pinned down by the mean constraint 
			\begin{equation*}
				(1-k_{b})(1-\theta_b)(1-\log (1-\theta_b))+k_{b}=p.
			\end{equation*}
			
			\item When $p\in\left[r_b,p_b\right]$, both $k_b=0$ and $\theta_0=0$ are
			corner solutions. The only parameter $x_{b}$ is then pinned down by the mean
			constraint $x_{b}(1-\log (x_{b}))=p.$
			
			\item When $p\in\left(0,r_b \right)$, $k_b=0$ is a corner solution and $%
			\theta_0 $ is an interior solution. $\left(\theta_0,\theta_b\right)$ is
			jointly determined by the Euler-Lagrange condition and the mean constraint, 
			\begin{align*}
				& \theta_0^{n-1}=\frac{\left( \theta_b ^{n-1}+\left( \sum_{i=1}^{n-1} \frac{
						\theta_b ^{i}}{i}+\log (1-\theta_b )\right) -\left( \sum_{i=1}^{n-1}\frac{
						\theta _{0}^{i}}{i}+\log (1-\theta _{0})\right) \right) }{\log (1-\theta_b
					)-\log (1-\theta _{0})}, \\
				& \left( 1-\theta_b\right) \left( 1-\log \left( 1-\theta _{b}\right) +\log
				(1-\theta _{0})\right) =p.
			\end{align*}
		\end{itemize}
	\end{itemize}
	
	Lastly, we show the following lemma that the benefit in \eqref{foc_ks} is greater than the
benefit in \eqref{foc_kb}. Hence, the buyer-optimal information designer is
more reluctant to raise the low virtual value than a seller-worst
information designer.

\begin{claim}\label{bsgeqbb}
	For any $n\geq3$, the benefit in \eqref{foc_ks} is greater than the
	benefit in \eqref{foc_kb}.
\end{claim}

	\begin{proof}
		Define the difference of the benefit in \eqref{foc_ks} and the benefit in %
		\eqref{foc_kb} by $\triangle J_{n}(\theta )\times \theta $. Hence, we have 
		\begin{equation*}
			\triangle J_{n}(\theta )=\left( n-\theta ^{n-1}\right) \left( 1+\frac{%
				(1-\theta )\log (1-\theta )}{\theta }\right) -\sum_{i=1}^{n-1}\frac{\theta
				^{i}}{i}.
		\end{equation*}%
		We then define 
		\begin{equation*}
			\triangle J_{n2}(\theta )=\left( n-\theta ^{2-1}\right) \left( 1+\frac{%
				(1-\theta )\log (1-\theta )}{\theta }\right) -\sum_{i=1}^{n-1}\frac{\theta
				^{i}}{i}\text{.}
		\end{equation*}%
		We have $\triangle J_{n}(\theta )\geq \triangle J_{n2}(\theta )$ for any $%
		n\geq 3$. Then, taking the derivative of $\triangle J_{n2}$ with respect to $%
		\theta $, we obtain 
		\begin{equation*}
			\frac{\partial \triangle J_{n2}}{\partial \theta }=\frac{n}{\theta ^{2}}%
			\left( -\log (1-\theta )-\theta \right) +\log (1-\theta )-\frac{1-\theta
				^{n-1}}{1-\theta }.
		\end{equation*}%
		Moreover, 
		\begin{align*}
			\frac{\partial ^{2}\triangle J_{n2}}{\partial \theta \partial n}=& \frac{%
				-\log (1-\theta )-\theta }{\theta ^{2}}+\frac{\theta ^{n-1}\log (\theta )}{%
				1-\theta } \\
			\geq & \left. \frac{\partial ^{2}\triangle J_{n2}}{\partial \theta \partial n%
			}\right\vert _{n=3}(\theta )>0,\quad \forall \theta \in (0,1).
		\end{align*}
		
		That is, $\frac{\partial \triangle J_{n2}}{\partial \theta }$ is increasing
		in $n$. Hence, to show that $\frac{\partial \triangle J_{n2}}{\partial
			\theta }>0$ for any $n\geq 3$, it suffices to show $\frac{\partial \triangle
			J_{n2}}{\partial \theta }>0$ when $n=3$. First, we have 
		\begin{equation*}
			\left. \frac{\partial ^{2}\triangle J_{n2}}{\partial \theta ^{2}}\right\vert
			_{n=3}=\frac{6\log (1-\theta )-\frac{(\theta -2)\theta \left( \theta
					^{2}-3\right) }{\theta -1}}{\theta ^{3}}.
		\end{equation*}%
		Then, we have 
		\begin{equation*}
			\frac{\partial }{\partial \theta }\left( 6\log (1-\theta )-\frac{(\theta
				-2)\theta \left( \theta ^{2}-3\right) }{\theta -1}\right) =\frac{\theta
				^{2}(-3\theta ^{2}+8\theta -3)}{(1-\theta )^{2}}.
		\end{equation*}%
		For $\theta \in (0,\frac{1}{3}\left( 4-\sqrt{7}\right) )$, we have $-3\theta
		^{2}+8\theta -3<0$ and for $\theta \in (\frac{1}{3}\left( 4-\sqrt{7}\right)
		,1)$, we have $-3\theta ^{2}+8\theta -3>0$.
		
		Therefore, $\left. \frac{\partial ^{2}\triangle J_{n2}}{\partial \theta ^{2}}%
		\right\vert _{n=3}$ is first decreasing and then increasing. Since $%
		\lim_{\theta \downarrow 0}\left. \frac{\partial ^{2}\triangle J_{n2}}{%
			\partial \theta ^{2}}\right\vert _{n=3}(\theta )=-1$ and $\lim_{\theta
			\uparrow 1}\left. \frac{\partial ^{2}\triangle J_{n2}}{\partial \theta ^{2}}%
		\right\vert _{n=3}(\theta )\rightarrow \infty $, $\left. \frac{\partial
			\triangle J_{n2}}{\partial \theta }\right\vert _{n=3}$ is also decreasing
		first and then increasing. Therefore, $\left. \frac{\partial \triangle J_{n2}%
		}{\partial \theta }\right\vert _{n=3}$ has a unique minimum over $\theta \in
		(0,1)$ which is strictly positive. Hence, when $n=3$, $\frac{\partial
			\triangle J_{n2}}{\partial \theta }>0$ for any $\theta $. Thus, $\triangle
		J_{n}(\theta )>\triangle J_{n2}(\theta )\geq \triangle J_{n2}(0)=0.$
		Therefore, the benefit in \eqref{foc_ks} is greater than the benefit in %
		\eqref{foc_kb}.
	\end{proof}
	
	\subsubsection{Proof of Lemma \protect\ref{nonneg}}
	
	\label{proofnonneg}
	
	\begin{proof}[Proof of nonnegativity]
		For each buyer $i$, let $G_{i}$ be a signal distribution which assigns 
		positive probability on negative virtual values for signals in $\left[ 0,1  %
		\right] $ and let $x_{i0}=\inf \{x|\hat{\varphi}(x|G_{i})\geq 0\}$.
		
		First, we construct a distribution $\tilde{G}_{i}^{\theta _{i0}}$ similar to
		that in Appendix \ref{proofalmost}:  
		\begin{equation*}
			\tilde{G}_{i}^{\theta _{i0}}(x)= 
			\begin{cases}
				\theta _{i0}, & \text{ if }x\in \lbrack 0,x_{\theta _{i0}}); \\ 
				1-\frac{x_{0}(1-G_i(x_{0}^{-}))}{x}, & \text{ if }x\in \lbrack x_{\theta
					_{i0}},x_{i0}); \\ 
				G_i(x), & \text{ if }x\in \lbrack x_{i0},1],%
			\end{cases}%
		\end{equation*}
		By similar arguments, $\{\tilde{G}_{i}^{\theta _{i0}}\}_{i}$ generates 
		weakly less revenue than $\{G_{i}\}$.
		
		Second, define another signal distribution $\hat{G}_{i}^{x_{1}}$ similar to 
		that in Appendix \ref{proofalmost}:  
		\begin{equation*}
			\hat{G}_{i}^{x_{i1}}(x)= 
			\begin{cases}
				1-\frac{1-\tilde{G}_i^{\theta _{i0}}(x_{i1})}{x}, & \text{ if }x\in \lbrack
				1- \hat{G}_{i}^{\theta _{i0}}(x_{i1}),x_{i1}]; \\ 
				\tilde{G}_{i}^{\theta _{i0}}\left( x\right), & \text{ if }x\in (x_{i1},1].%
			\end{cases}%
		\end{equation*}
		By similar arguments, $\{\hat{G}_{i}^{x_{i1}}\}_{i}$ generates strictly less
		revenue than $\{\tilde{G}_{i}^{\theta _{i0}}\}_{i}$. 
	\end{proof}
	
	\begin{proof}[Proof of regularity]
		For each buyer $i$, let $G_{i}$ induce ironed virtual value on $
		[x_{i1},x_{i2}]$. First, we construct a distribution $\hat{G}_{i}^{\theta
			_{i0}}$ similar to that in Appendix \ref{proofregular}:  
		\begin{equation*}
			\hat{G}_{i}^{\theta _{i0}}\left( x\right) = 
			\begin{cases}
				\theta _{i0}, & \text{ if }x\in \lbrack 0,x_{\theta _{i0}}), \\ 
				\hat{G}_{i}(x), & \text{ if }x\in \lbrack x_{\theta _{i0}},1],%
			\end{cases}%
		\end{equation*}
		where  
		\begin{equation*}
			\hat{G}_{i}(x)= 
			\begin{cases}
				G_{i}(x), & \text{ if }x\not\in \lbrack x_{i1},x_{i2}]; \\ 
				1-\frac{(1-G_i(x_{i1}))(x_{i1}-k)}{x-k}, & \text{ if }x\in \lbrack
				x_{i1},x_{i2}].%
			\end{cases}%
		\end{equation*}
		By similar arguments, $\{\hat{G}_{i}^{\theta _{i0}}\}_{i}$ generates weakly 
		less revenue than $\{G_{i}\}$.
		
		Although this modification still puts a positive mass on signal $0$, by the 
		nonnegativity part of Lemma \ref{nonneg}, we can further modify this 
		distribution into the one with nonnegative virtual values and thereby 
		achieve strictly less seller revenue. Thus, the seller-worst distribution 
		must also be regular. 
	\end{proof}
	
	\subsubsection{Proof of Proposition \protect\ref{asy_buyer}}
	
	\label{proofasybuyer}
	
	\begin{proof}[Proof for $n=2$]
		Within the class documented in the context, let $\theta _{0i}$ denote the 
		mass on signal $0$ and $1-\theta _{i}$ denote the mass on signal $1$ for 
		buyer $i$. Without loss of generality, we assume that $\theta _{01}\leq 
		\theta _{02}$, and then the buyer-optimal information design problem is  
		\begin{align*}
			\max_{\theta _{0i},\theta _{i}}\quad & \int_{0}^{1}x\mathrm{d}
			(G_{1}G_{2})-(1-\theta _{1}\theta _{2}) \\
			\text{s.t. }\quad & (1-\theta _{i})\left( 1-\log (1-\theta _{i})+\log
			(1-\theta _{0i})\right) =p,\quad \forall i=1,2.
		\end{align*}
		And the objective can be rewritten as  
		\begin{align*}
			& 1-\int_{0}^{1}G_{1}G_{2}\mathrm{d}x-(1-\theta _{1}\theta _{2}) \\
			=& 1-\int_{0}^{x_{1}}\theta _{01}\theta _{02}\mathrm{d}x-
			\int_{x_{1}}^{x_{2}}\theta _{02}\left( 1-\frac{a_{1}}{x}\right) \mathrm{d}
			x-\int_{x_{2}}^{1}\left( 1-\frac{a_{1}}{x}\right) \left( 1-\frac{a_{2}}{x}
			\right) \mathrm{d}x-(1-\theta _{1}\theta _{2}) \\
			=& 1-\theta _{01}\theta _{02}x_{1}-\theta _{02}(x_{2}-x_{1}-a_{1}(\log
			(x_{2})-\log (x_{1})))-(1-x_{2}) \\
			& -(a_{1}+a_{2})\log (x_{2})-a_{1}a_{2}\left( 1-\frac{1}{x_{2}}\right)
			-(1-\theta _{1}\theta _{2}) \\
			=& 2(1-\theta _{1})(\theta _{02}-\theta _{2})+\theta _{02}(1-\theta
			_{1})\left( \log (1-\theta _{01})-\log (1-\theta _{1})\right) \\
			& +(2-(1-\theta _{1})\theta _{02}-\theta _{1}-\theta _{2})\left( \log
			(1-\theta _{02})-\log (1-\theta _{2})\right) .
		\end{align*}
		
		Therefore, the Lagrangian is  
		\begin{align*}
			\mathcal{L}=& 2(1-\theta _{1})(\theta _{02}-\theta _{2})+\theta
			_{02}(1-\theta _{1})\left( \log (1-\theta _{01})-\log (1-\theta _{1}))\right)
			\\
			& +(2-(1-\theta _{1})\theta _{02}-\theta _{1}-\theta _{2})\left( \log
			(1-\theta _{02})-\log (1-\theta _{2}))\right) \\
			& -\sum_{i=1}^{2}\lambda _{i}\left( (1-\theta _{i})\left( 1-\log (1-\theta
			_{i})+\log (1-\theta _{0i})\right) -p\right) .
		\end{align*}
		
		Taking the first derivative with respect to $\theta _{01}$ yields  
		\begin{equation*}
			\frac{\partial \mathcal{L}}{\partial \theta _{01}}=\frac{-\theta
				_{02}(1-\theta _{1})}{1-\theta _{01}}+\frac{\lambda _{1}(1-\theta _{1})}{
				1-\theta _{01}}=\frac{(\lambda _{1}-\theta _{02})(1-\theta _{1})}{1-\theta
				_{01}}.
		\end{equation*}
		Since $\lambda _{1}$ is constant, $\frac{\partial \mathcal{L}}{\partial
			\theta _{01}}$ is either always non-positive or always nonnegative. 
		Therefore, the optimal $\theta _{01}$ is a boundary solution. That is $
		\theta _{01}=\theta _{02}$ or $\theta _{01}=0$.
		
		First, if $\theta _{01}=\theta _{02}$, then the solution becomes the 
		symmetric buyer-optimal information structure design problem when $n=2$.
		
		Second, if $\theta _{01}=0$, then the information designer only needs to 
		choose the optimal $\theta _{02}$ to maximize the buyers' surplus. Since $
		p<r_{b}$, the symmetric buyer-optimal information structure puts a positive 
		mass on signal $0$. Hence, $\theta _{02}>0$; otherwise, by the mean 
		constraint, both buyers distributions become identical and place mass only 
		on virtual values $0$ and $1$, which, by Theorem \ref{result_buyer}, is not 
		optimal.
		
		However, to guarantee that the asymmetric case with $\theta _{01}=0$ and $
		\theta _{02}>0$ is not vacuous, we have to know the exact value of $\lambda 
		_{1}$ and $\theta _{02}$. We appeal to simulation. For instance, let $p=0.4$%
		, under the symmetric buyer-optimal information, the mass on signal $0$ is $
		\theta _{0}\simeq 0.1251$. Then, the mean constraint implies $\theta_b 
		\simeq 0.8581$. Hence, the buyers' surplus is $0.3082$. If we take $\theta 
		_{01}=0$ and $\theta _{02}=0.3$, then the mean constraint implies that $
		\theta _{1}\simeq 0.8677$ and $\theta _{2}\simeq 0.8374$. The buyers' 
		surplus is $0.3107>0.3082$. Moreover, under the resulting asymmetric 
		information structure, buyer 1's surplus is $0.1443<0.1541=0.3082/2$, 
		whereas buyer 2's surplus is $0.1663>0.1541$. Hence, buyer 2 benefits from 
		the asymmetric information structure more than the loss incurred by buyer 1. 
	\end{proof}
	
	\begin{proof}[Proof for $n\rightarrow \infty $]
		Let $\theta _{0}$ denote the mass on signal $0$ and $1-\theta $ denote the 
		mass on signal $1$, let buyer $i$'s signal distribution be  
		\begin{equation*}
			G_{i}(x)= 
			\begin{cases}
				\theta _{0}, & \mbox{if }x\in \lbrack 0,x_{1}); \\ 
				1-\frac{(1-\theta )(1-p)}{x-p}, & \mbox{if }x\in \lbrack x_{1},1); \\ 
				1, & \mbox{if }x=1,%
			\end{cases}%
		\end{equation*}
		where $x_{1}=p+(1-\theta )(1-p)/(1-\theta _{0})$, and all the other buyers' 
		signal distributions be the degenerate distribution as in Corollary \ref%
		{equivalence}. First, for $x\in \lbrack x_{1},1)$, the induced virtual value
		is $p$. And since $x_{1}>p$, by the allocation rule, the good will be 
		allocated to buyer $i$ if buyer $i$'s signal belongs to $[x_{1},1)$. Since 
		the seller can only get $p$ which is strictly less than $x_{1}$ when $
		x_{i}\in \lbrack x_{1},1)$, buyer $i$ must obtain some positive information 
		rents. Hence, the buyers' total surplus is strictly positive.\footnote{
			For example, let $p=0.4$, take $\theta _{0}=0.4751$ and $\theta =0.8661$ 
			which satisfies the mean constraint. Then, the buyers' surplus is $0.1097$ 
			which is strictly larger than 0 under the symmetric buyer-optimal 
			information structure when $n\rightarrow \infty $.} 
	\end{proof}
	
	\subsubsection{Proof of Proposition \protect\ref{ap}}
	
	\label{proofap}
	
	\begin{proof}
		With change of variable, for $n=2$, the seller-worst information designer's 
		problem can be written as  
		\begin{align}
			\max_{\{F_{i}(k)\}_{i=1}^{2}}& \int_{0}^{1}\prod_{i=1}^{2}F_{i}(k)\mathrm{d}
			k-1  \label{prob-asy} \\
			\text{s.t. }& \int_{0}^{1}(1-F_{i}(k))(1-\log (1-F_{i}(k)))\mathrm{d}
			k=p_{i},\quad \forall i=1,2.  \notag
		\end{align}
		Consider the following Lagrangian $\mathcal{L}$:  
		\begin{equation*}
			\mathcal{L}(F_{i}(k),\lambda
			_{i})=\int_{0}^{1}\prod_{i=1}^{2}F_{i}(k)-\sum_{i=1}^{2}\lambda _{i}\left(
			(1-F_{i}(k))(1-\log (1-F_{i}(k)))\right) \mathrm{d}k+\sum_{i=1}^{2}\lambda
			_{i}p_{i}\text{.}
		\end{equation*}
		By Theorem 4.2.1 of \cite{van2004isoperimetric}, for any state $k$, the 
		Euler-Lagrange equation with respect to $F_{i}(k)$ is  
		\begin{equation*}
			\prod_{j\neq i}F_{j}(k)-\lambda _{i}\log (1-F_{i}(k))=0,\quad \forall i=1,2.
		\end{equation*}
		Denote $\theta _{i}=F_{i}(k)$ and we have  
		\begin{align*}
			\theta _{2}& =\lambda _{1}\log (1-\theta _{1}), \\
			\theta _{1}& =\lambda _{2}\log (1-\theta _{2}).
		\end{align*}
		
		First, note that $\lambda _{i}$ should be negative. Then, we claim that 
		there is a unique solution pair $(\theta _{1}^{\ast },\theta _{2}^{\ast })$ 
		such that the Euler-Lagrange equations are satisfied. Hence, given any 
		solution $(\theta _{i},\theta _{j})$, we have  
		\begin{equation}
			\theta _{1}=\lambda _{2}\log (1-\lambda _{1}\log (1-\theta _{1})).
			\label{fixed}
		\end{equation}
		Taking the first and second derivative of the right-hand side, we obtain  
		\begin{align*}
			\frac{\partial }{\partial \theta _{1}}\lambda _{2}\log (1-\lambda _{1}\log
			(1-\theta _{1}))& =\frac{\lambda _{1}\lambda _{2}}{(1-\theta _{1})(1-\lambda
				_{1}\log (1-\theta _{1}))}>0, \\
			\frac{\partial ^{2}}{\partial \theta _{1}^{2}}\lambda _{2}\log (1-\lambda
			_{1}\log (1-\theta _{1}))& =\frac{\lambda _{1}\lambda _{2}(1-\lambda
				_{1}\log (1-\theta _{1})-\lambda _{1})}{(1-\theta _{1})^{2}(1-\lambda
				_{1}\log (1-\theta _{1}))^{2}}>0.
		\end{align*}
		Hence, the right-hand side is convex and increasing in $\theta _{1}$. It 
		follows that there will be only one solution of $\theta _{1}$ such that 
		Equation \eqref{fixed} holds.
		
		Hence, an optimal $F_{1}^{\ast }(k)$ and $F_{2}^{\ast}(k)$ are both 
		constantly equal to $\theta _{1}$ and $\theta _{2}$ for $k<1$ That is, both $
		F_{1}^{\ast }(k)$ and $F_{2}^{\ast }(k)$ have binary support $\{k,1\}$. By 
		part (1) of Lemma \ref{existence2}, the uniqueness of $\theta _{i}$ implies $
		\left( F_{1}^{\ast },F_{2}^{\ast}\right) $ with the binary support is also a
		global maximizer.
		
		Again, let $\theta _{i}$ be the mass on the virtual value $k$ for buyer $i$.
		Then the information design problem is reduced into  
		\begin{align*}
			\max_{k,\theta _{1},\theta _{2}}& (1-k)\theta _{1}\theta _{2}-1 \\
			\text{s.t. }& k+(1-k)\left( (1-\theta _{i})(1-\log (1-\theta _{i}))\right)
			=p_{i}.
		\end{align*}
		The Lagrangian with multiplier $\lambda _{i}$ is  
		\begin{equation*}
			\mathcal{L}(k,\theta_1,\theta_2,\lambda_1,\lambda_2)=(1-k)\theta _{1}\theta
			_{2}-1-\sum_{i=1}^{2}\lambda _{i}\left( k+(1-k)\left( (1-\theta _{i})(1-\log
			(1-\theta _{i}))\right) -p_{i}\right).
		\end{equation*}
		The Euler-Lagrange equation with respect to $\theta _{i}$ is  
		\begin{equation}  \label{prop2theta}
			\frac{\partial \mathcal{L}}{\partial \theta _{i}}=(1-k)\left( \theta
			_{-i}-\lambda _{i}\log (1-\theta _{i})\right) =0\Longrightarrow \lambda
			_{i}= \frac{\theta _{-i}}{\log (1-\theta _{i})}.
		\end{equation}
		Also the Euler-Lagrange equation with respect to $k$ is  
		\begin{equation}  \label{prop2k}
			\frac{\partial \mathcal{L}}{\partial k}=-\theta _{1}\theta
			_{2}-\sum_{i=1}^{2}\lambda _{i}\left( \theta _{i}+(1-\theta _{i})\log
			(1-\theta _{i})\right) \text{.}
		\end{equation}
		Plugging Equation \eqref{prop2theta} into Expression \eqref{prop2k}, we have
		\begin{equation*}
			\frac{\partial \mathcal{L}}{\partial k}=\theta _{1}\theta _{2}-\theta
			_{1}\left( 1+\frac{\theta _{2}}{\log (1-\theta _{2})}\right) -\theta
			_{2}\left( 1+\frac{\theta _{1}}{\log (1-\theta _{1})}\right) .
		\end{equation*}
		We claim that $\frac{\partial \mathcal{L}}{\partial k}\leq 0$. To see this, 
		it suffices to show that $1+\frac{\theta _{i}}{\log (1-\theta _{i})}\geq  
		\frac{\theta _{i}}{2}$. Indeed,  
		\begin{equation*}
			\frac{\partial }{\partial \theta _{i}}\left( -\left( 1-\frac{\theta _{i}}{2}
			\right) \log (1-\theta _{i})-\theta _{i}\right) =\frac{1}{2}\left( \frac{
				\theta _{i}}{1-\theta _{i}}+\log (1-\theta _{i})\right) \geq 0.
		\end{equation*}
		Hence, $-\left( 1-\frac{\theta _{i}}{2}\right) \log (1-\theta _{i})-\theta 
		_{i}\geq 0$. That is, $1+\frac{\theta _{i}}{\log (1-\theta _{i})}\geq \frac{
			\theta _{i}}{2}$. Therefore, in order the maximize the objective, the 
		information designer should choose $k=0$. 
	\end{proof}
	
	\subsubsection{Auctions under an irregular distribution}
	
	\label{irexample}Suppose that the support of the signal $x$ is $[1,2]$.
	Consider the signal distribution: 
	\begin{equation*}
		G(x)= 
		\begin{cases}
			2x-2, & \mbox{if }x\in \lbrack 1,\frac{4}{3}); \\ 
			\frac{x}{2}, & \mbox{if }x\in \lbrack \frac{4}{3},2].%
		\end{cases}%
	\end{equation*}
	The quantile function $x(\tau )$ of distribution $G$ is given by: 
	\begin{equation*}
		x(\tau )= 
		\begin{cases}
			\frac{\tau }{2}+1, & \mbox{if }\tau \in \lbrack 0,\frac{2}{3}); \\ 
			2\tau , & \mbox{if }\tau \in \lbrack \frac{2}{3},1].%
		\end{cases}%
	\end{equation*}
	The virtual value without ironing is given by 
	\begin{equation*}
		\varphi (x)=x-\frac{1-G(x)}{g(x)}= 
		\begin{cases}
			x-\frac{1-(2x-2)}{2}=2x-\frac{3}{2}, & \mbox{if }x\in \lbrack 1,\frac{4}{3});
			\\ 
			x-\frac{1-x/2}{1/2}=2x-2, & \mbox{if }x\in \lbrack \frac{4}{3},2].%
		\end{cases}%
	\end{equation*}
	Denote $\Phi (x)=\int_{1}^{x}\varphi (t)g(t)\mathrm{d}t$. Then, 
	\begin{equation*}
		\Phi (x)= 
		\begin{cases}
			\int_{1}^{x}2(2t-3/2)\mathrm{d}t=(x-1)(2x-1), & \mbox{if }x\in \lbrack 1, 
			\frac{4}{3}); \\ 
			\int_{1}^{4/3}2(2t-3/2)\mathrm{d}t+\int_{4/3}^{x}\frac{1}{2}(2t-2)\mathrm{d}
			t=\frac{1}{2}x(x-2)+1, & \mbox{if }x\in \lbrack \frac{4}{3},2].%
		\end{cases}%
	\end{equation*}
	Therefore, $\Phi (\tau )=\Phi (x(\tau ))$ is given by 
	\begin{equation*}
		\Phi (\tau )= 
		\begin{cases}
			(x(\tau )-1)(2x(\tau )-1)=\frac{\tau }{2}(\tau +1), & \mbox{if }\tau \in
			\lbrack 0,\frac{2}{3}); \\ 
			\frac{1}{2}x(\tau )(x(\tau )-2)+1=1-2\tau +2\tau ^{2}, & \mbox{if }\tau \in
			\lbrack \frac{2}{3},1].%
		\end{cases}%
	\end{equation*}
	Denote $\Psi (\tau )$ be the largest convex function such that $\Psi (\tau
	)\leq \Phi (\tau )$. Then 
	\begin{equation*}
		\Psi (\tau )= 
		\begin{cases}
			\frac{\tau }{2}(\tau +1), & \mbox{if }\tau \in \lbrack 0,\frac{1}{2}); \\ 
			\tau -\frac{1}{8}, & \mbox{if }\tau \in \lbrack \frac{1}{2},\frac{3}{4}); \\ 
			1-2\tau +2\tau ^{2}, & \mbox{if }\tau \in \lbrack \frac{3}{4},1].%
		\end{cases}%
	\end{equation*}
	Therefore, the ironed virtual value $\hat{\varphi}(\tau )=\Psi ^{\prime
	}(\tau )$ is given by 
	\begin{equation*}
		\hat{\varphi}(\tau )= 
		\begin{cases}
			\tau +\frac{1}{2}, & \mbox{if }\tau \in \lbrack 0,\frac{1}{2}); \\ 
			1, & \mbox{if }\tau \in \lbrack \frac{1}{2},\frac{3}{4}); \\ 
			4\tau -2, & \mbox{if }\tau \in \lbrack \frac{3}{4},1].%
		\end{cases}%
	\end{equation*}
	Replace $\tau $ by $G(x)$ and the ironed virtual value $\hat{\varphi}(x)$ in
	terms of $x$ is given by 
	\begin{equation*}
		\hat{\varphi}(x)= 
		\begin{cases}
			2x-\frac{3}{2}, & \mbox{if }x\in \lbrack 1,\frac{5}{4}); \\ 
			1, & \mbox{if }x\in \lbrack \frac{5}{4},\frac{3}{2}); \\ 
			2x-2, & \mbox{if }x\in \lbrack \frac{3}{2},1].%
		\end{cases}%
	\end{equation*}
	Note that $\hat{\varphi}(x)\geq \hat{\varphi}(1)=1/2$ is always positive.
	Hence, the optimal reserve price is zero. Let us then compute the expected
	highest ironed virtual value $\hat{\varphi}(x)$ and the revenue under a
	second-price auction.
	
	\begin{enumerate}
		\item First, by symmetry, the largest value $x^{(1)}$ induces the highest
		ironed virtual value and the highest value $x^{(1)}$ follows the
		distribution $G^{2}$. Then, we have 
		\begin{align*}
			& \mathbb{E}[\hat{\varphi}(x^{(1)})\vert G]=\int_{1}^{2}\hat{\varphi}(x) 
			\mathrm{d}G^{2}(x)=\int_{1}^{2}\hat{\varphi}(x)2\cdot g(x)G(x)\mathrm{d}x \\
			=& \int_{1}^{5/4}(2x-3/2)2\cdot 2(2x-2)\mathrm{d}x+\int_{5/4}^{4/3}1\cdot
			2\cdot 2(2x-2)\mathrm{d}x+\int_{4/3}^{3/2}1\cdot \frac{2x}{4}\mathrm{d}
			x+\int_{3/2}^{2}(2x-2)\cdot \frac{2x}{4}\mathrm{d}x \\
			=& \frac{5}{24}+\frac{7}{36}+\frac{17}{144}+\frac{2}{3}=\frac{19}{16}.
		\end{align*}
		
		\item Second, the lowest value $x^{(2)}$ follows the distribution $2G-G^2$,
		then we have 
		\begin{align*}
			&\mathbb{E}[x^{(2)}\vert G]=\int_{1}^{2}x\mathrm{d}(2G-G^2)=\int_{1}^{2}x
			\cdot2\cdot g(x)(1-G(x))\mathrm{d}x \\
			=&\int_{1}^{4/3}x\cdot2\cdot 2(3-2x)\mathrm{d}x+\int_{4/3}^{2}x\cdot 2\cdot 
			\frac{1}{2}(1-\frac{x}{2})\mathrm{d}x \\
			=&\frac{82}{81}+\frac{14}{81}=\frac{96}{81}=\frac{32}{27}=\frac{19}{16}- 
			\frac{1}{432}<\frac{19}{16}.
		\end{align*}
	\end{enumerate}
	
	Therefore, a second-price auction with an optimal reserve price $0$ obtains
	strictly less revenue than the expected highest ironed virtual value. Hence,
	the second-price auction with an optimal reserve price $0$ is not an optimal
	auction.
	
	Note that although $G(x)$ can induce the ironed virtual value $\hat{\varphi}
	(x)$, the regular distribution $\hat{G}$ which also induces the same virtual
	value $\hat{\varphi}(x)$ will first-order stochastically dominate $G$.
	Hence, the expectation of the lowest/second-highest value $x^{(2)}$ of $\hat{%
		G}$ is strictly larger than that of $G$. \ Formally, $\hat{G}$ is given by 
	\begin{equation*}
		\hat{G}(x)= 
		\begin{cases}
			2x-2, & \mbox{if }x\in \lbrack 1,\frac{5}{4}]; \\ 
			1-\frac{1}{8(x-1)}, & \mbox{if }x\in (\frac{5}{4},\frac{3}{2}); \\ 
			\frac{x}{2}, & \mbox{if }x\in \lbrack \frac{3}{2},2].%
		\end{cases}%
	\end{equation*}
	For $x\in\left[1,\frac{5}{4}\right]\cup \left[\frac{3}{2},2\right]$, we have 
	$\hat{		G	}(x)=G(x)$. For $x\in \left(\frac{5}{4},\frac{3}{2}\right)$, $\hat{%
		G}(x)<G(x)$. Hence $\hat{G}$ first-order stochastically dominates $G$.
	Moreover, we have ${\ \mathbb{E}}\left[x^{(2)}|\hat{G}\right]=\frac{19}{16}= 
	\mathbb{E}\left[\hat{\varphi} (x^{(1)})\vert G\right] $. Of course, this is
	an example of the standard result that the optimal auction is a second-price
	auction with an optimal reserve price, provided that the signal distribution
	is regular.
	
	\subsubsection{Inequivalence under a continuous prior with negative virtual
		values}
	
	\label{inequivir}
	
	\begin{claim}
		The buyer-optimal information structure and the seller-worst information
		structure are inequivalent under the following regular continuous prior $%
		\hat{H}$ with $\mathbb{E}_{\hat{H}}[x]\simeq0.26785$, 
		\begin{align*}
			\hat{H}(x)= 
			\begin{cases}
				5x, & \text{if } x\in[0,1/20]; \\ 
				1-\frac{9/80}{x-(-1/10)}, & \text{if } x\in[1/20,0.999); \\ 
				\frac{112500}{1099}x-\frac{111401}{1099}, & \text{if } x\in[0.999,1].%
			\end{cases}%
		\end{align*}
	\end{claim}
	
	It follows that $\hat{H}$ induces negative virtual values on $[0,0.999)$. In
	particular, the virtual value is $-1/10$ on $[1/20,0.999)$.
	
	\begin{proof}
		First, we provide an upper bound of the buyers' surplus under some 
		seller-worst information structure $G_{s}^*$ .
		
		\begin{itemize}
			\item Suppose that we only have the mean constraint instead of the 
			mean-preserving spread constraint. By Theorem \ref{result_seller}, the 
			following signal distribution is the unique seller-worst information 
			structure,  
			\begin{equation*}
				G_{s}= 
				\begin{cases}
					1-\frac{x_{s}}{x}, & \text{if }x\in \lbrack x_{s},1); \\ 
					1, & \text{if }x=1,%
				\end{cases}%
			\end{equation*}
			where $x_{s}\simeq 0.07445$ solves $x_{s}(1-\log (x_{s}))=\mathbb{E}_{\hat{H}%
			}[x]$. Since the  information designer faces a more strict mean-preserving
			spread constraint  under the prior $H$, the seller-worst information
			structure $G_{s}^{\ast }$  will generate a weakly higher seller's revenue
			than $G_{s}$ does.
			
			\item Even though $G_{s}$ is infeasible under the mean-preserving spread 
			constraint, we claim that the seller-worst information structure $
			G_{s}^{\ast }$ must generate a lower buyers' surplus than $G_{s}$ does.
			
			Since the seller-worst signal distribution must be regular and admit only 
			nonnegative virtual values even under continuous priors, $G_{s}$ must be a 
			mean-preserving spread of $G_{s}^{\ast }$. 
			
			\begin{itemize}
				\item To explain, since the virtual value  is $\varphi (x)=x-\frac{1-G(x)}{%
					G^{\prime }(x)}$, at an intersection of two  regular distributions, both of
				the distributions have the same $x$ and $G(x)$ , and thereby the
				distribution with a higher virtual value $\varphi (x)$  must have a higher
				slope $G^{\prime }(x)$ at $\left( x,G(x)\right) $.  Moreover, $G_{s}$ has
				zero virtual values on $[x_{s},1)$ and $G_{s}$ has  nonnegative virtual
				values on $[0,1]$. As illustrated in Figure \ref{single_crossing}, we claim
				that $G_{s}^{\ast }$ will only cross $G_{s}$  from the below. Suppose $%
				G_{s}^{\ast }$ crosses $G_{s}$ from the above at $ x_{0}$ as the blue curve
				in Figure \ref{single_crossing}. Since $G_{s}$ has  a higher slope than $%
				G_{s}^{\ast }$ at $\left( x_{0},G_{s}(x_{0})\right) $, $ \varphi
				(x_{0}|G_{s}^{\ast })<\varphi (x_{0}|G_{s})=0$ which contradicts  with the
				fact that $G_{s^{\ast }}$ must induce nonnegative virtual values. 
				Therefore, $G_{s}^{\ast }$ must cross $G_{s}$ from the below once and only 
				once. 
			\end{itemize}
			
			Therefore, by Lemma \ref{convex}, $G_{s}$ generates more total surplus  than 
			$G_{s}^{\ast }$. Since the good is always allocated under both $G_{s}^{\ast }
			$ and $G_{s}$,  the buyers' surplus (as the total surplus minus the seller's
			revenue) under $ G_{s}^{\ast }$ will be lower than that under $G_{s}$.
			
			\item The buyers' surplus under $G_{s}$ is $0.24898$, hence the buyers' 
			surplus under the seller-worst signal distribution $G_{s}^{\ast }$ can not 
			exceed $0.24898$. 
		\end{itemize}
		
		Second, to show the inequivalence, it suffices to find another feasible 
		signal distribution $\hat{G}$ which generates a higher buyers' surplus than 
		0.24898. And the construction is as follows,  
		\begin{align*}
			\hat{G}(x)= 
			\begin{cases}
				5x, & \text{if } x\in[0,1/50]; \\ 
				0.1, & \text{if } x\in(1/50,x_0); \\ 
				1-\frac{0.9x_0}{x}, & \text{if } x\in[x_0,0.999); \\ 
				1, & \text{if } x\in[0.999,1],%
			\end{cases}%
		\end{align*}
		where $x_0\simeq 0.0858$.
		
		The yield buyers' surplus under $\hat{G}$ is $0.25205> 0.24898$.
		
		Finally, we use Figure \ref{fig} to illustrate that $\hat{H}$ is a 
		mean-preserving spread of $\hat{G}$. 
	\end{proof}
	
	\begin{figure}[tbp]
		\begin{minipage}[l]{0.45\textwidth}
			\centering
			\begin{tikzpicture}
				\begin{axis}[
					tick label style={font=\scriptsize},
					xlabel={$x$},
					ylabel={ Distributions},
					ytick={0,1},
					yticklabels={,1 },
					xtick={0, 0.2,0.255, 1},
					xticklabels={0, , $x_0$, 1},
					no markers,
					line width=0.3pt,
					cycle list={{red,solid}},
					samples=200,
					smooth,
					domain=0:1.3,
					xmin=0, xmax=1,
					ymin=0, ymax=1,
					width=6cm, height=6cm,
					legend cell align=left,
					legend pos=  north west,
					legend style={draw=none,fill=none,name=legend},
					]
					\addplot[black,thick,domain=0.2:1]{1-0.2/x};
					\addplot[red,thick,domain=0.3:0.75]{1/(0.75-0.3)*(x-0.3)};
					\legend{$G_s$,$G_{s^*}$};		
					\addplot[blue,thick,domain=0.05:1]{1/(0.95)*(x-0.05)};
					\addplot[dashed] coordinates {
						(0.255,0)
						(0.255,0.2158)
					}; 
				\end{axis}
			\end{tikzpicture}
			\caption{$G_{s^*}$ single-crosses $G_s$}
			\label{single_crossing}
		\end{minipage}
		\begin{minipage}[r]{0.45\textwidth}
			\centering
			\begin{tikzpicture}
				\begin{axis}[
					tick label style={font=\scriptsize},
					xlabel={$x$},
					ylabel={ Distributions},
					ytick={0.1,1},
					yticklabels={0.1,1 },
					xtick={0, 1},
					xticklabels={0, 1},
					no markers,
					line width=0.3pt,
					cycle list={{red,solid}},
					samples=200,
					smooth,
					domain=0:1.3,
					xmin=0, xmax=1,
					ymin=0, ymax=1,
					width=6cm, height=6cm,
					legend cell align=left,
					legend pos=  north west,
					legend style={draw=none,fill=none,name=legend},
					]
					\addplot[black,thick,domain=0:0.05]{5*x};
					\addplot[red,thick,domain=0.0858:0.999]{1-0.0858*(1-0.1)/x};
					
					\legend{$\hat{H}$,$\hat{G}$};		
					
					\addplot[dashed, red] coordinates {
						(0.02,0.1)
						(0.0858,0.1)
					};
					\addplot[red,thick,domain=0:0.02]{5*x};
					\addplot[black,thick,domain=0.05:0.999]{1-9/(80*x+8)};
				\end{axis}
			\end{tikzpicture}
			\caption{$\hat{G}$ is feasible}
			\label{fig}
		\end{minipage}
	\end{figure}
	
\end{document}